\documentclass[format=acmsmall, screen=true,nonacm,pdftex,svgnames]{acmart}

\acmJournal{PACMPL}
\acmYear{2027}
\acmVolume{11}
\acmNumber{POPL}
\acmPrice{}
\startPage{1}

     \usepackage{booktabs}   \usepackage{subcaption} \usepackage{bussproofs}
\usepackage[cal=boondoxo]{mathalfa}
\DeclareMathAlphabet{\mathpzc}{OT1}{pzc}{m}{it}
\usepackage{amsmath}
\usepackage{color}
\usepackage{xspace}

\newcommand{\ignore}[1]{{}}

\newcommand{\syntaxDef}[3]{\rulebox{\syntaxKeyword$#1\mathrel{::=}{#2}$ \ifthenelse{\equal{#3}{}}{}{[#3]}}}

\makeatletter
\newcommand{\shorteq}{\settowidth{\@tempdima}{-}\resizebox{\@tempdima}{\height}{=}}
\makeatother

  \newcommand{\dom}[1]{\ensuremath{\mathtt{dom}(#1)}}

 \usepackage{dafny}
\usepackage{xcolor}
\usepackage{colortbl}

\usepackage[utf8]{inputenc}

\usepackage{amsmath}
\usepackage{stmaryrd}
\usepackage{hyperref}

\usepackage{subcaption}
\usepackage{wrapfig}

\usepackage{xspace}
\usepackage{float}
\usepackage{balance}

\usepackage{textcomp} 

\usepackage{changepage}
\usepackage{array,etoolbox}

\preto\tabular{\setcounter{magicrownumbers}{0}}
\newcounter{magicrownumbers}

\usepackage{cleveref} 

\usepackage{wrapfig}
\usepackage{caption}

\usepackage{longtable}
\usepackage{tabu}

\usepackage{wasysym}
\usepackage{mathpartir}

\usepackage{mathtools} \usepackage{xfrac}

\usepackage[qm,braket]{qcircuit}

\newcommand{\qafnyc}{$\text{Qafny}^{\texttt{c}}$\xspace}

\newcommand{\qinv}{\textcolor{teal}{\cn{invariant}}\xspace}

\newcommand{\qass}[2]{{#1}\;{\leftarrow}\;{#2}}

\newcommand{\sfori}[4]{\texttt{{for} }{#1}\texttt{ \(\in\) }[{#2},{#3})\;#4}
\newcommand{\qif}[2]{\texttt{{if}\,(}{#1}\texttt{)\,\{}{#2}\texttt{\}}}
\newcommand{\sifq}[2]{\texttt{if}~\cn{(}{#1}\cn{)}~{#2}}
\newcommand{\qbool}[4]{\cn{(}#1\,\cn{#2}\,#3\cn{)\,@\,}#4}

\newcommand{\Hassert}[1]{\big{\{}\ #1\ \big{\}}}
\colorlet{kwd}{black!80!green}
\definecolor{spec1}{RGB}{78, 131, 162}
\definecolor{spec0}{RGB}{66, 102, 136}
\definecolor{lespec}{RGB}{30, 80, 180}
\colorlet{spec}{lespec}
\colorlet{auto}{lespec!35!lightgray}
\colorlet{stack}{magenta}

\newcommand{\qafny}{{Qafny}\xspace}

\newcommand{\oqasm}{\textsc{oqasm}\xspace}

\newcommand{\qbricks}{{QBricks}\xspace}

\newcommand{\myparagraph}[1]{\noindent\paragraph{\textbf{#1}}}

\newcommand{\cmode}{\texttt{C}}

\newcommand{\qmode}[1]{\texttt{Q}(#1)}

\newcommand{\slen}[1]{|#1|}

\newcommand{\CC}{\ensuremath{\mathbb{C}}\xspace}

\usepackage{pgf}

\usepackage{tikz} \usetikzlibrary{arrows,shapes.misc,shapes.geometric, shapes.arrows,shapes.callouts,
  shapes.gates.logic.US,
  chains,matrix,positioning,scopes,decorations.pathmorphing,decorations.text,
  decorations.pathreplacing, shadows,automata,
  fit, calc, arrows.meta
}

\tikzset{ machine/.style={
rectangle,
minimum width=25mm,
    minimum height=18mm,
    text width=24mm,
align=center,
very thick,
    draw=black,
color=black,
    fill=white,
}
}

\DeclarePairedDelimiter\abs{\lvert}{\rvert}
\DeclarePairedDelimiter\norm{\lVert}{\rVert}

\makeatletter
\let\oldabs\abs
\def\abs{\@ifstar{\oldabs}{\oldabs*}}
\let\oldnorm\norm
\def\norm{\@ifstar{\oldnorm}{\oldnorm*}}
\makeatother

\newcommand{\ie}{\emph{i.e.,}\xspace}

\usepackage{booktabs}

\usepackage{alltt}
\usepackage{lstcoq}
\definecolor{ltblue}{rgb}{0,0.4,0.4}
\definecolor{dkblue}{rgb}{0,0.1,0.6}
\definecolor{dkgreen}{rgb}{0,0.35,0}
\definecolor{dkviolet}{rgb}{0.3,0,0.5}
\definecolor{dkred}{rgb}{0.5,0,0}
\lstdefinestyle{customc}{
  belowcaptionskip=1\baselineskip,
  breaklines=true,
  breakatwhitespace, numbers=left,
  xleftmargin=\parindent,
  language=C,
  columns=flexible,      
  showstringspaces=false,
  basicstyle=\small\sffamily,
  otherkeywords={uint,uintptr_t,let,assert},
  literate={{<-}{{$\leftarrow\,$}}2
            {->}{{$\rightarrow\,$}}2
            {<=}{{$\leq\,$}}2},
  numberstyle=\tiny\rmfamily,
  keywordstyle=\bfseries\color{green!40!black},
  commentstyle=\itshape\color{purple!40!black},
  identifierstyle=\color{blue!80!black},
}
\usepackage[export]{adjustbox}

\newcommand{\code}[1]{{\small\texttt{#1}}}

\newcommand{\rulelab}[1]{{\small \textsc{#1}}}

\newcommand{\cfor}[4]{\texttt{{for}}\,{#1}\,\texttt{\(\in\)}\,[{#2},{#3})\,{#4}}
\newcommand{\eexp}[1]{\texttt{exp}(2\pi i\,#1)}
\newcommand{\sexp}[3]{\texttt{let}~{#1}\,\texttt{=}\,{#2}~\texttt{in}~{#3}}

\newcommand{\sskip}{\texttt{\{\}}}

\newcommand{\tch}[1]{\texttt{EN}({#1})}
\newcommand{\tcht}{\texttt{EN}}

\newcommand{\ssum}[3]{\Msum_{#1}^{#2}{#3}}

\DeclareMathOperator*{\Motimes}{\text{\raisebox{0.25ex}{\scalebox{0.8}{$\bigotimes$}}}}
\DeclareMathOperator*{\Msum}{\text{\raisebox{0.25ex}{\scalebox{0.8}{$\sum$}}}}
\newcommand{\frz}{\mathpzc{F}}
\newcommand{\mea}{\mathpzc{M}}
\newcommand{\ufrz}{\mathpzc{U}}

\newcommand{\qpule}[4]{#1\vdash\textcolor{spec}{\big{\{}#2\big{\}}}\, #3 \,\textcolor{spec}{\big{\{}#4\big{\}}}}

\newcommand{\cif}[3]{\texttt{if} \; #1 \; \THEN\; #2 \;\ELSE\; #3 \,\cn{fi}}
\newcommand{\triple}[3]{\textcolor{spec}{\big{\{}#1\big{\}}}\, #2 \,\textcolor{spec}{\big{\{}#3\big{\}}}}
\newcommand{\qhltriple}[3]{\textcolor{purple}{\big{\{}#1\big{\}}}\, #2 \,\textcolor{purple}{\big{\{}#3\big{\}}}}
\newcommand{\slambda}[2]{\lambda \,{#1}\textbf{.}\, {#2}}

\newcommand{\ssassign}[3]{{#1} \xleftarrow{#2} {#3}}
\newcommand{\cassign}[2]{{#1} := {#2}}
\newcommand{\qassign}[2]{{#1} \leftarrow {#2}}

\newcommand{\smea}[1]{\texttt{measure}\cn{(}#1\cn{)}}

\newcommand{\sminus}{\texttt{-}}
\newcommand{\splus}{\texttt{+}}

\newcommand{\srange}[2]{[#1,#2)}

\newcommand{\sseq}[2]{{#1}\,\texttt{;}\,{#2}}

\newcommand{\insttwo}[4][ ]{\texttt{#2}^{#1}~{#3}~{#4}}

\newcommand{\iqft}[3][ ]{\insttwo[#1]{QFT}{#2}{#3}}

\newcommand{\ihadh}{\texttt{H}}

\newcommand{\iassign}[2]{{#1}\leftarrow{#2}}

\newcommand{\cn}[1]{\texttt{#1}}
\newcommand{\bil}{\big{ \{ }}
\newcommand{\bir}{\big{ \} }}

\def\keyfont#1{\texttt{#1}}

\newcommand{\THEN}{\keyfont{then}\xspace}
\newcommand{\ELSE}{\keyfont{else}\xspace}

\newcommand{\WHILE}{\keyfont{while}\xspace}

\newcommand{\denote}[1]{\llbracket #1 \rrbracket\xspace}

\newcommand{\qifb}[3]{\texttt{{qif}\,}{[#1]}\texttt{\,}{(\ket{0}\to #2)}\,\Box\,(\ket{1}\to #3)\,\cn{fiq}}

\usepackage{newunicodechar}
\let\Alpha=A
\let\Beta=B
\let\Epsilon=E
\let\Zeta=Z
\let\Eta=H
\let\Iota=I
\let\Kappa=K
\let\Mu=M
\let\Nu=N
\let\Omicron=O
\let\omicron=o
\let\Rho=P
\let\Tau=T
\let\Chi=X

\newunicodechar{Α}{\ensuremath{\Alpha}}
\newunicodechar{α}{\ensuremath{\alpha}}
\newunicodechar{Β}{\ensuremath{\Beta}}
\newunicodechar{β}{\ensuremath{\beta}}
\newunicodechar{Γ}{\ensuremath{\Gamma}}
\newunicodechar{γ}{\ensuremath{\gamma}}
\newunicodechar{Δ}{\ensuremath{\Delta}}
\newunicodechar{δ}{\ensuremath{\delta}}
\newunicodechar{Ε}{\ensuremath{\Epsilon}}
\newunicodechar{ε}{\ensuremath{\epsilon}}
\newunicodechar{ϵ}{\ensuremath{\varepsilon}}
\newunicodechar{Ζ}{\ensuremath{\Zeta}}
\newunicodechar{ζ}{\ensuremath{\zeta}}
\newunicodechar{Η}{\ensuremath{\Eta}}
\newunicodechar{η}{\ensuremath{\eta}}
\newunicodechar{Θ}{\ensuremath{\Theta}}
\newunicodechar{θ}{\ensuremath{\theta}}
\newunicodechar{ϑ}{\ensuremath{\vartheta}}
\newunicodechar{Ι}{\ensuremath{\Iota}}
\newunicodechar{ι}{\ensuremath{\iota}}
\newunicodechar{Κ}{\ensuremath{\Kappa}}
\newunicodechar{κ}{\ensuremath{\kappa}}
\newunicodechar{Λ}{\ensuremath{\Lambda}}
\newunicodechar{λ}{\ensuremath{\lambda}}
\newunicodechar{Μ}{\ensuremath{\Mu}}
\newunicodechar{μ}{\ensuremath{\mu}}
\newunicodechar{Ν}{\ensuremath{\Nu}}
\newunicodechar{ν}{\ensuremath{\nu}}
\newunicodechar{Ξ}{\ensuremath{\Xi}}
\newunicodechar{ξ}{\ensuremath{\xi}}
\newunicodechar{Ο}{\ensuremath{\Omicron}}
\newunicodechar{ο}{\ensuremath{\omicron}}
\newunicodechar{Π}{\ensuremath{\Pi}}
\newunicodechar{π}{\ensuremath{\pi}}
\newunicodechar{ϖ}{\ensuremath{\varpi}}
\newunicodechar{Ρ}{\ensuremath{\Rho}}
\newunicodechar{ρ}{\ensuremath{\rho}}
\newunicodechar{ϱ}{\ensuremath{\varrho}}
\newunicodechar{Σ}{\ensuremath{\Sigma}}
\newunicodechar{σ}{\ensuremath{\sigma}}
\newunicodechar{ς}{\ensuremath{\varsigma}}
\newunicodechar{Τ}{\ensuremath{\Tau}}
\newunicodechar{τ}{\ensuremath{\tau}}
\newunicodechar{Υ}{\ensuremath{\Upsilon}}
\newunicodechar{υ}{\ensuremath{\upsilon}}
\newunicodechar{Φ}{\ensuremath{\Phi}}
\newunicodechar{φ}{\ensuremath{\phi}}
\newunicodechar{ϕ}{\ensuremath{\varphi}}
\newunicodechar{Χ}{\ensuremath{\Chi}}
\newunicodechar{χ}{\ensuremath{\chi}}
\newunicodechar{Ψ}{\ensuremath{\Psi}}
\newunicodechar{ψ}{\ensuremath{\psi}}
\newunicodechar{Ω}{\ensuremath{\Omega}}
\newunicodechar{ω}{\ensuremath{\omega}}

\newunicodechar{ℕ}{\ensuremath{\mathbb{N}}}
\newunicodechar{∅}{\ensuremath{\emptyset}}

\newunicodechar{∙}{\ensuremath{\bullet}}
\newunicodechar{≈}{\ensuremath{\approx}}
\newunicodechar{≅}{\ensuremath{\cong}}
\newunicodechar{≡}{\ensuremath{\equiv}}
\newunicodechar{≤}{\ensuremath{\le}}
\newunicodechar{≥}{\ensuremath{\ge}}
\newunicodechar{≠}{\ensuremath{\neq}}
\newunicodechar{∀}{\ensuremath{\forall}}
\newunicodechar{∃}{\ensuremath{\exists}}
\newunicodechar{±}{\ensuremath{\pm}}
\newunicodechar{∓}{\ensuremath{\pm}}
\newunicodechar{·}{\ensuremath{\cdot}}
\newunicodechar{⋯}{\ensuremath{\cdots}}
\newunicodechar{…}{\ensuremath{\ldots}}
\newunicodechar{∷}{~\mathrel{:\!\!\!:}~}
\newunicodechar{×}{\ensuremath{\times}}
\newunicodechar{∞}{\ensuremath{\infty}}
\newunicodechar{→}{\ensuremath{\to}}
\newunicodechar{←}{\ensuremath{\leftarrow}}
\newunicodechar{⇒}{\ensuremath{\Rightarrow}}
\newunicodechar{↦}{\ensuremath{\mapsto}}
\newunicodechar{↝}{\ensuremath{\leadsto}}
\newunicodechar{∨}{\ensuremath{\vee}}
\newunicodechar{∧}{\ensuremath{\wedge}}
\newunicodechar{⊢}{\ensuremath{\vdash}}
\newunicodechar{⊣}{\ensuremath{\dashv}}
\newunicodechar{∣}{\ensuremath{\mid}}
\newunicodechar{∈}{\ensuremath{\in}}
\newunicodechar{⊆}{\ensuremath{\subseteq}}
\newunicodechar{⊂}{\ensuremath{\subset}}
\newunicodechar{∪}{\ensuremath{\cup}}
\newunicodechar{⋓}{\ensuremath{\Cup}}
\newunicodechar{∉}{\ensuremath{\not\in}}
\newunicodechar{√}{\ensuremath{\sqrt}}

\newunicodechar{⊸}{\ensuremath{\multimap}}
\newunicodechar{⊗}{\ensuremath{\otimes}}
\newunicodechar{⨂}{\ensuremath{\bigotimes}}
\newunicodechar{⊕}{\ensuremath{\oplus}}
\newunicodechar{〈}{\ensuremath{\langle}}
\newunicodechar{⟨}{\ensuremath{\langle}}
\newunicodechar{⟩}{\ensuremath{\rangle}}
\newunicodechar{〉}{\ensuremath{\rangle}}
\newunicodechar{¡}{\ensuremath{\upsidedownbang}}
\newunicodechar{∘}{\ensuremath{\circ}}
\newunicodechar{†}{\ensuremath{\dagger}}
\newunicodechar{⊤}{\ensuremath{\top}}
\newunicodechar{⊥}{\ensuremath{\bot}}

\newunicodechar{〚}{\ensuremath{\llbracket}}
\newunicodechar{〛}{\ensuremath{\rrbracket}}

\usepackage{etoolbox} \newtoggle{comments}
\toggletrue{comments}
\usepackage[normalem]{ulem}

\iftoggle{comments}{
  \newcommand{\fixme}[1]{\textbf{\textcolor{red}{[ Fixme: #1]}}}
  \newcommand{\todo}[1]{\textbf{\textcolor{green}{[ TODO: #1 ]}}}
}{
  \newcommand{\fixme}[1]{}
  \newcommand{\todo}[1]{}
}
 
\usepackage{pifont}\usepackage{multicol,tabularx,capt-of}
\usepackage{multirow}

\begin{document}

\title{Connecting Quantum Proof Systems and Hoare Logic}                      

\def\titlerunning{Connecting Quantum Proof Systems and Hoare Logic}
\def\authorrunning{L. Li, Z. Tagba, A. Potanin}

\author{Liyi Li}
\affiliation{
  \institution{Iowa State University}
  \country{USA}
}
\email{liyili2@iastate.edu}

\author{Zoukarneini Difaizi Tagba}
\affiliation{
  \institution{Iowa State University}
  \country{USA}
}
\email{difaizi@iastate.edu}

\author{Alex Potanin}
\affiliation{
  \institution{Australian National University}
  \country{Australia}
}
\email{alex.potanin@anu.edu.au}

\begin{abstract}
 Verifying quantum programs today requires specialized proof systems built on interactive theorem provers such as Rocq, Lean, Isabelle, or Why3, and even modest quantum algorithms cost substantial human effort to verify. Because quantum programs cannot be effectively checked by runtime testing, formal verification is the only viable path to correctness, making the cost of these bespoke systems the chief barrier to scalable quantum software assurance. We show that verifying a quantum program should be no harder than verifying a classical array program, a paradigm commonly used in parallel programming: the Qafny quantum proof system translates into Hoare logic with array theory completely (every Qafny proof becomes a valid Hoare triple) and soundly (a valid translated triple certifies only true quantum specifications), and we mechanize the completeness direction in Rocq. The embedding determinizes quantum proof obligations and exposes the full machinery of classical automated verifiers (e.g. Dafny and SMT-based separation logic) to quantum code, removing the need for dedicated quantum proof system infrastructure. In particular, the embedding supplies the correctness argument that the existing compilation of Qafny into Dafny lacks.
 \end{abstract}

\maketitle



\section{Introduction}
\label{sec:intro}

Quantum computers aim to solve highly complex problems by leveraging quantum superposition, which enables substantially faster algorithms, such as Shor's algorithm, than classical algorithms.
Because of the exponential state space of quantum superpositions and the inherently probabilistic, nondeterministic nature of quantum systems, traditional validation techniques based on run-time testing are virtually impossible to develop for large quantum programs.
Many frameworks have been proposed to verify quantum programs \cite{qbricks,VOQC} using interactive theorem provers, such as Isabelle, Rocq, and Why3, while others have proposed some proof automation by creating proof systems \cite{qhoreusage,qhoare,qsepa,qseplocal,10.1145/3290346,li2024qafnyquantumprogramverifier,10.1145/3763157}, but quantum program verification is still time-consuming.

While quantum computation might be more powerful than classical computation, an interesting question is whether verifying quantum programs is more difficult than verifying classical programs. The key is to establish an effective compilation from quantum program verification to a verification problem in a classical framework, thereby relating quantum and classical operations and providing a classical analogy for quantum programs.

In this paper, we propose a set of classical aggregate operations as the target for compiling quantum operations.
Such operations are well-known in higher-order functional programming and are associated with a large and effective verification library.
Based on these operations, we transform the \qafny quantum proof system to classical Hoare logic, and illustrate the same approach for Quantum Hoare Logic (QHL) \cite{10.1145/3456877} to show its generality.
While these classical operations might be computationally expensive, their verification might not be expensive, i.e., it suffices to verify correctness by examining a symbolic representative value of the quantum state after applying a quantum operation. When one applies a quantum operation to a quantum state that can be thought of as a set of classical state values, it cannot manipulate a singleton classical state value without affecting the others. In fact, the manipulation of the operation has a similar effect across different classical state values, which can be summarized by a symbolic state representing the quantum state.

\begin{figure}[h]
  \includegraphics[width=0.9\textwidth]{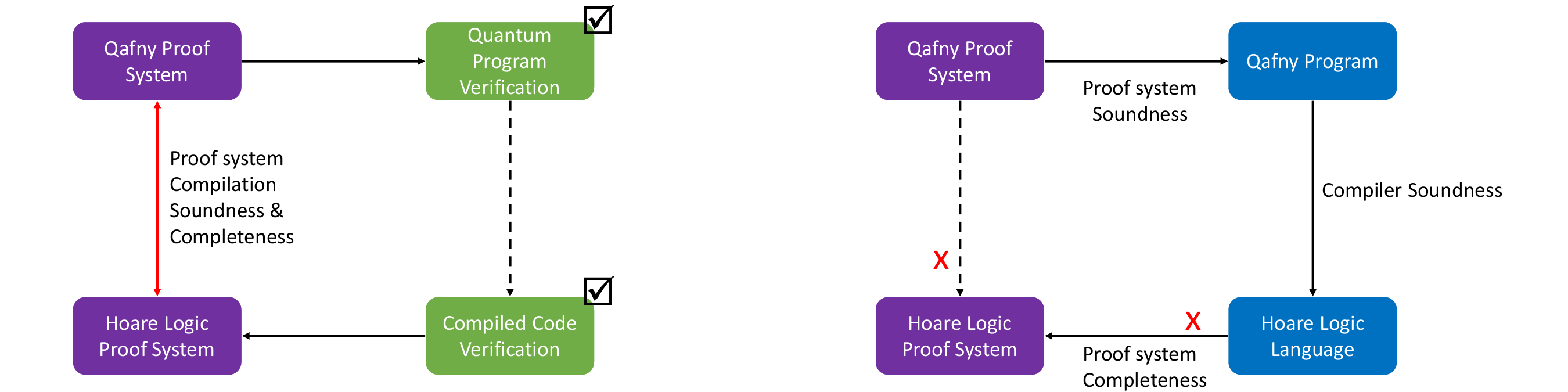}
  \caption{Two different kinds of soundness and completeness theorems (ours on the left).}
\label{fig:arch}
\end{figure}

To connect quantum and classical proof systems, we develop a new methodology for relating them. In verifying the transformation from quantum proof systems to classical Hoare logic, we develop soundness and completeness theorems based on such theorems in the compilation sense (\Cref{fig:arch}), rather than relying on the proof system soundness and relative completeness, which relate proof systems and their underlying language semantics.
The right demonstrates the flow of what would happen by using proof system soundness and relative completeness for our transformation, which fails in place of utilizing the Hoare logic relative completeness to transform a semantic system back to logical assertions, because relative completeness cannot suggest the existence of an effective program verification via the transformation.

Essentially, the main issue arises at predicate transformation.
Consider a simple quantum operation ($\qassign{\overline{q}}{\slambda{x}{\ket{x+3}}}$) to add $3$ on each basis-vector of a quantum state $\varphi$ for $\overline{q}$, where $\varphi$ is a superposition state $\sum_j\alpha_j\ket{j}$. When performing verification, users also need to support a logical post-condition $\overline{q}=\textcolor{spec}{\sum_j\alpha_j\ket{j+3}}$, where $\textcolor{spec}{\sum_j\alpha_j\ket{j+3}}$ is an expression for describing the quantum state $\varphi$ after the application, resulting in a quantum state $\varphi'$.
The post-condition transformation based on the right side of \Cref{fig:arch} is listed below.

{\small
\[
\overline{q}=\textcolor{spec}{\sum_j\alpha_j\ket{j+3}}
\to 
\{\overline{q}:\varphi'\}
\to 
\{\overline{q}:\cn{trans}(\varphi')\}
\to
\overline{q} = \cn{list}(\cn{trans}(\varphi'))
\]
}

The first transformation above transforms the post-condition to a post-state, a map structure mapping from $\overline{q}$ to a quantum state $\varphi'$. Writing down the state can potentially be exponential.
The next step is to transform the state by using some classical state structure ($\cn{trans}$), e.g., lists or arrays.
The last step turns the classical state structure back into a predicate, yielding a possibly exponential listing of different singleton states in the quantum superposition state $\varphi'$.

To connect quantum and classical proof systems properly, as in the left of \Cref{fig:arch}, we develop completeness and soundness theorems for the translation itself, in the style of compiler correctness, to relate the validity of the two systems; i.e., the translation completeness theorem states that everything derivable in the quantum proof system is also derivable in Hoare logic, while the translation soundness theorem states that if the translated program is derivable in Hoare logic, then the quantum proof system also derives it, up to semantic equivalence of the assertions.

\subsection{Contributions and Roadmap}

\qafny \cite{li2024qafnyquantumprogramverifier} proposes a quantum proof system, and \qafnyc \cite{10.1145/3763157} realizes it by compiling \qafny verification conditions into Dafny, reporting 37 verified programs; that compilation, however, comes with no correctness argument.
This paper supplies the missing foundation, developed for \qafny and illustrated for QHL: every quantum program verification via the \qafny proof system can be viewed as a verification in the Hoare logic system, which is our translation completeness theorem. Translation soundness adds that a valid translated Hoare triple certifies only specifications that \qafny itself derives, so the transformation from \qafny to Hoare logic does not produce garbage.
Below are the detailed contributions.

\begin{itemize}
\item We invent an intermediate Hoare logic IR (HLIR) as the target language. We then develop a transformation from a quantum proof system to Hoare logic based on HLIR. We demonstrate the transformation by compiling the \qafny proof system to Hoare logic with array theory (HLIR), demonstrating that quantum program verification can be carried out via a classical Hoare logic proof system.

\item We prove translation completeness and translation soundness theorems for the compilation, establishing a theoretical foundation: every quantum program verification can be done effectively in Hoare logic, indicating that quantum program verification is no harder than classical program verification. 

\item We also examine tensor-based separation conjunctions in \qafny as well as other proof systems \cite{YING2026105417,qseplocal}. We show that these quantum separation conjunctions are pure, in the sense that they can be replaced by classical conjunctions.
\end{itemize}

 \section{Background}
\label{sec:background}

\noindent\textbf{\textit{Quantum Data}.} A quantum datum consists of one or more quantum bits (\emph{qubits}), which can be expressed as a two-dimensional vector $\begin{psmallmatrix} z_1 \\ z_2 \end{psmallmatrix}$ where the \emph{amplitudes} $z_1$ and $z_2$ are complex numbers and $|z_1|^2 + |z_2|^2 = 1$.
We frequently write the qubit vector as $\textcolor{spec}{z_1\ket{0} + z_2\ket{1}}$ (the Dirac notation \cite{Dirac1939ANN}), where $\ket{0} = \begin{psmallmatrix} 1 \\ 0 \end{psmallmatrix}$ and $\ket{1} = \begin{psmallmatrix} 0 \\ 1 \end{psmallmatrix}$ are \emph{computational basis-kets}. When both $z_1$ and $z_2$ are non-zero, we can think of the qubit being ``both 0 and 1 at once,'' a.k.a. in a \emph{superposition} \cite{mike-and-ike}, e.g., $\frac{1}{\sqrt{2}}(\ket{0} + \ket{1})$ represents a superposition of $\ket{0}$ and $\ket{1}$.
Larger quantum data can be formed by composing smaller ones with the \emph{tensor product} ($\otimes$) from linear algebra, e.g., the two-qubit datum $\ket{0} \otimes \ket{1}$ (also written as $\ket{01}$) corresponds to the vector $[~0~1~0~0~]^T$.
However, many multi-qubit data cannot be \emph{separated} and expressed as the tensor product of smaller data; such inseparable data are called \emph{entangled}, e.g.\/ $\frac{1}{\sqrt{2}}(\ket{00} + \ket{11})$, known as a \emph{Bell pair}, which can be rewritten to $\Msum_{d=0}^{1}{\frac{1}{\sqrt{2}}}{\ket{dd}}$, where $dd$ is a bit string consisting of two bits, each of which must be the same datum (i.e., $d=0$ or $d=1$). Each term $\frac{1}{\sqrt{2}}\ket{dd}$ is named a \emph{basis-ket} \cite{mike-and-ike}, consisting of an amplitude ($\frac{1}{\sqrt{2}}$) and a basis vector $\ket{dd}$.
     
\noindent\textbf{\textit{Quantum Computation and Measurement.}} 
Computation on a quantum datum consists of a series of \emph{quantum operations}, each acting on a subset of qubits in the quantum datum. In the standard form, quantum computations are expressed as \emph{circuits}, as in \Cref{fig:background-circuit-examplea}, which depicts a circuit that prepares the Greenberger-Horne-Zeilinger (GHZ) state \cite{Greenberger1989} --- an $n$-qubit entangled datum of the form: $\ket{\text{GHZ}^n} = \frac{1}{\sqrt{2}}(\ket{0}^{\otimes n}+\ket{1}^{\otimes n})$, where $\ket{d}^{\otimes n}=\Motimes_{d=0}^{n\sminus 1}\ket{d}$.
In these circuits, each horizontal wire represents a qubit, and boxes on these wires indicate quantum operations, or \emph{gates}. The circuit in \Cref{fig:background-circuit-examplea} uses $n$ qubits and applies $n$ gates: a \emph{Hadamard} (\coqe{H}) gate and $n-1$ \emph{controlled-not} (\coqe{CNOT}) gates. Applying a gate to a quantum datum \emph{evolves} it.
Its traditional semantics is expressed by multiplying the datum's vector form by the gate's corresponding matrix representation: $n$-qubit gates are $2^n$-by-$2^n$ matrices.
Except for measurement gates, a gate's matrix must be \emph{unitary} and thus preserve appropriate invariants of quantum data's amplitudes. 
A \emph{measurement} operation extracts classical information from a quantum datum. It collapses the datum to a basis state with a probability related to the datum's amplitudes (\emph{measurement probability}), e.g., measuring $\frac{1}{\sqrt{2}}(\ket{0} + \ket{1})$ collapses the datum to $\ket{0}$ with probability $\frac{1}{2}$,  and likewise for $\ket{1}$, returning classical value $0$ or $1$, respectively. 
A more general form of quantum measurement is \emph{partial measurement}, which measures a subset of qubits in a qubit array;
such operations often have simultaneity effects due to entanglement, \ie{} in a Bell pair $\frac{1}{\sqrt{2}}(\ket{00} + \ket{11})$, measuring one qubit guarantees the same outcome for the other --- if the first bit is measured as $0$, the second bit is too.

\noindent\textbf{\textit{High-level Operation Abstraction.}}
Controlled quantum gates, such as controlled-not gates (\coqe{CNOT}), can be thought of as quantum versions of classical conditionals, called quantum conditionals, where we view a quantum datum as an array of basis-kets and apply an array map corresponding to the classical operation to each basis-ket.
This is evident when using Dirac notation. 
For example, in preparing a two-qubit GHZ state (\Cref{fig:background-circuit-examplea}, also a Bell pair) for qubit array $x$, the \coqe{H} gate evolves the datum to
$\frac{1}{\sqrt{2}}\ket{00}+\frac{1}{\sqrt{2}}\ket{10}$ (same as $\frac{1}{\sqrt{2}}(\ket{0} + \ket{1})\otimes \ket{0}$).
The quantum conditional maps the classical conditional $\qif{x[0]}{\qass{x[1]}{x[1] \,\splus\, 1}}$ onto the two basis-kets, where the operation $x[1]\, \splus\, 1$ acts as a modulo 2 addition to flip $x[1]$'s bit.
Here, we do not flip the $x[1]$ position in the first basis-ket ($\frac{1}{\sqrt{2}}\ket{00}$) due to $x[0]=0$, and we flip $x[1]$ in the second basis-ket because of $x[0]=1$.

\begin{wrapfigure}{r}{2.8cm}
            {\qquad
              \footnotesize
              \Qcircuit @C=0.5em @R=0.5em {
                \lstick{\ket{0}} & \gate{H} & \ctrl{1} & \qw &\qw & & \dots & \\
                \lstick{\ket{0}} & \qw & \targ & \ctrl{1} & \qw & &  \dots &  \\
                \lstick{\ket{0}} & \qw & \qw   & \targ & \qw & &  \dots &  \\
                & \vdots &   &  &  & & & \\
                & \vdots &  & \dots & & & \ctrl{1} & \qw  \\
                \lstick{\ket{0}} & \qw & \qw & \qw &\qw &\qw & \targ & \qw
              }
            }
            \caption{GHZ Circuit}
            \label{fig:background-circuit-examplea}
\end{wrapfigure}
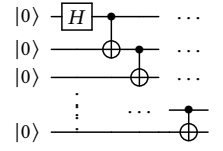 

Another abstraction is quantum oracle operations.
Quantum algorithms manipulate input information encoded in ``oracles,'' which are callable black-box circuits. 
Quantum oracles are usually quantum-reversible implementations of classical operations, especially \emph{arithmetic operations}.
The behavior of a quantum arithmetic operation resembles a high-order array map operation in \Cref{fig:onetime-hw} line (7) over a single arithmetic operation applied to every basis-ket of a quantum datum.
\oqasm \cite{oracleoopsla} is a language that permits the definitions of quantum oracles.

\noindent\textbf{\textit{Hoare Logic with Array Theory.}} Hoare logic (also known as Floyd–Hoare logic or Hoare rules) \cite{HoareLogic,Gordon12backgroundreading} is a formal system with a set of logical rules for reasoning rigorously about the correctness of computer programs. The central feature of Hoare logic is the Hoare triple. A triple describes how the execution of a piece of code changes the state of the computation. A Hoare triple is of the form $\triple{P}{C}{Q}$ where $P$ and $Q$ are assertions and $C$ is a C-like program. Here, $P$ is named the pre-condition and $Q$ the post-condition.  When the pre-condition is met, executing the command establishes the post-condition. Assertions are formulae in predicate logic.

{
\[
\triple{P[f(A)/A_1]}{A_1 := f(A)}{P}
\qquad
\triple{P[A\{E_1 \leftarrow E_2\}/A]}{A[E_1] := E_2}{P}
\]
}

The left-hand side shows the key assignment rule in Hoare logic, which establishes the relation between the post-condition $P$ and the pre-condition $P[f(A)/A_1]$ for the assignment operation $A_1 := f(A)$. The variables $A$ and $A_1$ are not restricted to any particular data structure, as they certainly can represent array variables.
To verify array assignments, we can define the operation $A[E_1]:= E_2$, i.e., the index $E_1$ of the array $A$ is assigned to $E_2$. As the proof rule on the right-hand above, starting from the post-condition $P$, the pre-condition $P[A\{E_1 \leftarrow E_2\}/A]$ is a replacement of every array variable $A$ with the updated array $A\{E_1 \leftarrow E_2\}$, i.e., we assume the $\leftarrow$ operation updates the $E_1$ position with $E_2$ in $A$.

We assume some useful predicates and conventions to capture the properties of array indexing.
First, we use $\overline{R}$ to refer to an array of elements having type $R$, $\cn{range}(\overline{m})$ refers to a list of ranges $[0,m_0),...,[0,m_{n\sminus1})$ for some $n$, and $\overline{k} \in \cn{range}(\overline{m'})$ means $k_0\in[0,m_0),...,k_{n\sminus 1}\in[0,m_{n\sminus 1})$.
Second, the predicate $Q_{\cn{len}}(A,\overline{m})$ ensures that $A$, a multi-dimensional array, has lengths in different dimensions as $\overline{m}$, and we refer $Q_{\cn{len}}(A_0,...,A_n,\overline{m})$ to that the arrays $A_0$,..,$A_n$ all have lengths $\overline{m}$.
Third, we refer $\overline{m}\sqsubseteq \overline{m'}$ to that $\overline{m}$ is a prefix of $\overline{m'}$, i.e., if $\overline{m}$ has $n$ terms and $\overline{m'}$ has $n+k$ terms, then $\overline{m'}=m_0,...,m_{n\sminus 1},m'_{n},...,m'_{n\splus k\sminus 1}$.

The following lemma links a pointwise specification of an array-producing statement to the assignment rule: once a statement is known to fill $A'$ with $f(A[\overline{k}])$ at every index, any post-condition $P$ can be pushed backwards through it by substituting the pointwise-updated array for $A'$, exactly as the assignment rule does for a scalar. It is the step that turns each HLIR primitive's loop into a single substitution in \Cref{sec:trans}.

\begin{lemma}[Pointwise Update]\label[lemma]{cor:link}\rm
Given $Q \triangleq Q_{\cn{len}}(A,\overline{m})\wedge Q_{\cn{len}}(A',\overline{m'})$, $\overline{m}\sqsubseteq \overline{m'}$, and
\[\triple{Q}{A' := \varphi(A)}{\forall \overline{k} \in \cn{range}(\overline{m}) . A'[\overline{k}]=f(A[\overline{k}])},\]
we have
\[\triple{P[A'\{\forall \overline{k}\in\cn{range}(\overline{m}). \overline{k} \mapsto f(A[\overline{k}])\}/A'] \wedge Q}{A' := \varphi(A)}{P\wedge Q}.\]
\end{lemma}

\begin{proof}
By using the Hoare logic assignment rule, we have $\triple{P[f(A)/A']}{A' := \varphi(A)}{P}$. With the assumption in the corollary, we result in the following post-condition,
\[
\textcolor{spec}{P[A'\{\forall \overline{k}\in\cn{range}(\overline{m}).\overline{k} \mapsto f(A[\overline{k}])\}/A']\wedge \forall \overline{k} \in \cn{range}(\overline{m}) . A'[\overline{k}]=f(A[\overline{k}])}\], 
which is equivalent to $P$.
\end{proof}

The $Q_{\cn{len}}$ predicate is widely used in various proofs in the paper, so we will sometimes assume it without explicit mention.
 \section{Overview}
\label{sec:overview}

Here, we provide an overview of embedding quantum program proofs in Hoare logic.
We utilize a running example based on the Deutsch–Jozsa algorithm in \Cref{fig:onetime-hw}: given a function $f$ that is either constant or balanced, represented as an oracle $f$ in the circuit, the algorithm determines whether $f$ is constant by outputting $v=0$ or another value.
In the algorithm, we first use a length $n$ qubit array $x$ and an extra qubit $y$, where $x\srange{0}{n}$ and $y[0]$ is in a uniform superposition by applying $n+1$ Hadamard gates in lines $2$ and $5$. The superposition state spans all possible values for the two-qubit arrays.

We then apply the oracle operation $f$ in line 7 to the locus $x\srange{0}{n}\sqcupplus y[0]$, i.e., the list of qubit ranges, possibly entangled, that the operation acts on (loci are defined in \Cref{sec:qafnystates}), and $n$ Hadamard gates to $x\srange{0}{n}$ in line 9.
The oracle operation queries the function $f$ and stores the result in the $y[0]$ qubit.
The $n$ \cn{H} gates associate the query results with the amplitudes, determining the probability of each basis-ket in the quantum state.
Here, the state $\ket{0}$ yields an amplitude $1$ if $f$ is a constant function, and the measurement operation in line $11$ outputs $v=0$ with probability $r=1$.

Line 1 specifies the pre-condition: $x$ starts in the $\ket{\overline{0}}$ state and $y[0]$ in the $\ket{1}$ state.
Line 12 specifies the post-condition to be verified: that the measurement result $v=0$ correctly indicates that the function $f$ is constant, with output probability $r=1$. 
The other assertions are the step-wise post-conditions.
The $\tch{n}$ type information is specific to \qafny and is introduced in the next section.

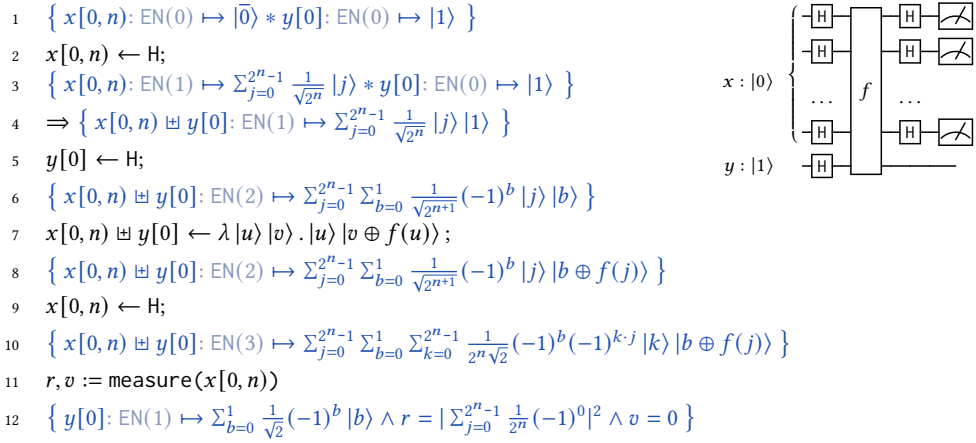
\begin{figure}[t]
\begin{minipage}[t]{.64\textwidth}
  {\small
    \[\hspace{-1em}
      \begin{array}{r l}
        \textcolor{black}{\text{\tiny 1}}
        &  \textcolor{spec}{ \Hassert{x\srange{0}{n} \textcolor{auto}{: \tch{0}}\mapsto \ket{\overline{0}} *  y[0] \textcolor{auto}{: \tch{0}} \mapsto \ket{1}}}
        \\[0.3em]
       \textcolor{black}{\text{\tiny 2}}
        &
          \qass{x\srange{0}{n}}{\cn{H}};
        \\
\textcolor{black}{\text{\tiny 3}}
        &
             \textcolor{spec}{ \Hassert{x\srange{0}{n} \textcolor{auto}{: \tch{1}} \mapsto \ssum{j=0}{2^{n}\sminus 1}{\frac{1}{\sqrt{2^n}}\ket{j}} *  y[0] \textcolor{auto}{: \tch{0}} \mapsto \ket{1}}}
            \\[0.3em]
                \textcolor{black}{\text{\tiny 4}}
        &
          \Rightarrow   \textcolor{spec}{ \Hassert{x\srange{0}{n}\sqcupplus y[0] \textcolor{auto}{: \tch{1}} \mapsto \ssum{j=0}{2^{n}\sminus 1}{\frac{1}{\sqrt{2^n}}\ket{j}}\ket{1}}}
            \\[0.3em]
           \textcolor{black}{\text{\tiny 5}}
        &
          \qass{y[0]}{\cn{H}};
            \\[0.3em]
        \textcolor{black}{\text{\tiny 6}}
        &            
        \textcolor{spec}{ \Hassert{x\srange{0}{n}\sqcupplus y[0] \textcolor{auto}{: \tch{2}} \mapsto \ssum{j=0}{2^{n}\sminus 1}{\ssum{b=0}{1}{\frac{1}{\sqrt{2^{n\splus 1}}}(-1)^{b}\ket{j}}\ket{b}}}}
            \\[0.3em]

        \textcolor{black}{\text{\tiny 7}}
        &
            \qass{x\srange{0}{n}\sqcupplus y[0]}{\lambda \ket{u}\ket{v}.\ket{u}\ket{v\oplus f(u)};}
        \\[0.3em]
        \textcolor{black}{\text{\tiny 8}}
        &
        \textcolor{spec}{ \Hassert{x\srange{0}{n}\sqcupplus y[0] \textcolor{auto}{: \tch{2}}\mapsto \ssum{j=0}{2^{n}\sminus 1}{\ssum{b=0}{1}{\frac{1}{\sqrt{2^{n\splus 1}}}(-1)^{b}\ket{j}}\ket{b\oplus f(j)}}}}
        \\
        \textcolor{black}{\text{\tiny 9}}
        &
          \qass{x\srange{0}{n}}{\cn{H}};
        \\[0.3em]
     \textcolor{black}{\text{\tiny 10}}
        &
        \textcolor{spec}{ \Hassert{x\srange{0}{n}\sqcupplus y[0] \textcolor{auto}{: \tch{3}}\mapsto \ssum{j=0}{2^{n}\sminus 1}{\ssum{b=0}{1}{\ssum{k=0}{2^n\sminus 1}{\frac{1}{2^n\sqrt{2}}(-1)^{b}(-1)^{k\cdot j}\ket{k}}\ket{b\oplus f(j)}}}}}
        \\[0.3em]
        \textcolor{black}{\text{\tiny 11}}
        &
              \cassign{r,v}{\smea{x\srange{0}{n}}}
        \\[0.3em]
    \textcolor{black}{\text{\tiny 12}}
        &

        \textcolor{spec}{\Hassert{y[0] \textcolor{auto}{:\tch{1}}\mapsto{\ssum{b=0}{1}{\frac{1}{\sqrt{2}}(-1)^{b}\ket{b}}} \wedge r=\slen{\ssum{j=0}{2^{n}\sminus 1}{{\frac{1}{2^n}(-1)^{0}}}}^2\wedge v= 0}}
      \end{array}
    \]
  }
  \end{minipage}
\begin{minipage}[t]{0.35\textwidth}
  \vspace{0.5em}
  {\scriptsize\quad
  \Qcircuit @C=0.5em @R=0.5em {
    &                    & & \gate{\cn{H}}  & \qw & \multigate{6}{\texttt{$f$}} & \qw  & \gate{\cn{H}} & \qw & \meter \\
    &                      & & \gate{\cn{H}} & \qw & \ghost{\texttt{$f$}}        & \qw  & \gate{\cn{H}} & \qw & \meter \\
    &    \push{x:\ket{0}\quad}                                & &               &     &  &      &  & &  \\
    &                                    & & \dots         &     &                             &      & \dots & &  \\
    &                                    & &               &     &                             &      & & &  \\
    &                      & & \gate{\cn{H}} & \qw & \ghost{\texttt{$f$}}       & \qw  & \gate{\cn{H}} & \qw & \meter   \\
    &     \push{y:\ket{1}\quad}                 & & \gate{\cn{H}} & \qw & \ghost{\texttt{$f$}}       & \qw  & \qw  & \qw & \qw 
    \gategroup{1}{3}{6}{3}{1em}{\{}
    }
}
  \end{minipage}
    \caption{Running example for Deutsch–Jozsa algorithm when $f$ is constant; circuit on the right; the grayed out part is the \qafny state types; $\textcolor{spec}{\ket{\overline{0}}}$ is an abbreviation of $\textcolor{spec}{\bigotimes^{n}\ket{0}}$; line 12 has several simplification and $r$ is actually equal to $1$ after simplification; $y[0]$ state in lines 6 to 12 are separable from $x\srange{0}{n}$, as $\ssum{b=0}{1}{\frac{1}{\sqrt{2}}(-1)^{b}\ket{b}}$.}
  \label{fig:onetime-hw}
\end{figure}

\myparagraph{Relating Quantum Operations with Classical Ones in a Proof Triple}\label{sec:relationclassical}
Here, we provide an overview of viewing a quantum operation as a behavioral equivalence sequence of classical operations.
The transformation does not introduce computational benefits, but these classical operations are well-known structures that can be effectively verified; thus, the transformation could provide benefits in verifying quantum programs.

We now highlight the key strategy by showing the transformation of the proof step for \Cref{fig:onetime-hw} lines 6-8.
We list the quantum triple with conditions in \qafny as an example.
Such a triple captures the key features in many other quantum proof systems \cite{li2024qafnyquantumprogramverifier,10.1145/3456877}, where we discuss one of them next.

{\small
\[
\begin{array}{l}
\quad\textcolor{spec}{
\big{\{}
x\srange{0}{n}\sqcupplus y[0] \textcolor{auto}{: \tch{2}}\mapsto \ssum{j=0}{2^{n}\sminus 1}{\ssum{b=0}{1}{\frac{1}{\sqrt{2^{n\splus 1}}}(-1)^{b}\ket{j}\ket{b}}}
\big{\}}}
\\[0.5em]
\qquad\qass{x\srange{0}{n}\sqcupplus y[0]}{\lambda \ket{u}\ket{v}.\ket{u}\ket{v\oplus f(u)};}
\\[0.5em]
\quad\textcolor{spec}{
 \Hassert{x\srange{0}{n}\sqcupplus y[0] \textcolor{auto}{: \tch{2}}\mapsto \ssum{j=0}{2^{n}\sminus 1}{\ssum{b=0}{1}{\frac{1}{\sqrt{2^{n\splus 1}}}(-1)^{b}\ket{j}\ket{b\oplus f(j)}}}}
}
\end{array}
\]
}

In handling the oracle operation $\lambda \ket{u}\ket{v}.\ket{u}\ket{v\oplus f(u)}$, we pattern-match it with the ket structure ($\textcolor{spec}{\ket{j}\ket{b}}$) of each basis-ket in the pre-condition state, and infer that the ket structure of each basis-ket in the post-condition state should be $\textcolor{spec}{\ket{j}\ket{b\oplus f(j)}}$.
In relating the quantum state predicates to a classical Hoare logic format, we think of a linear sum operator as a single dimension array concatenation; thus, the nested sum operator ($\textcolor{spec}{\ssum{j=0}{2^{n}\sminus 1}{\ssum{b=0}{1}{\frac{1}{\sqrt{2^{n\splus 1}}}(-1)^{b}\ket{j}\ket{b}}}}$) introduces a two dimensional array where each element property can be described by the expression $\textcolor{spec}{\frac{1}{\sqrt{2^{n\splus 1}}}(-1)^{b}\ket{j}\ket{b}}$. 
In our system, we can think of the number ($n$) of sum operators in a quantum state predicate as being transformed into an $ n$-dimensional array capturing the predicate relation (this is the essence of the $\tch{n}$ type). 
When $n=1$, conceptually, we can think of the state as the array below.

{
\begin{equation}
\textcolor{spec}{
[[\frac{1}{2}(-1)^{0}\ket{0}\ket{0},\frac{1}{2}(-1)^{1}\ket{0}\ket{1}],[\frac{1}{2}(-1)^{0}\ket{1}\ket{0},\frac{1}{2}(-1)^{1}\ket{1}\ket{1}]]
}
\label{eq:example1}
\end{equation}
}

Note that the locus for the above state is $x\srange{0}{n}\sqcupplus y[0]$, meaning that the two ket structures, $\textcolor{spec}{\ket{j}}$ and $\textcolor{spec}{\ket{b}}$, have different lengths if we think of $j$ and $b$ in their bitstring representations.
The above simple transformation can be summarized by viewing a quantum state as a predicate over an array of the same dimension as the number of linear-sum terms.
Although the state representation might be the right choice, such a simple transformation might not work for many cases.
In the example above, if we bit-flip $x[i]$ ($i\in[0,n)$), using $\qass{x[i]}{\lambda \ket{u}.\ket{(u\splus 1)\%2}}$, we need to develop a way to locate the specific bit in the right ket structure. Additionally, one cannot assume the qubit order in a ket structure is always well-ordered; e.g., the range $x\srange{0}{n}$ points to the ket $\textcolor{spec}{\ket{j}}$ above, but we cannot assume that $x[0]$ must point to the first bit in $\textcolor{spec}{\ket{j}}$'s bitstring and $x[n\sminus 1]$ points to the last, i.e., the order depends on the layout of the array.

We compile the above operation and its proof to a sequence of classical operations, easily verifiable in Hoare logic, and develop a general procedure to verify the operation via the proof for the sequence.
Assume that a set of loci partitions all quantum qubits in $\varphi$ into separable regions, and the qubits in the $\tch{2}$ typed locus $x\srange{0}{n}\sqcupplus y[0]\sqcupplus \eta$ are stored in a two-dimensional array $A$. The operation $\cn{ind}^{\varphi}(x\srange{0}{n}\sqcupplus y[0]\sqcupplus \eta)$ produces a set of indices in $A$ for the qubits in $x\srange{0}{n}\sqcupplus y[0]\sqcupplus \eta$ with the state being $\varphi$.
The array $A$ can have arbitrary dimensions, i.e., an $\tch{k}$ typed state is compiled to $k$-dimensional. Each element $\theta$ in $A$ has the following record type.

{\small
\[
\theta ::= \cn{element}\,\{\cn{amp} = z, \cn{vector} = c\}
\qquad
A ::= [ ... [ \theta ]_{\textcolor{spec}{n\sminus 1}} ...]_{\textcolor{spec}{0}}
\]
}

Assuming that $A$ is two-dimensional with two-level lengths $m_0$ and $m_1$, we identify four consecutive operations to construct the Hoare logic proof for the oracle operation above, as follows.

{\small
\[
\begin{array}{l}
\cassign{A_1}{\cn{locateN}(A,(m_0,m_1),\cn{ind}^{\varphi}(x\srange{0}{n}\sqcupplus y[0]))};
\\[0.2em]
\quad\cassign{A_2}{\cn{mapN}(A_1,(m_0,m_1),\lambda (t,(u,v)).(t,(u,v \oplus f(u))))};
\\[0.2em]
\qquad\cassign{A_3}{\cn{assignN}(A,A_2,(m_0,m_1),\cn{ind}^{\varphi}(x\srange{0}{n}\sqcupplus y[0]),\lambda u\,v.\cn{amp}(v))};
\\[0.2em]
\qquad\quad\cassign{A}{A_3}
\end{array}
\]
}

Below are example transitions that apply the above operations to the state in \Cref{eq:example1}.

{
\[
\begin{array}{l}
\xrightarrow{\cn{locateN},A_1}
\textcolor{spec}{
[[\frac{1}{2}(-1)^{0}\ket{0}\ket{0},\frac{1}{2}(-1)^{1}\ket{0}\ket{1}],[\frac{1}{2}(-1)^{0}\ket{1}\ket{0},\frac{1}{2}(-1)^{1}\ket{1}\ket{1}]]
}
\\
\xrightarrow{\cn{mapN},A_2}
\textcolor{spec}{
[[\frac{1}{2}(-1)^{0}\ket{0}\ket{1},\frac{1}{2}(-1)^{1}\ket{0}\ket{0}],[\frac{1}{2}(-1)^{0}\ket{1}\ket{1},\frac{1}{2}(-1)^{1}\ket{1}\ket{0}]]
}
\\
\xrightarrow{\cn{mapN},A_3}
\textcolor{spec}{
[[\frac{1}{2}(-1)^{0}\ket{0}\ket{1},\frac{1}{2}(-1)^{1}\ket{0}\ket{0}],[\frac{1}{2}(-1)^{0}\ket{1}\ket{1},\frac{1}{2}(-1)^{1}\ket{1}\ket{0}]]
}
\\
\xrightarrow{\cassign{A}{A_3}}
\textcolor{spec}{
[[\frac{1}{2}(-1)^{0}\ket{0}\ket{1},\frac{1}{2}(-1)^{1}\ket{0}\ket{0}],[\frac{1}{2}(-1)^{0}\ket{1}\ket{1},\frac{1}{2}(-1)^{1}\ket{1}\ket{0}]]
}
\end{array}
\]
}

$\cn{locateN}(A,(m_0,m_1),\cn{ind}^{\varphi}(x\srange{0}{n}\sqcupplus y[0]))$ shrinks high-dimensional element vectors to focus on the indices (e.g., $\overline{j}$) representing the locus fragment $x\srange{0}{n}\sqcupplus y[0]$ in $A$, as we assign the result to a new array $A_1$ (same dimension as $A$), as $\forall \overline{j}.A_1[\overline{j}]=\cn{vector}(A[\overline{j}])[\cn{ind}^{\varphi}(x\srange{0}{n}\sqcupplus y[0])]$. The $\cn{vector}$ operation accesses the vector $A[\overline{j}] $ but restricts it to the indices of the locus $x\srange{0}{n}\sqcupplus y[0]$. For the state in \Cref{eq:example1}, $A_1$ and $A$ are identical. 

The \cn{mapN} operation applies the function $(\lambda (t,(u,v)).(t,(u,v \oplus f(u))))$ to each element in $A_1$ and assigns the result to a new array $A_2$,
where $t$ and $(u,v)$ refer to the element's amplitude and vector, respectively.
In \Cref{fig:onetime-hw}, $f$ flips the result bit because it is constant, as the result after $\xrightarrow{\cn{mapN},A_2}$.
$\cn{assignN}$ merges the vector in $A_2$ back to $A$, i.e., each vector element ($\cn{vector}(A_3[\overline{j}])$) is the concatenation of $\cn{vector}(A_2[\overline{j}])[\cn{ind}^{\varphi}(x\srange{0}{n}\sqcupplus y[0])]$ and $\cn{vector}(A[\overline{j}])[\cn{ind}^{\varphi}(\eta)]$, where $\eta$ is the locus fragment outside of $x\srange{0}{n}\sqcupplus y[0]$. The function ($\lambda u\,v.\cn{amp}(v)$) determines how we merge the amplitudes of the two arrays; we use $A_2$'s amplitudes here.
Finally, we assign $A_3$ back to $A$.

The transformation is generalizable to other quantum proof systems. 
In Quantum Hoare Logic (QHL) \cite{10.1145/3456877}, a proof triple can be written as $\triple{P,q_1}{e}{Q,q_2}$, where $P$ and $Q$ are pre- and post-conditions for classical variables and $q_1$ and $q_2$ are quantum state predicates.
For example, the above proof triple for the oracle operation can be expressed as follows.

{\small
\[
\begin{array}{l}
\quad\textcolor{spec}{
 \Hassert{\cn{true},\ssum{j=0}{2^{n}\sminus 1}{\ssum{b=0}{1}{\frac{1}{\sqrt{2^{n\splus 1}}}(-1)^{b}\ket{j}_{x[0,n)}\ket{b}_{y[0]}}}}
 }
\\[0.5em]
\qquad\qass{x\srange{0}{n}\sqcupplus y[0]}{\lambda \ket{u}\ket{v}.\ket{u}\ket{v\oplus f(u)};}
\\[0.5em]
\quad\textcolor{spec}{
 \Hassert{\cn{true},\ssum{j=0}{2^{n}\sminus 1}{\ssum{b=0}{1}{\frac{1}{\sqrt{2^{n\splus 1}}}(-1)^{b}\ket{j}_{x[0,n)}\ket{b\oplus f(j)}_{y[0]}}}}
}
\end{array}
\]
}

Here, there are no conditions for classical variables.
In the quantum state predicates, we encode each range in the locus information in the corresponding basis-vector, e.g., $\ket{j}_{x[0,n)}\ket{b}_{y[0]}$ \footnote{QHL names $x[0,n)$ and $y[0]$ as quantum variables, while \qafny names them as ranges in a locus.}
Essentially, transforming the above triple is similar to the \qafny one, in which we still transform a state into a high-dimensional array.
We illustrate the QHL case by examples here and in \Cref{sec:frame}; a detailed discussion is in \Cref{appx:qhl} and the full QHL translation is left to future work.

\myparagraph{The main theorems.}\label{sec:guidedtypeexp}

We first present the paper's main result: quantum program verification can be effectively carried out within the classical Hoare logic proof system.
We use \qafny as the example proof system for the main theorems in the paper.
Given a \qafny proof triple, $\qpule{\sigma}{\varphi}{\mu}{\varphi'}$, there are three main theorems in the paper.

\begin{proposition}[Translation Completeness]\label[proposition]{prop:soundness}\rm 
Translating a derivable \qafny triple $\qpule{\sigma}{\varphi}{\mu}{\varphi'}$ results in a valid Hoare logic triple $\triple{P}{e}{Q}$.
\end{proposition}

\begin{proposition}[Translation Soundness]\label[proposition]{prop:Tightness}\rm 
Let a well-typed \qafny program $\mu$ and well-typed conditions $\varphi$ and $\varphi'$ translate to $e$, $P$, and $Q$. If the Hoare triple $\triple{P}{e}{Q}$ is valid, then \qafny derives a triple $\qpule{\sigma}{\varphi_1}{\mu}{\varphi'_1}$ with $\denote{\varphi}\equiv\denote{\varphi_1}$ and $\denote{\varphi'}\equiv\denote{\varphi'_1}$; by the semantic soundness of \qafny, $\mu$ satisfies the specification $(\varphi,\varphi')$.
\end{proposition}

\begin{proposition}[Translation Size]\label[proposition]{prop:polytrans}\rm 
The translation produces at most eleven HLIR statements per \qafny statement, and the translated assertions and the Hoare proof skeleton for each statement have size polynomial in the size of the \qafny triple and the array dimension, measured in the symbolic array representation of \Cref{sec:trans}.
\end{proposition}

Throughout, the ground truth is the quantum semantics. A \qafny triple is derived by the \qafny proof system, which is sound with respect to that semantics \cite{li2024qafnyquantumprogramverifier}; the classical side never sees quantum states, only the translated program $e$ and assertions $P$ and $Q$. Our two theorems relate derivability in \qafny to validity of the translated Hoare triple in both directions, and we name them from the perspective of Hoare logic as the verifier. Translation completeness says that nothing \qafny can prove is lost in translation: Hoare logic with array theory is expressive enough. Translation soundness says that the classical verifier cannot certify a false quantum claim: a valid translated triple corresponds to a \qafny derivation with semantically equivalent conditions and hence, by \qafny's semantic soundness, to a true statement about the program. Both directions are needed for the claim in the abstract.
The size bound shows that the translation is effective: the verification conditions it generates are polynomial in the input, so the classical verifier faces a problem no larger than the quantum one, up to its own proving power. The bound concerns the size of the conditions, not solver time.

We formally detail the compilation of this operation and others in the next section.

 \section{The Source Proof System: \qafny}

We introduce a part of \qafny as an example quantum proof system for demonstrating the compilation procedure.

\subsection{\qafny Syntax and A Running Example}\label{sec:runningexample}

\begin{figure*}[h]
{
  \small
  \[\begin{array}{llcl} 
      \text{\oqasm Expr} & \iota \\[0.3em]
      \text{Arith Expr} & a & ::= & x \mid x\srange{i}{j} \mid v \mid a \,\splus\, a \mid a \cdot a \mid ... \\[0.3em]
      \text{Bool Expr} & b & ::= & x[a] \mid \qbool{a_1}{=}{a_2}{{x[a]}} \mid \qbool{a_1}{<}{a_2}{{x[a]}} \mid ... \\[0.3em]
      \text{Quantum Expr} & op & ::= & \texttt{H} \mid \iqft[\lbrack -1 \rbrack]{}{} \mid \iota \\[0.3em]
      \text{Statement} & \mu & ::= & \ssassign{\eta}{}{op}\mid\cassign{p,x}{\smea{\eta}} \mid \sifq{b}{\mu} \mid
                               \\[0.3em] & & \textcolor{red}{ \mid} & \textcolor{red}{\sskip \mid \sseq{\mu}{\mu} \mid \sfori{j}{a_1}{a_2}{(\qinv\,P)^*}\;{\{\mu\}}}
\end{array}
  \]
}
  \caption{\qafny syntax. We focus on the operations marked in black, and the red ones are essentially classical. }
  \label{fig:vqimp-full}
\end{figure*}

\Cref{fig:vqimp-full} shows \qafny's syntax, which is imperative and designed with high-level quantum programming operations. 
A \qafny program consists of statements ($\mu$), including control flow constructs (\texttt{SKIP} (\(\sskip\)), for-loops, and quantum conditionals), classical assignments ($\cassign{x}{a}$), computational basis quantum measurement ($\cassign{p,x}{\smea{\eta}}$), and quantum \emph{heap mutations} ($\ssassign{\eta}{}{op}$). A measurement operation is a let-binding that measures the locus $\eta$, stores the classical value of the result in $x$ with the probability $p$, and uses the result in $\mu$. A heap mutation applies a quantum operation -- such as state preparation operations (Hadamard gate \cn{H} and quantum Fourier transformation gate \cn{QFT}), and a quantum arithmetic oracle operation ($\iota$) \footnote{
  $\iota$ can define all quantum arithmetic operations, e.g., \Cref{fig:onetime-hw} line 7. See \cite{oracleoopsla}.
  }--to a locus $\eta$; such behavior is analogous to mutating a qubit heap array by applying the operation to each basis-ket element in the array.

Quantum reversible Boolean guards $b$ are implemented as \oqasm oracle operations, expressed by one of $\qbool{a_1}{=}{a_2}{x[a]}$, $\qbool{a_1}{<}{a_2}{x[a]}$, and $x[a]$, which intuitively amounts to computing $a_1 = a_2$, $a_1 < a_2$ and $\cn{false}$ respectively as \(b_0\) and storing the result of $b_0 \oplus x[a]$ as a binary in qubit $x[a]$.\footnote{$a_1$ and $a_2$ can possibly apply to a range, like $y\srange{0}{n}$, in an entangled locus.}
In both conditionals and loops, guards \(b\) are used to represent the qubits being controlled.

\subsection{Quantum State Predicates and Qafny Triples}\label{sec:qafnystates}

Quantum program proofs are based on the \qafny triples $\qpule{\sigma}{\varphi}{e}{\varphi'}$, meaning that a quantum program verification for $e$ is valid if we start with a pre-condition $\varphi$ and the execution of $e$ satisfies the post-condition $\varphi'$, where $\sigma$ is the typed environment for $\varphi$ and $e$.
Here, $\varphi$ and $\varphi'$ are syntactic logical forms representing quantum states, e.g., the one in \Cref{fig:qafny-state}; such quantum state predicates are generic to different quantum proof systems.

\begin{figure*}[t]
{
\footnotesize
\[\hspace*{-1em}
\begin{array}{l}
\textcolor{blue}{\text{Basic Terms:}}\\[0.2em]
\begin{array}{l@{\;\;\;}l@{\;\;}c@{\;\;}l @{\qquad} l@{\;\;\;}l@{\;\;}c@{\;\;}l @{\qquad} l@{\;\;\;}l@{\;\;}c@{\;\;}l @{\qquad} l@{\;\;\;}l@{\;\;}c@{\;\;}l}
\text{Nat. Num} & \ell & \in & \mathbb{N}
&
      \text{Complex Amplitude} & z & \in & \mathbb{C}
&
      \text{Real} & p,r & \in & \mathbb{R}
&
\text{Phase} & \omega(r) & ::= & e^{2\pi i r}
\\
      \text{Variable} & x,y &&&
  \text{Function Variable} & f &
  &

 & \text{Bit/Bool} & d & ::= & 0\mid 1       
&
      \text{Bitstring} & c & \in & d^{+}
\end{array}
\\[0.7em]
\textcolor{blue}{\text{Modes, Kinds, Types, and Classical/Quantum Values:}}\\[0.2em]
\begin{array}{llcl@{\quad}c@{\quad}l@{\quad}c@{\quad}l@{\quad}c@{\quad}l} 
      \text{Classical Scalar Value} & v & ::= & \ell \mid ... \\
      \text{Kind} & \zeta & ::= & \cmode & \mid & \qmode{n} \\
            \text{Basis Vector}& \beta & ::= & (\ket{c})^*\\
      \text{Basic Ket} & w & ::= & z \beta\\[0.3em]
      \text{\textcolor{red}{Quantum Type}} & \textcolor{red}{\tau} & \textcolor{red}{::=} & \textcolor{red}{\tcht(n)} \\[0.5em]
      \text{Quantum Data} & q & ::= & {\displaystyle \Msum_{j_0=0}^{m_0}...\Msum_{j_{n-1}=0}^{m_{n-1}} w_{j_0,...,j_{n-1}}  }
    \end{array}
\\[0.5em]
\textcolor{blue}{\text{Quantum Loci and States}}\\[0.2em]
\begin{array}{llclcl} 
      \text{Qubit Array Range} & s & ::= & x\srange{a_1}{a_2} \\
      \text{Locus} & \eta & ::= & \overline{s}  & \text{concatenated op} & \sqcupplus \\
        \text{Type Environment} & \sigma & ::= & \overline{\eta : \textcolor{red}{\tau} } & \text{concatenated op} & \uplus \\
      \text{Typed Quantum State Predicate} & \varphi & ::= & \overline{ \eta :  \textcolor{red}{\tau} \mapsto q } &  \text{concatenated op} & \uplus\\
    \end{array}
\\[2em]
\textcolor{blue}{\text{Syntax Abbreviations and Basis/Locus Equations}}\\[0.2em]
   \Msum_{j=0}^{0}w_j \simeq w_0
\quad\;
  1 \beta \simeq \beta
\quad\;
\ket{c_1}\ket{c_2} \simeq \ket{c_1 c_2}
\quad\;
x[n,n) \simeq  \emptyset
\quad\;
\emptyset \sqcupplus \eta \simeq  \eta
\quad\;
x[n,m)  \simeq  x[n,j)\sqcupplus x[j,m)
\;\;
\texttt{if}\;\;n \le j \le m
\end{array}
  \]
}
  \caption{State syntax. The part marked in red is specific to \qafny, while the other part is generic across different quantum proof systems. A range $x\srange{a_1}{a_2}$ in a locus is the interval $\srange{a_1}{a_2}$ in a qubit array $x$. Loci are finite lists; states are finite sets. Natural-number literals are ranged over by $\ell$; range bounds are arithmetic expressions and may mention program variables, such as $n$ in $x\srange{0}{n}$. The ops after "concatenated op" are concatenations.}
  \label{fig:qafny-state}
\end{figure*}

A program verification in \qafny can be expressed as a separation logic triple $\qpule{\sigma}{\varphi}{e}{\varphi'}$.
\qafny verifies such a triple by constructing a proof tree, as in \citet{li2024qafnyquantumprogramverifier}.
We explain below how \qafny state representations are organized, and discuss the program syntax and proof translation in the following subsections.
As in \Cref{fig:qafny-state}, the states in \qafny triples are syntactic entities that represent quantum states and program transitions.
A state $\varphi$ describes program variable behavior, where variables are partitioned into two different \emph{kinds}: $\cmode$ for classical scalar values, such as natural numbers or reals, and $\qmode{n}$ for a physical $n$-length qubit array conceptually residing in a quantum heap.
We permit users to write scalar types, e.g., $\cn{nat}$ and $\cn{real}$, in a \qafny program, recognized as \cmode-kind.

In \qafny, we abstract a quantum state as a predicate that represents a heap structure, with qubits corresponding to physical heap locations. A state predicate is expressed as $\overline{\eta : \tau \mapsto q}$, meaning that a group of \emph{possibly entangled} qubits $\eta$ maps to a quantum datum $q$ of type $\tau$. Elements in a quantum state predicate are separated by the $\uplus$ operation, conceptually equivalent to the tensor operator $\otimes$. Similar to the separation conjunction in separation logic, it enables modular reasoning about the quantum state predicate and ensures that different unentangled qubit groups are disjoint.
Before we explain the types and loci, we show an example of expressing a \qafny triple from the Deutsch–Jozsa algorithm in \Cref{fig:onetime-hw} (left).
Line $1$ is the pre-state predicate before any operation is executed, with $\tch{0}$ typed data $\ket{\overline{0}}$ (abbreviation of $\textcolor{spec}{\bigotimes^{n}\ket{0}}$) for the separable qubit group $x\srange{0}{n}$.
Applying the Hadamard operation to $x\srange{0}{n}$ in line $2$ turns $x\srange{0}{n}$'s state to a $\tch{1}$ typed datum in line $3$.
The pre- and post-conditions and the operation form a \qafny triple as follows.

{\footnotesize
\[
\qpule{\sigma}{x\srange{0}{n}: \tch{0}\mapsto \ket{\overline{0}}}{\iassign{x\srange{0}{n}}{\ihadh}}{x\srange{0}{n}: \tch{1}\mapsto \ssum{j=0}{2^{n}\sminus 1}{\frac{1}{\sqrt{2^n}}\ket{j}}}
\]
}

\qafny utilizes types to classify quantum data, allowing quantum program verification managed in different views of a quantum state, e.g., every \qafny triple comes with a type environment $\sigma$, where its well-typedness is defined as $\sigma(\eta) = \cn{type}(\varphi(\eta))$ for all $x$. We use $\sigma$ to track changes across different quantum state types, leading to distinct verification strategies, as captured in the proof translation in \Cref{sec:semantics}.

As mentioned above, many entangled quantum data are not expressible as tensor compositions of individual qubits. Instead, they can be expressed in the basis-ket view, i.e., the quantum datum, classified as $\tch{n}$ type in \qafny, can be viewed as a linear sum (array) of basis-kets ($q$), each representing a possible outcome for measuring the qubits.
The $n$ flag represents the number of summation operations in a quantum state representation syntactically.
For example, we merge two qubit groups $x\srange{0}{n}$ and $y[0]$ into the entangled datum ($\tch{1}$) in line $4$, as the locus $x\srange{0}{n}\sqcupplus y[0]$.
Here, the whole datum has $2^n$ possible basis-ket states, $\textcolor{spec}{\ket{j} \ket{1}}$ for $j\in[0,2^n)$, i.e., the basis-vector fragment for $x\srange{0}{n}$ is represented as $j$ while the basis-vector fragment for $y[0]$ is $\textcolor{spec}{\ket{1}}$.

The $\tch{0}$ typed datum, e.g., $x\srange{0}{n}$'s state in line $1$, represents a non-superposition state, since it contains only a singleton basis-ket.
A single qubit state could be in superposition but not entangled, e.g., the state in \Cref{fig:onetime-hw} line $3$ can be represented as $n$-length qubit array ($\bigotimes^{n\sminus 1}_{j=0} x\srange{j}{j\splus 1}$), each of which is a $\tch{1}$ superposition state $\textcolor{spec}{\frac{1}{\sqrt{2}}(\ket{0}+\ket{1})}$, for the range $x\srange{j}{j\splus 1}$; the state is not entangled.

As a state predicate, physical qubit locations might be hard to track, e.g., the entangled qubit information in line $8$ is spread into different basis-kets. Therefore, we develop a \emph{locus} structure $x\srange{0}{n} \sqcupplus y[0]$ to indicate that the two \emph{ranges} $x\srange{0}{n}$ and $y[0]$ are possibly entangled. It is also an information location indicator, i.e., the range $x\srange{0}{n}$ points to the first $n$ bits in each basis-ket and $y[0]$ points to the last bit ($\textcolor{spec}{\ket{b\oplus f(j)}}$).

 \section{Embedding Separation Conjunction and the Frame Rule}\label{sec:frame}

Here, we discuss the translation of the frame rule in \qafny and QHL, and separation conjunctions, and present our first discovery: tensor-based quantum separation conjunctions are pure, a key contribution of the paper.

In separation logic, an assertion is pure when its truth depends only on the store and not on the heap.
When an assertion is pure, a separation conjunction ($*$) is equivalent to an ordinary conjunction ($\wedge$).
Here, we generalize the pureness as follows. Given $P*Q$, for any program operation, if we can show that modifying the operation in $Q$ will never affect the evaluation of $P$, there is pureness on the particular separation conjunction, and $P*Q$ is equivalent to $P \wedge Q$.

A non-pureness scenario is based on a predicate $P * l \mapsto v$, where $l$ is a memory heap location pointing to $v$. Even though $P$ might not mention $l$, a variable $x$ in $P$ could alias $l$, so modifying $v$ might affect the value in $x$.
This scenario cannot arise in \qafny: there is no qubit reference type, loci are ranges of named qubit arrays, and the type environment keeps the loci of a state disjoint, so two names never denote the same physical qubit. In a language with qubit references, or under a reference semantics, the disjointness assumption must be discharged separately, which is what a separating conjunction provides there.

\begin{example}[Pureness Example]\label[example]{example:pure}\rm As an example of showing pureness, the separation conjunction in the proof lines 1-3 in \Cref{fig:onetime-hw} can be replaced by a normal conjunction.

{\footnotesize
      \begin{mathpar}
        \inferrule[]{\triple{x\srange{0}{n}: \tch{0}\mapsto \ket{\overline{0}} }{\qass{x\srange{0}{n}}{\cn{H}};}{x\srange{0}{n}: \tch{1}\mapsto \ssum{j=0}{2^{n}\sminus 1}{\frac{1}{\sqrt{2^n}}\ket{j}}}}
      {\triple{x\srange{0}{n}: \tch{0}\mapsto \ket{\overline{0}} \wedge  y[0] : \tch{0} \mapsto \ket{1}}{\qass{x\srange{0}{n}}{\cn{H}};}{x\srange{0}{n}: \tch{1}\mapsto \ssum{j=0}{2^{n}\sminus 1}{\frac{1}{\sqrt{2^n}}\ket{j}} \wedge  y[0] : \tch{0} \mapsto \ket{1}}}
      \end{mathpar}
}

If we have the example in QHL, the same separation conjunction can also be rewritten as a normal conjunction. Here, $\varphi_{\overline{s}}$ can be viewed as a predicate $\overline{s} \mapsto \varphi$, e.g., $\ket{\overline{0}}_{x\srange{0}{n}}$ means $x\srange{0}{n} \mapsto \overline{0}$.

{\footnotesize
      \begin{mathpar}
        \inferrule[]{\triple{\ket{\overline{0}}_{x\srange{0}{n}} }{\qass{x\srange{0}{n}}{\cn{H}};}{\ssum{j=0}{2^{n}\sminus 1}{\frac{1}{\sqrt{2^n}}\ket{j}_{x\srange{0}{n}}}}}
      {\triple{ \ket{\overline{0}}_{x\srange{0}{n}} \wedge \ket{1}_{y[0]}}{\qass{x\srange{0}{n}}{\cn{H}};}{\ssum{j=0}{2^{n}\sminus 1}{\frac{1}{\sqrt{2^n}}\ket{j}_{x\srange{0}{n}}} \wedge \ket{1}_{ y[0]}}}
      \end{mathpar}
}
\end{example}

In the above example, we replace the tensor operators connecting different unentangled qubit regions with normal conjunctions.
In quantum operations, if we know that the operation, e.g., $\qass{x\srange{0}{n}}{\cn{H}}$, does not mention a qubit region $y[0]$, the operation will not affect the region.
This is different from the need for separation conjunction, where a point operation might affect a memory field even if the operation does not explicitly mention the field.
We have the following theorem.

\begin{theorem}[Pureness]\label{def:pureness}\rm 
Given $\mu$, $\varphi_1$, $\varphi'_1$, and $\varphi_2$, where $\varphi_1$ and $\varphi_2$ describe disjoint sets of qubits and $FV(\mu) \cap FV(\varphi_2) = \emptyset$, we have the following in QHL.
\[\triple{P,\varphi_1 \otimes \varphi_2}{\mu}{Q,\varphi'_1 \otimes \varphi_2} \longleftrightarrow \triple{P,\varphi_1 \wedge \varphi_2}{\mu}{Q,\varphi'_1 \wedge \varphi_2},\]
as well as the following in \qafny.
\[\qpule{\sigma_1 \uplus \sigma_2}{\varphi_1 * \varphi_2}{\mu}{\varphi'_1 * \varphi_2} \longleftrightarrow \qpule{\sigma_1 \uplus \sigma_2}{\varphi_1 \wedge \varphi_2}{\mu}{\varphi'_1 \wedge \varphi_2}.\]

Here, $T_1 \longleftrightarrow T_2$ means that $T_1$ is valid if and only if $T_2$ is valid. We also require $\sigma_1 \vdash \mu \triangleright \sigma'_1$, $\sigma_1 \vdash \varphi_1$, $\sigma_2\vdash\varphi_2$, and $\dom{\sigma_1}\cap\dom{\sigma_2}=\emptyset$: the loci of $\varphi_1$ and $\varphi_2$ are disjoint, which the \qafny type environment enforces and which the translated frame rule below turns into the classical side condition $FV(R)\cap \cn{mod}(\cn{trans}(\mu))=\emptyset$.
\end{theorem}

The proof of the above theorem for the \qafny version relies on its type system, and the proof for QHL relies on a clear separation for distinguishing qubit locus regions and demonstrates that the execution of $\mu$ does not affect $\varphi_2$.
The \qafny locus type system purely separates two state conditions if they mention disjoint loci, as we have a stronger result for \qafny below.

\begin{theorem}[Qafny Pureness]\label{def:qafnypureness}\rm 
Given $\sigma_1 \vdash \mu \triangleright \sigma'_1$ and $\qpule{\sigma_1 \uplus \sigma_2}{\varphi_1 * \varphi_2}{\mu}{\varphi'}$, such that $\sigma_1 \vdash \varphi_1$ and $\sigma_2\vdash\varphi_2$, there exists $\varphi'_1$, such that $\sigma'_1 \vdash \varphi'_1$ and $\varphi' = \varphi'_1 * \varphi_2$.
\end{theorem}

\begin{proof}
The proof utilizes the \qafny proof system soundness \cite{li2024qafnyquantumprogramverifier} \footnote{This is \qafny's soundness with respect to the quantum semantics, not one of the translation theorems of this paper.}, as we inductively show that for any $e$, if a state $\varphi$ can be partitioned into $\varphi_1$ and $\varphi_2$ based on disjoint loci in $\dom{\sigma_1}$ and $\dom{\sigma_2}$ and $e$ only operates on $\varphi_1$, then $\varphi_2$ is never affected.
\end{proof}

{\small
      \begin{mathpar}
    \inferrule[]{\dom{\sigma}\cap FV(\varphi_2)=\emptyset \\\\
       \cn{mod}(\mu)\subseteq \dom{\sigma} \\ \qpule{\sigma}{\varphi_1}{\mu}{\varphi'_1}}
                {\qpule{\sigma\uplus\sigma'}{\varphi_1 * \varphi_2}{\mu}{\varphi'_1 * \varphi_2}}
                \;\;
\longrightarrow\;\;
    \inferrule[]{FV(R)\cap \cn{mod}(\cn{trans}(\mu))=\emptyset\\\triple{P}{\cn{trans}(\mu)}{Q}}
                {\triple{P \wedge R}{\cn{trans}(\mu)}{Q \wedge R}}
                      \end{mathpar}
}

\Cref{def:pureness} suggests that separation conjunctions can be replaced by ordinary ones if the two pieces describe disjoint loci. It indicates that the separation conjunctions in the \qafny and QHL frame rule -- where the \qafny one is listed on the left above -- can be rewritten to a formula containing only ordinary conjunctions; thus, translating the two quantum proof systems to Hoare logic will not require adding separation conjunctions.
Here, $\cn{trans}(\mu)$ represents the translated Hoare logic program for $\mu$, $P$, $Q$, and $R$ are the translated predicates for $\varphi_1$, $\varphi'_1$, and $\varphi_2$.
The details of these translated functions are given in \Cref{sec:trans}.

In fact, the above result is general for all the tensor ($\otimes$) based quantum separation conjunctions \cite{qsepa,qseplocal}: once the disjointness of the two qubit regions is tracked, as the \qafny type environment does, these separation conjunctions are pure, i.e., they can be treated as normal conjunctions.

 \section{The Intermediate Representation Language based on Hoare Logic}
\label{sec:semantics}

We describe HLIR based on Hoare logic, where every operation is definable and provable as a composition of operations in the classical Hoare logic language with array assignments. Mainly, we discuss the syntax (\Cref{fig:hlir-syntax}), the informal meaning, and array-state representations. 
We detail HLIR with proof examples, along with our translation proof from \qafny to HLIR in the next section.

\begin{figure*}[h]
{\hspace*{-1em}
  \small
  \[\begin{array}{llcl} 
      \text{Index Function} & \rho & ::= & \cn{ind}^{\varphi}(\eta)\\[0.3em]
      \text{Array Element} &\theta & ::=& \cn{element}\,\{\cn{amp} = z, \cn{vector} = c\} \\[0.3em]
      \text{Function} &\chi & ::=& \cn{up}^a(\theta,z) \mid \cn{up}^b(\theta,c,\rho) \mid \cn{range}(\overline{m}) \mid Q_{\cn{len}}(A,\overline{m}) \mid \cn{enum}(n) \mid \cn{merge}(\rho,c,c)\\[0.3em]
      \text{Statement} & e & ::= &  \cassign{A'}{\cn{mapN}(A,\overline{m},f)} \mid \cassign{A'}{\cn{filterN}(A,b,\rho)} \mid \cassign{A_2}{\cn{assignN}(A,A_1,\overline{m},\rho,f)}\\[0.2em]
      && \mid & \cassign{r}{\cn{sqrtSum}(A,\overline{m})}\mid \cassign{A_2}{\cn{concat}(A,A_1)} \mid\cassign{\overline{m}}{\slen{A}} 
      \\[0.2em]
            && \mid & \textcolor{blue}{\cassign{A'}{\cn{locateN}(A,\overline{m},\rho)}} \mid \textcolor{blue}{\cassign{A'}{ \cn{copyN}(A,\overline{m},\rho)}} \mid \textcolor{blue}{\cassign{A'}{\cn{enumN}(A,\overline{m},\rho)}}
      \\[0.3em] & &
     \mid & \textcolor{blue}{\cassign{A_2}{\cn{ampN}(A,A_1,\overline{m},\rho)}}
                                 \mid \textcolor{blue}{\cassign{A'}{\cn{upDim}(A, \overline{m},\overline{m'})}}
    \end{array}
  \]
}
  \caption{HLIR Syntax }
  \label{fig:hlir-syntax}
\end{figure*}

\noindent\textbf{\textit{HLIR State Representation.}}
As we mentioned, for an $\tch{n}$-typed quantum state, loci partition the state into disjoint qubit regions, where every fragment in a locus is disjoint from other fragments in the locus. We model the $\tch{n}$-typed state in Hoare logic as an $n$-dimensional array, containing elements $\theta$, modeled as a record in \Cref{fig:hlir-syntax}.
The array indexing operation $A[j_0,...,j_{n\sminus 1}]$ can be abbreviated as $A[\overline{j}]$, which represents a basis-ket state; the array can be thought of as a linear sum of basis-kets. In each element, one can access the amplitude via $\cn{amp}$ and the vector via $\cn{vector}$. The two functions $\cn{up}^a(\theta,z)$ and $\cn{up}^b(\theta,c,\rho)$ update the $\cn{amp}$ of an element and the $\rho$ positions of its \cn{vector}, respectively.

Essentially, \cn{vector} is a bit sequence. We maintain a map $\rho \triangleq \cn{ind}^{\varphi}(\eta)$ from the qubits in the locus fragment $\eta$ to the indices in the bit sequence, to record the index positions.
As an instance, \Cref{eq:example1} is an array representation of a quantum state, pointed to by the locus $x[0]\sqcupplus y[0]$, where $\cn{ind}^{\varphi}(x[0])$, with $\varphi$ being the state in \Cref{eq:example1}, points to the index position $0$, while $y[0]$ points to the index position $1$.
We have the following lemma for $\cn{ind}$.

\begin{lemma}[Disjointness Preservation]\label[lemma]{def:disjointness}\rm 
Given a locus $\eta$, an $\tch{n}$ typed state $\varphi$, and a locus map $\cn{ind}$, where  $\cn{indices}(A)=\cn{ind}^{\varphi}(\eta)$ and for every $j_0$,...,$j_{n\sminus 1}$, we have $\slen{A[j_0,...,j_{n\sminus 1}]}=\slen{\eta}$, and for every locus fragment $\eta'$, s.t. $\eta' \subseteq \eta$ and $\cn{indices}(A')=\cn{ind}^{\varphi}(\eta')$, for $j_0,...,j_{n\sminus 1}$, $\cn{vector}(A'[j_0,...,j_{n\sminus 1}])$  $\subseteq$  $\cn{vector}(A[j_0,...,j_{n\sminus 1}])$.
\end{lemma}

The above lemma has a hidden assumption when applied to \qafny: we assume we can always refer to a finite global universe containing all qubits in the state $\varphi$. Such a statement is valid in \qafny as well as many other quantum proof systems. We present the following lemma, which states that the resulting quantum state after executing a \qafny program is always a subset of the input state. Thus, for any program, there is a quantum state $\varphi$ that acts as its universe.

\begin{lemma}[Quantum State Shrinking]\label[lemma]{def:shrinking}\rm 
Given a state $\varphi$, a \qafny program $\mu$, and a type environment $\sigma$, such that $\sigma \vdash \varphi$, $\sigma \vdash \mu \triangleright \sigma'$, and $(\varphi,\mu)\Downarrow \varphi'$ (the \qafny program $\mu$ is evaluated to $\varphi'$ by the input state $\varphi$), we have $\dom{\varphi'}\subseteq \dom{\varphi}$.
\end{lemma}

\noindent\textbf{\textit{Briefing the HLIR Operation Syntax.}}
The functions and statements in \Cref{fig:hlir-syntax} are used in the translation proof in \Cref{sec:trans}.
We provide informal descriptions below.

We have a set of helper functions to define small array properties.
The operations $\cn{range}(\overline{m})$ and $Q_{\cn{len}}(A,\overline{m})$ are described in \Cref{sec:background}.
$\cn{enum}(n)$ enumerates numbers from $0$ to $n\sminus 1$, $\cn{merge}(\rho, c_1,c_2)$ creates a new vector by replacing $c_1$'s $\rho$ indices with $c_2$, and $\cn{vector}(c)[\rho]$ shrinks vector $c$ to $\rho$.

We divide the statements into two sets: a set (marked black) of essential operations and a set (marked \textcolor{blue}{blue}) of compositional operations definable in terms of the essential set.
Most of the operations in the essential set are standard functional programming operations.
$\cn{mapN}$ is a high-order map operation mapping the function $f$ to every element in the array $A$,
$\cn{filterN}$ filters out elements in an array ($A$) that do not satisfy the boolean function $b$,
$\cn{concat}$ concatenates two arrays if they have the same dimensions,
$\slen{A}$ produces the array lengths, which can be high-dimensional,
and $\cn{sqrtSum}$ computes the square root of the sum of the squared amplitude norms over all array elements, as it can be built on top of a higher-order folding function.
$\cn{assignN}(A,A_1,\overline{m}, \rho,f)$ is an update function that updates $A$ with $A_1$'s element amplitudes and vectors (the shrunk $\rho$ fragments), assuming $\slen{A}=\slen{A_1}$, where $f$ defines how to combine the amplitudes of $A$ and $A_1$'s elements.

The set of compositional operations is definable by the essential set. We examine them in \Cref{sec:trans} and briefly describe them here.
$\cn{locateN}$ creates a new array $A'$ by shrinking the vectors in every element to the indices $\rho$.
$\cn{copyN}(A,\overline{m}, \cn{ind}^{\rho}(\eta))$ and $\cn{enumN}(A,\overline{m}, \cn{ind}^{\rho}(\eta))$ are special case of \cn{mapN}.
The former creates an $n+1$ dimensional array $A'$ if $A$ is $n$ dimensional, by replacing each element with a list of $2^{\slen{\eta}}$ copies of that element.
If the inner dimension has length $2^{\slen{\eta}}$, $\cn{enumN}$ creates an array $A'$ by copying $A$ but differing in terms of the vectors, where $\cn{vector}(A'[\overline{j},i],\cn{ind}^{\rho}(\eta)) = i$ and $\cn{vector}(A'[\overline{j},i], \{x|x\not\in\cn{ind}^{\rho}(\eta)\}) = A[\overline{j},i]$.
$\cn{ampN}(A,A_1,\overline{m},\rho)$ multiplies $A_1$ and $A$'s element amplitudes.
$\cn{upDim}(A,\overline{m}, \overline{m'})$ upgrades $A$'s dimensions to $\overline{m'}$, assuming that $\overline{m}\sqsubseteq \overline{m'}$ (defined in \Cref{sec:background}).
The operation fills the new dimensions by copying the elements in the original $A$.

\section{Translating Qafny Operations to HLIR}\label{sec:trans}

This section describes the translation proof from \qafny to HLIR.
The translation of Qafny to HLIR turns Qafny proof triples ($\qpule{\sigma}{\varphi}{\mu}{\varphi'}$) to Hoare logic triples ($\triple{P}{e}{Q}$), including two pieces: translating the \qafny program $\mu$ to an IR code $e$ and, respectively, translating the pre- and post- conditions $\varphi$ and $\varphi'$ to $P$ and $Q$.
The two pieces can be described by two judgments: $\sigma \vdash \mu \triangleright \sigma' \gg e$ and $\sigma \vdash \varphi \gg P$.
The former extends the \qafny type system ($\sigma \vdash \mu \triangleright \sigma'$) and the latter extends the well-typed quantum state ($\sigma \vdash \varphi$).
In \Cref{sec:oracletrans,sec:hadtrans,sec:condtrans,sec:meatrans,sec:transclassical}, we discuss translations for different \qafny operations, along with validity arguments, the completeness direction: for any derivable \qafny triple for an operation, we show that its translated Hoare triple is valid. 
Finally, we state the translation completeness, translation soundness, and translation size theorems in \Cref{sec:proofs}.

The translations are organized by the effect of an operation on the array representation.
(1) Operations that create superposition add a dimension: \cn{H} and \cn{QFT} (and, in the same way, rotations that create superposition) translate to \cn{copyN}, \cn{enumN}, and \cn{ampN} (\Cref{sec:hadtrans}); the amplitudes need not be uniform, since \cn{ampN} multiplies by an arbitrary function of the old and new basis vectors, of which the uniform Hadamard amplitude is one instance.
(2) Operations that map elements in place, i.e., basis permutations and phases such as \cn{X}, \cn{RZ}, and the arithmetic oracles $\iota$, translate to \cn{locateN}, \cn{mapN}, and \cn{assignN} (\Cref{sec:oracletrans}).
(3) Controlled operations, i.e., quantum conditionals, split the array with \cn{filterN}, translate the body, and reassemble it with \cn{concat} (\Cref{sec:condtrans}).
(4) Measurement filters and renormalizes (\Cref{sec:meatrans}), and classical statements are translated structurally (\Cref{sec:transclassical}).
Any gate given by its action on basis kets falls into class (1) or (2), possibly under a control of class (3).

\subsection{Translating Oracle Operations}\label{sec:oracletrans}
Here, we discuss the transformation of the oracle operation proof rule below to HLIR.

{
      \begin{mathpar}
      \inferrule[]{}
      {\qpule{\{\eta \sqcupplus \eta' : \tch{n}\}}{\eta \sqcupplus \eta' \mapsto q}
        {\ssassign{\eta}{}{\mu}}{\eta \sqcupplus \eta' \mapsto \denote{\mu}^{\slen{\eta}}q}}
        \end{mathpar}
}

This rule relies on a frame rule to focus state predicates to the locus $\eta \sqcupplus \eta'$. Because the \qafny separation conjunction is pure, it is no different from a normal conjunction. We can rewrite the above \qafny proof rule to the version below, which relates a general post-condition $\varphi$ to a pre-condition, in a style similar to the Hoare logic array assignment rule.

{
      \begin{mathpar}
      \inferrule[]{}
      {\qpule{\{\eta \sqcupplus \eta' : \tch{n}\}\uplus \sigma}{\varphi[(\eta \sqcupplus \eta')\{\eta \rightarrow \denote{\mu}(\eta)\}/\eta \sqcupplus \eta']}{\ssassign{\eta}{}{\mu}}{\varphi}}
        \end{mathpar}
}

To translate the proof rule to Hoare logic, we construct a composition of classical operations representing the effect of the operation $\ssassign{\eta}{}{\mu}$, including three components: the operations \cn{locateN}, \cn{mapN}, and \cn{assignN}, mentioned in \Cref{sec:guidedtypeexp}. 
One example is given in \Cref{sec:relationclassical}.

\myparagraph{Validity Proof.}
We state the translated Hoare triple as a lemma below, showing its validity.

\begin{lemma}[Valid Oracle Application Transformation]\label[lemma]{def:mult-mut}\rm 
Given a quantum state $\varphi$ and an operation $\ssassign{\eta}{}{\mu}$ in the locus $\eta \sqcupplus \eta'$, with the $\tch{n}$-typed quantum datum $\varphi(\eta\sqcupplus \eta')$, transformed to an $n$ dimensional array $A$ having the collection of indices ($\cn{ind}^{\varphi}(\eta\sqcupplus \eta')$), where $\cn{lens}(A)=\overline{m} =m_0,...,m_{n\sminus 1}$, we have the following valid Hoare triple.

{\footnotesize
\[\hspace*{-1em}
\begin{array}{l}
\textcolor{spec}{\bil 
\begin{array}{l}
P[A\{\forall \overline{j}\in \cn{range}(\overline{m}).\overline{j}\mapsto \cn{up}^{b}(A[\overline{j}],\denote{\mu}(\cn{vector}(A[\overline{j}])[\cn{ind}^{\varphi}(\eta)]),\cn{ind}^{\varphi}(\eta)))\}/A] 
\\[0.2em]
  \wedge Q_{\cn{len}}(A\{\forall \overline{j}\in \cn{range}(\overline{m}).\overline{j}\mapsto \cn{up}^{b}(A[\overline{j}],\denote{\mu}(\cn{vector}(A[\overline{j}])[\cn{ind}^{\varphi}(\eta)]),\cn{ind}^{\varphi}(\eta)))\},\overline{m})
  \\[0.2em]
\end{array}  
\bir }
\\[0.8em]
\quad\cassign{A_1}{\cn{locateN}(A,\overline{m},\cn{ind}^{\varphi}(\eta))}; \cassign{A_2}{\cn{mapN}(A_1,\overline{m},\lambda (x,y).(x,\denote{\mu}y))};
\\[0.2em]
\quad\qquad
 \cassign{A_3}{\cn{assignN}(A,A_2,\overline{m},\cn{ind}^{\varphi}(\eta),\lambda u\,v.\cn{amp}(u))};\cassign{A}{A_3}
\\[0.5em]
\textcolor{spec}{\bil P \wedge Q_{\cn{len}}(A,\overline{m}) \bir}
\end{array}
\]
}
\end{lemma}

Here, $\forall \overline{j}\in \cn{range}(\overline{m})$ refers to $\forall j_0\in [0,m_0),..., j_n\in [0,m_n)$ where $\overline{j}$ and $\overline{m}$ are expanded to $(j_0,...,j_n)$ and $(m_0,...,m_n)$; $\cn{up}^{b}$ and $\cn{up}^{a}$ are defined as follows, where $\cn{merge}(\rho,c,c^{\prime})$ (\Cref{fig:hlir-syntax}) replaces the $\rho$ positions of $c$ with the bits $c^{\prime}$, with $\slen{c^{\prime}}=\slen{\rho}$.

{\small
\[
\begin{array}{l}
\cn{up}^{b}(\cn{element}\,\{\cn{amp} = z, \cn{vector} = c\},c',\rho) =\cn{element}\,\{\cn{amp} = z, \cn{vector} = \cn{merge}(\rho,c,c')\}
\\[0.2em]
\cn{up}^{a}(\cn{element}\,\{\cn{amp} = z, \cn{vector} = c\},z') =\cn{element}\,\{\cn{amp} = z', \cn{vector} = c\}
\end{array}
\]
}

\begin{proof}\rm 
The validity proof utilizes \Cref{def:locate,def:gassign}. Its mechanized proof is in Rocq.
\end{proof}

We define and prove below the three operations $\cn{locateN}$, $\cn{mapN}$, and $\cn{assignN}$.
Essentially, \cn{locateN} and \cn{mapN} are similar, where the former is a special case of the latter, described below
\[\cn{locateN}(A, \overline{m}, \cn{ind}^{\varphi}(\eta)) = \cn{mapN}(A,\overline{m}, \lambda (x,y).(x,y[\cn{ind}^{\varphi}(\eta)])).\]
We now show the proof for \cn{mapN}.
For all the proofs, we assume the array length axioms $Q_{\cn{len}}(A,A',...,\overline{m})$, stating that the $ n$-dimensional arrays $A,A',...$ have $ n$-dimensional lengths $\overline{m}=m_0,...,m_{n\sminus 1}$.

\begin{lemma}[The Valid MapN Definition]\label[lemma]{def:locate}\rm 
Given $n$ dimensional arrays $A$ and $A'$, with $Q_{\cn{len}}(A,A',\overline{m})$, and a function $f$, we have the valid Hoare triple for $\cassign{A'}{\cn{mapN}(A,\overline{m},f)}$:

{\small
\[
\begin{array}{l}
\textcolor{spec}{\bil P[A'\{\forall \overline{j}\in\cn{range}(\overline{m}).\overline{j}\mapsto f(A[\overline{j}])\}/A'] \wedge Q_{\cn{len}}(A,A',\overline{m}) \bir}\\[0.2em]
\qquad \cassign{A'}{\cn{mapN}(A,\overline{m},f)}\\[0.2em]
\textcolor{spec}{\bil P \wedge Q_{\cn{len}}(A,A',\overline{m}) \bir }\\[0.2em]
\end{array}
\]
}

where $\cassign{A'}{\cn{mapN}(A,\overline{m},f)}$ is defined as: 

{\small
\[
\begin{array}{l}
\cassign{\overline{j}}{0};\\
\WHILE(j_{0} < m_{0})\{ ... \\
\qquad\WHILE(j_{n\sminus 1} < m_{n\sminus 1}) \{ \\
\qquad\qquad \cassign{A'[k_0,...,k_{n\sminus 1}]}{f(A[k_0,...,k_{n\sminus 1}])};  \cassign{j_{n\sminus 1}}{j_{n\sminus 1}+1} 
\\\qquad
\}
\\
...\; \cassign{j_{0}}{j_{0}+1}\}
\end{array}
\]
}
\end{lemma}

\begin{proof}\rm 
By \Cref{cor:link}, it is enough to show the validity of the following triple.

{\small
\[
\textcolor{spec}{\bil  Q_{\cn{len}}(A,A',\overline{m}) \bir}
\;\;
\cassign{A'}{\cn{mapN}(A,\overline{m},f)}\;\;
\textcolor{spec}{\bil \forall \overline{j} \in \cn{range}(\overline{m}) . A'[\overline{j}]=f(A[\overline{j}])  \wedge Q_{\cn{len}}(A,A',\overline{m}) \bir }
\]
}

The proof-tree derivation, with loop invariants at each level, is given in \Cref{appx:mapN}.
\end{proof}

$\cn{assignN}(A,A_2,\overline{m},\rho,\lambda u\,v.f(u,v))$ updates $A$'s element amplitudes with the combination of $A$ and $A_2$'s amplitudes via $\lambda u\,v.f(u,v)$, as well as updates $A$'s bit-vectors, only for the $\rho$ part, with $A_2$'s.  
We now show the validity lemma for \cn{assignN}.

\begin{lemma}[The Valid AssignN Definition]\label[lemma]{def:gassign}\rm 
Given $n$ dimensional arrays $A$, $A_1$, and $A_2$ (assuming that $A_2$ is pre-initialized to an $n$ dimensional array with bounds $\overline{m}$), with $Q_{\cn{len}}(A,A_1,A_2,\overline{m})$, we have the following \emph{valid} Hoare triple.

{\small
\[\hspace*{-1em}
\begin{array}{l}
\textcolor{spec}{\bil 
\begin{array}{l}
P[A_2\{\forall \overline{j} \in \cn{range}(\overline{m}). \overline{j} \mapsto  \cn{up}^a(\cn{up}^b(A[\overline{j}],\cn{vector}(A_1[\overline{j}]),\cn{ind}^{\varphi}(\eta)),f(A[\overline{j}],A_1[\overline{j}]))\}/A_2]
\\\quad
\wedge
Q_{\cn{len}}(A_2,A,A_1,\overline{m})
\end{array}
\bir}\\[0.2em]
\qquad {\cassign{A_2}{\cn{assignN}(A,A_1,\overline{m},\cn{ind}^{\varphi}(\eta),\lambda u\,v.f(u,v))}}\\[0.2em]
\textcolor{spec}{
\bil
P
\wedge
Q_{\cn{len}}(A_2,A,A_1,\overline{m})
\bir
}
\end{array}
\]
}
\end{lemma}

\begin{proof}\rm 
By \Cref{cor:link}, it is enough to show the validity of the following triple.

{\small
\[
\begin{array}{l}
\textcolor{spec}{\bil 
Q_{\cn{len}}(A_2,A,A_1,\overline{m})
\bir}\\[0.2em]
\qquad {\cassign{A_2}{\cn{assignN}(A,A_1,\overline{m},\cn{ind}^{\varphi}(\eta),\lambda u\,v.f(u,v))}}\\[0.2em]
\textcolor{spec}{
\bil
\forall \overline{j} \in \cn{range}(\overline{m}).
A_2[\overline{j}] =
\cn{up}^a(\cn{up}^b(A[\overline{j}],\cn{vector}(A_1[\overline{j}]),\cn{ind}^{\varphi}(\eta)),f(A[\overline{j}],A_1[\overline{j}]))
\wedge
Q_{\cn{len}}(A_2,A,A_1,\overline{m})
\bir
}
\end{array}
\]
}

The rest of the proof derives a Hoare logic proof tree and has a procedure as the one in \Cref{def:locate}, except that we update $A_2[\overline{j}]$ with $\textcolor{spec}{\cn{up}^a(\cn{up}^b(A[\overline{j}], \cn{vector}(A_1[\overline{j}]),\cn{ind}^{\varphi}(\eta)),f(A[\overline{j}],A_1[\overline{j}]))}$. The proof-tree derivation is given in \Cref{sec:assignproof}.

\ignore{
{\small
\[
\cn{let  }\textcolor{spec}{E[\overline{i}]} := \textcolor{spec}{\cn{up}^a(\cn{up}^b(A[\overline{i}],
\cn{vector}(A_1[\overline{i}]),\cn{ind}^{\varphi}(\eta)),f(A[\overline{i}],A_1[\overline{i}]))}
\]
}

{\small
\[
\begin{array}{ll}
\textcolor{spec}{\bil Q_{\cn{len}}(A_2,A,A_1,(m_0,...,m_{n\sminus 1})) \bir}\\[0.2em]
\qquad \cassign{j_0}{0};... \cassign{j_{n\sminus 1}}{0}; & \textbf{Assign}\\[0.2em]
\textcolor{spec}{\bil j_0 = 0 \wedge ... \wedge j_{n\sminus 1} = 0 \wedge Q_{\cn{len}}(A_2,A,A_1,(m_0,...,m_{n\sminus 1})) \bir}
\\[0.2em]
\Rightarrow
\textcolor{spec}{\bil \begin{array}{l} Q_{\cn{len}}(A_2,A,A_1,(m_0,...,m_{n\sminus 1})) \wedge j_0 \in[0,m_0] \wedge ...\wedge j_{n\sminus 1} \in[0,m_{n\sminus 1}] \\
\wedge \forall k_0 \in [0,j_0).\,...\forall k_{n\sminus 1} \in[0,j_{n\sminus 1}).\, A_2[k_0,...,k_{n\sminus 1}] = E[k_0,...,k_{n\sminus 1}]
\end{array}
\bir}
&
\textbf{Pre-Strength}
\\[0.2em]
\qquad \WHILE(j_0 < m_0) 
&
\textbf{While}
\\[0.2em]
\qquad \textcolor{spec}{
\bil
\begin{array}{l}
Q_{\cn{len}}(A_2,A,A_1,(m_0,...,m_{n\sminus 1})) \wedge
j_0 \in[0,m_0] \wedge ...\wedge j_{n\sminus 1} \in[0,m_{n\sminus 1}] \\
\wedge \forall k_0 \in [0,j_0).\,...\forall k_{n\sminus 1} \in[0,j_{n\sminus 1}).\, A_2[k_0,...,k_{n\sminus 1}] = E[k_0,...,k_{n\sminus 1}]
\end{array}
\bir}
&
\textbf{Invariant}
\\[0.2em]
\qquad \{\\[0.2em]
\qquad\quad 
\{...
\\[0.2em]
\qquad\qquad
\WHILE(j_{n\sminus 1} < m_{n\sminus 1}) \\[0.2em]
\qquad\qquad \textcolor{spec}{
\bil
\begin{array}{l}
Q_{\cn{len}}(A_2,A,A_1,(m_0,...,m_{n\sminus 1})) \wedge
j_0 \in[0,m_0] \wedge ...\wedge j_{n\sminus 1} \in[0,m_{n\sminus 1}] \\
\wedge \forall k_{n\sminus 1} \in[0,j_{n\sminus 1}).\, A[j_0,...,j_{n\sminus 2},k_{n\sminus 1}] = E[j_0,...,j_{n\sminus 2},k_{n\sminus 1}]
\end{array}
\bir
}
&
\textbf{Invariant}
\\[0.2em]
\qquad\qquad\{\\[0.2em]
\qquad\qquad\quad
\cassign{tmp}{\cn{up}^b(A[k_0,...,k_{n\sminus 1}],\cn{vector}(A_1[k_0,...,k_{n\sminus 1}]),\cn{ind}^{\varphi}(\eta))};
\\[0.2em]
\qquad\qquad\quad
\cassign{A_2[k_0,...,k_{n\sminus 1}]}{\cn{up}^a(tmp,f(A[k_0,...,k_{n\sminus 1}],A_1[k_0,...,k_{n\sminus 1}]))}
\\[0.2em]
\qquad\qquad\quad
\cassign{j_{n\sminus 1}}{j_{n\sminus 1}+1}
\\[0.2em]
\qquad\qquad \}\\[0.2em]
\qquad\qquad ...\\[0.2em]
\qquad\quad \}\\[0.2em]
\qquad\quad
\cassign{j_{0}}{j_{0}+1}\\[0.2em]
\qquad \}\\[0.2em]
\textcolor{spec}{
\Big{ \{ }
\begin{array}{l}
Q_{\cn{len}}(A_2,A,A_1,(m_0,...,m_{n\sminus 1})) \wedge
\\
\neg (j_0 < m_0)\wedge ... \wedge \neg (j_{n\sminus 1} < m_{n\sminus 1}) \wedge j_0 \in[0,m_0]  \wedge ...\wedge j_{n\sminus 1} \in[0,m_{n\sminus 1}] 
\\
\qquad
\wedge \forall k_0 \in [0,j_0).\,...\forall k_{n\sminus 1} \in[0,j_{n\sminus 1}).\, A\{k_0,...,k_{n\sminus 1} \mapsto E[k_0,...,k_{n\sminus 1}] \}
\end{array}
\Big{ \} }
}\\[0.5em]
\Rightarrow
\textcolor{spec}{
\bil 
\begin{array}{l}
Q_{\cn{len}}(A_2,A,A_1,(m_0,...,m_{n\sminus 1})) \wedge
\\[0.2em]
\forall k_0 \in [0,m_0).\,...\forall k_{n\sminus 1} \in[0,m_{n\sminus 1}).\, A\{k_0,...,k_{n\sminus 1} \mapsto E[k_0,...,k_{n\sminus 1}]\}
\end{array}
\bir
}
&
\textbf{Post-Weakening}
\end{array}
\]
}
}
\end{proof}

\subsection{Translating Hadamard and QFT Operations}\label{sec:hadtrans}

The translations of Hadamard and QFT operations to HLIR are similar; it suffices to demonstrate the Hadamard operation translation. We list the proof rule below.

{\small
      \begin{mathpar}
      \inferrule[]{}
      {\qpule{\{\eta \sqcupplus \eta' : \tch{n}\}}{\eta \sqcupplus \eta' \mapsto q}
        {\ssassign{\eta}{}{\ihadh}}{\eta \sqcupplus \eta' \mapsto \denote{\ihadh}^{\slen{\eta}}q}}
        \end{mathpar}
}

Similar to the oracle function case above, we develop the translation based on a generalized version connecting the pre- and post-conditions $P$, listed below.

{\small
      \begin{mathpar}
      \inferrule[]{}
      {\qpule{\{\eta \sqcupplus \eta' : \tch{n}\}}{P[(\eta \sqcupplus \eta')\{\eta \rightarrow \denote{\ihadh}(\eta)\}/\eta \sqcupplus \eta']}{\ssassign{\eta}{}{\ihadh}}{P}}
        \end{mathpar}
}

We first show the effect of the operation. The operation's \qafny type rule, $\sigma[\eta \sqcupplus \eta' \mapsto \tch{n}]\vdash \ssassign{\eta}{}{\ihadh} \triangleright \sigma[\eta \sqcupplus \eta' \mapsto \tch{n\splus 1}]$, indicates that a Hadamard operation lifts an $n$ dimensional array $A$ to $n\splus 1$ dimensional.
An example is to apply $\ssassign{x\srange{0}{2}}{}{\ihadh}$ to a $\tch{1}$ typed state with the locus $x\srange{0}{3}$. The quantum datum is shown on the left below, with its one-dimensional array view on the right.

{
\[
\textcolor{spec}{
\frac{1}{\sqrt{2}}(\ket{000}+\ket{111})
\qquad\qquad
[(\frac{1}{\sqrt{2}},\ket{000}),(\frac{1}{\sqrt{2}},\ket{111})]
}
\]
}

To construct a proof in Hoare logic, we can split the operation into three steps, (\cn{copyN}, \cn{enumN}, and \cn{ampN}), on the first two bits, corresponding to the locus fragment $x\srange{0}{2}$. 
Below are their transitions.

{\footnotesize
\[
\begin{array}{c@{\;}l}
&
\textcolor{spec}{
[(\frac{1}{\sqrt{2}},\ket{000}),(\frac{1}{\sqrt{2}},\ket{111})]
}
\\[0.4em]
\xrightarrow{\cn{copyN}}
&
\textcolor{spec}{
[[(\frac{1}{\sqrt{2}},\ket{000}),(\frac{1}{\sqrt{2}},\ket{000}),(\frac{1}{\sqrt{2}},\ket{000}),(\frac{1}{\sqrt{2}},\ket{000})],
[(\frac{1}{\sqrt{2}},\ket{111}),(\frac{1}{\sqrt{2}},\ket{111}),(\frac{1}{\sqrt{2}},\ket{111}),(\frac{1}{\sqrt{2}},\ket{111})]]
}
\\[0.4em]
\xrightarrow{\cn{enumN}}
&
\textcolor{spec}{
[[(\frac{1}{\sqrt{2}},\ket{000}),(\frac{1}{\sqrt{2}},\ket{100}),(\frac{1}{\sqrt{2}},\ket{010}),(\frac{1}{\sqrt{2}},\ket{110})],
[(\frac{1}{\sqrt{2}},\ket{001}),(\frac{1}{\sqrt{2}},\ket{101}),(\frac{1}{\sqrt{2}},\ket{011}),(\frac{1}{\sqrt{2}},\ket{111})]]
}
\\[0.6em]
\xrightarrow{\cn{ampN}}
&
\textcolor{spec}{\big{[}}
\begin{array}{l}
\textcolor{spec}{[(\frac{1}{\sqrt{2}}*\frac{(-1)^{(00)\cdot(00)}}{2},\ket{000}),(\frac{1}{\sqrt{2}}*\frac{(-1)^{(00)\cdot(10)}}{2},\ket{100}),(\frac{1}{\sqrt{2}}*\frac{(-1)^{(00)\cdot(01)}}{2},\ket{010}),(\frac{1}{\sqrt{2}}*\frac{(-1)^{(00)\cdot(11)}}{2},\ket{110})],}\\
\textcolor{spec}{[(\frac{1}{\sqrt{2}}*\frac{(-1)^{(11)\cdot(00)}}{2},\ket{001}),(\frac{1}{\sqrt{2}}*\frac{(-1)^{(11)\cdot(10)}}{2},\ket{101}),(\frac{1}{\sqrt{2}}*\frac{(-1)^{(11)\cdot(01)}}{2},\ket{011}),(\frac{1}{\sqrt{2}}*\frac{(-1)^{(11)\cdot(11)}}{2},\ket{111})]}
\end{array}
\textcolor{spec}{\big{]}}
\end{array}
\]
}

The \cn{copyN} step turns each array element into an array of $2^{\slen{x\srange{0}{2}}}$ ($4$) elements by copying the original element, i.e., for the element $(\frac{1}{\sqrt{2}},\ket{000})$, we make an array of $4$ copied elements. This step also turns the array dimension from one (typed $\tch{1}$) to two (typed $\tch{2}$).
For the first two bits of each element vector in the inner array, \cn{enumN} enumerates from $0$ to $2^{\slen{x\srange{0}{2}}}-1$ in binary.
As shown above, we change the first two bits at each basis-vector position to $00$ (integer $0$), $10$ ($1$), $01$ ($2$), and $11$ ($3$) for each inner array.
\cn{ampN} multiplies each amplitude position by two pieces. We first multiply by $\frac{1}{(\sqrt{2})^{\slen{x\srange{0}{2}}}}$, and then by a phase $(-1)^{\cn{old} \cdot \cn{new}}$, where $\cn{old}$ and $\cn{new}$ refer to the first two bits of element vectors before and after applying \cn{enumN}, respectively, and $\cdot$ is the dot product.
QFTs are similar to Hadamard operations, with the addition of phase multiplication.

\myparagraph{Validity Proof.}
Below, we present the valid Hoare triple definition as a lemma for Hadamard operations, followed by the proofs for individual operations.

\begin{lemma}[Valid Hadamard Application]\label[lemma]{def:mult-had}\rm 
Given a quantum state $\varphi$, an operation $\ssassign{\eta}{}{\ihadh}$ in the locus $\eta \sqcupplus \eta'$ with the $\tch{n}$-typed quantum datum $\varphi(\eta\sqcupplus \eta')$, transformed to an $n$-dimensional array $A$, with $\cn{ind}^{\varphi}(\eta\sqcupplus \eta')$ and $\cn{lens}(A)=\overline{m}$,
the result array has $n+1$ dimension, with the inner array bound $2^{\slen{\eta}}$. The valid Hoare logic triple is given as follows.

{\small
\[
\begin{array}{l}
\cn{let }\textcolor{spec}{E} := \textcolor{spec}{[\forall k \in [0,2^{\slen{\eta}}). 
\cn{up}^{a}(\cn{up}^{b}(A[\overline{j}],k,\cn{ind}^{\varphi}(\eta)),(-1)^{k\cdot \cn{vector}(A[\overline{j}])[\cn{ind}^{\varphi}(\eta)]})]}
\cn{ in }
\\[1em]
\textcolor{spec}{\bil 
P[A\{\forall \overline{j}\in \cn{range}(\overline{m}).\overline{j}\mapsto E\}/A]
\wedge Q_{\cn{len}}(A,\overline{m})
\wedge Q_{\cn{len}}(A\{\forall \overline{j}\in \cn{range}(\overline{m}).\overline{j}\mapsto E\},(\overline{m},2^{\slen{\eta}}))
  \bir}\\[0.8em]
\quad
{\cassign{A_1}{\cn{copyN}(A,\overline{m},\cn{ind}^{\varphi}(\eta))};\cassign{A_2}{\cn{enumN}(A_1,\overline{m},\cn{ind}^{\varphi}(\eta))};\cassign{A_3}{\cn{ampN}(A_1,A_2,\overline{m},\cn{ind}^{\varphi}(\eta))}};\cassign{A}{A_3}\\[0.5em]
\textcolor{spec}{\bil P
\wedge Q_{\cn{len}}(A,\overline{m})
\wedge Q_{\cn{len}}(A\{\forall \overline{j}\in \cn{range}(\overline{m}).\overline{j}\mapsto E\},(\overline{m},2^{\slen{\eta}})) \bir}
\end{array}
\]
}
\end{lemma}

\begin{proof}\rm 
The validity proof below comprises the results for \cn{copyN}, \cn{enumN}, and \cn{ampN}. Its mechanized proof is in Rocq.
\end{proof}

We now demonstrate the definitions and proofs for the three key IR operations.
We first show the valid Hoare triple for \cn{copyN}.

\begin{lemma}[The Valid copyN Hoare Triple]\label[lemma]{def:copyN}\rm 
Given a locus $\eta' \sqcupplus \eta$, an $n$ dimensional array $A$, with the bounds $\overline{m}$, and the valid Hoare triple of $\cn{copyN}(A,\overline{m},\cn{ind}^{\varphi}(\eta))$ is listed below with its implementation. We assume that the result array $A_1$ is initialized to have $n+1$ dimensions, with the additional dimension bounded by $2^{\slen{\eta}}$.

{\small
\[\hspace{-0.8em}
\begin{array}{l}
\textcolor{spec}{\bil P[A_1\{\forall \overline{j}\in\cn{range}(\overline{m}). \forall k\in[0,2^{\slen{\eta}}). (\overline{j},k) \mapsto A[\overline{j}]\}/A_1] \wedge Q_{\cn{len}}(A,\overline{m}) \wedge Q_{\cn{len}}(A_1,(\overline{m},2^{\slen{\eta}}))  \bir}\\[0.2em]
\;\;\cassign{A_1}{\cn{copyN}(A,\overline{m},\cn{ind}^{\varphi}(\eta))}
\\[0.2em]
\textcolor{spec}{\bil 
P \wedge Q_{\cn{len}}(A,\overline{m}) \wedge Q_{\cn{len}}(A_1,(\overline{m},2^{\slen{\eta}}))
\bir }
\end{array}
\]
}

where ${\cn{copyN}(A,\overline{m},\cn{ind}^{\varphi}(\eta))}$ is defined as: ${\cn{mapN}(A,\overline{m},\lambda (x,y) . \cn{copy}(2^{\slen{\eta}},(x,y)))}$;
$\cn{copy}(N,v)$ is a function to generate an array of $N$ copies of the element $v$.
\end{lemma}

\begin{proof}\rm
By \Cref{cor:link}, it is enough to show the validity of the following triple.

\vspace*{-0.5em}
{\small
\[\hspace{-0.8em}
\begin{array}{l}
\textcolor{spec}{\bil Q_{\cn{len}}(A,\overline{m}) \wedge Q_{\cn{len}}(A_1,(\overline{m},2^{\slen{\eta}}))\bir}\\[0.2em]
\;\;\cassign{A_1}{\cn{mapN}(A,\overline{m},\lambda (x,y) . \cn{copy}(2^{\slen{\eta}},(x,y)))}
    \\[0.2em]
\textcolor{spec}{\bil 
\forall \overline{j}\in \cn{range}(\overline{m}).\forall k\in[0,2^{\slen{\eta}}).A_1[\overline{j},k] = A[\overline{j}]
  \wedge Q_{\cn{len}}(A_1,\overline{m},2^{\slen{\eta}}) \wedge Q_{\cn{len}}(A,\overline{m}) \wedge Q_{\cn{len}}(A_1,(\overline{m},2^{\slen{\eta}}))
\bir }
\end{array}
\]
}

The proof is given by deriving a proof tree and using a single operation \cn{mapN}. Its details are mechanized in Rocq.
\end{proof}

We now present the proof of \cn{enumN} below, where we construct an $n+1$ dimensional array ($A_2$) based on $A_1$, by replacing the vector in the locus $\eta$ with the enumeration from $0$ to $2^{\slen{\eta}}$.

\begin{lemma}[The Valid EnumN Hoare Triple]\label[lemma]{def:enumN}\rm 
Given a locus $\eta' \sqcupplus \eta$, an $n+1$ dimensional array $A_1$, with the bounds $(\overline{m},2^{\slen{\eta}})$, the valid Hoare triple of $\cn{enumN}(A_1,(\overline{m},2^{\slen{\eta}}),\cn{ind}^{\varphi}(\eta))$ is listed below with its implementation. The result array $A_2$ is initialized as an $n+1$ dimensional array.
$\overline{j},k$ has the first $n$ indices $\overline{j}$ and the most inner index $k$.

{\small
\[\hspace{-0.8em}
\begin{array}{l}
\textcolor{spec}{\bil 
P[A_2\{\forall \overline{j} \in \cn{range}(\overline{m}).\forall k\in[0,2^{\slen{\eta}}). (\overline{j},k)\mapsto \cn{up}^{b}(A_1[\overline{j},k],k, \cn{ind}^{\varphi}(\eta))\} / A_2]
\wedge
Q_{\cn{len}}(A_1,A_2,(\overline{m},2^{\slen{\eta}})) \bir }\\[0.2em]
\;\;
\cn{enumN}(A_1,(\overline{m},2^{\slen{\eta}}),\cn{ind}^{\varphi}(\eta))
\\[0.2em]
\textcolor{spec}{\bil 
P
\wedge Q_{\cn{len}}(A_1,A_2,(\overline{m},2^{\slen{\eta}})) \bir}
\end{array}
\]
}
\end{lemma}

where \cn{enumN} is defined as follows; $\cn{enum}(2^{\slen{\eta}})$ produces an array listing all numbers in $[0,2^{\slen{\eta}})$.

{\small
\[
\begin{array}{l}
\cassign{A_a}{\cn{locateN}(A_1,\overline{m},\cn{ind}^{\varphi}(\eta))};\\
   \quad \cassign{A_b}{\cn{mapN}(A_a,\overline{m},\lambda (x,y) . (x,\cn{enum}(2^{\slen{\eta}})))}; \\
    \qquad\cassign{A_2}{\cn{assignN}(A_1,A_b,\overline{m},\cn{ind}^{\varphi}(\eta),\lambda u\,v.\cn{amp}(u))}
    \end{array}
 \]
}

\begin{proof}\rm 
By \Cref{cor:link}, it is enough to show the validity of the following triple.

{\small
\[\hspace{-0.8em}
\begin{array}{l}
\textcolor{spec}{\bil Q_{\cn{len}}(A_1,A_2,(\overline{m},2^{\slen{\eta}})) \bir }\\[0.2em]
\;\; \cn{enumN}(A_1,(\overline{m},2^{\slen{\eta}}),\cn{ind}^{\varphi}(\eta))\\[0.2em]
\textcolor{spec}{\bil 
 \forall \overline{j} \in \cn{range}(\overline{m}).\forall k\in[0,2^{\slen{\eta}}). A_2[\overline{j},k]=\cn{up}^{b}(A_1[\overline{j},k],k, \cn{ind}^{\varphi}(\eta)) 
\wedge Q_{\cn{len}}(A_1,A_2,(\overline{m},2^{\slen{\eta}})) \bir}
\end{array}
\]
}

The proof-tree derivation is given in \Cref{appx:enumN}.
\end{proof}

We now demonstrate the validity of the Hoare triple for \cn{ampN}, and it suffices to show the $n$-dimensional case. Here, we generate an $n$-dimensional array $A_3$, having the same basis-vectors as $A_1$, but multiplying its \cn{amp} elements by a phase relating the \cn{amp} elements of $A_1$ and $A_2$.

\begin{lemma}[The ampN Proof]\label[lemma]{def:ampN}\rm 
Given a locus $\eta' \sqcupplus \eta$, two $n$ dimensional arrays $A_1$ and $A_2$, with the bounds of each dimension to be $\overline{m}$, and the valid Hoare triple of the $\cn{ampN}(A_1,A_2,\overline{m},\cn{ind}^{\varphi}(\eta))$ function is listed below. The result array $A_3$ is initialized to be $ n$-dimensional with bounds $\overline{m}$.

{\small
\[\hspace*{-1em}
\begin{array}{l}
\textcolor{spec}{\bil 
\begin{array}{l}
P[A_3\{\forall \overline{j}\in \cn{range}(\overline{m}). \overline{j}\mapsto \cn{up}^{a}(\cn{up}^b(A_1[\overline{j}],\cn{vector}(A_2[\overline{j}])[\cn{ind}^{\varphi}(\eta)],\cn{ind}^{\varphi}(\eta)), \\[0.2em]
\qquad
(-1)^{\cn{vector}(A_1[\overline{j}])[\cn{ind}^{\varphi}(\eta)]\cdot \cn{vector}(A_2[\overline{j}])[\cn{ind}^{\varphi}(\eta)]}) \} / A_3]
\wedge Q_{\cn{len}}(A_1,A_2,A_3,(\overline{m},2^{\slen{\eta}}))
  \end{array}
  \bir}\\[0.8em]
\cassign{A_3}{\cn{ampN}(A_1,A_2,\overline{m},\cn{ind}^{\varphi}(\eta))}
  \\[0.2em]
\textcolor{spec}{\bil P \wedge Q_{\cn{len}}(A_1,A_2,A_3,(\overline{m},2^{\slen{\eta}}))  \bir}\\[0.2em]
\end{array}
\]
}

where $\cn{ampN}(A_1,A_2,\overline{m},\cn{ind}^{\varphi}(\eta))$ is defined as:
{\small
\[{\cn{assignN}(A_1,A_2,\overline{m},\cn{ind}^{\varphi}(\eta),\lambda u\,v.\cn{amp}(v)*(-1)^{\cn{vector}(u)[\cn{ind}^{\varphi}(\eta)]\cdot\cn{vector}(v)[\cn{ind}^{\varphi}(\eta)]})}\]
}
\end{lemma}

\begin{proof}\rm 
The proof is given as deriving a proof tree and building on top of \cn{assignN}. The detail is mechanized in Rocq.
\end{proof}

\subsection{Translating Quantum Conditionals}\label{sec:condtrans} 

We now present the translation of a quantum conditional, whose proof and type rules in \qafny are given below and contain two predicate transformers, $\frz$ and $\ufrz$.

{\small
      \begin{mathpar}
      \inferrule[P-If]{ FV(b) = \eta  \\ \qpule{\{\eta' : \tch{n}\}}{\varphi[\frz(b,\eta,\eta')/ \eta \sqcupplus \eta']}{\mu}{\varphi'}}
      {\qpule{\{\eta \sqcupplus \eta' : \tch{n}\}}{\varphi}{\sifq{b}{\mu}}{\varphi[\ufrz(\neg b,\eta,\eta \sqcupplus \eta')/\eta \sqcupplus \eta'] * \varphi'[\ufrz(b,\eta,\eta \sqcupplus \eta')/\eta']}}
           \end{mathpar}
                 \begin{mathpar}
              \inferrule[T-If]{ 
          FV(b)=\eta \\
          FV(e) \subseteq \eta'
          \\
          \sigma\uplus\{\eta' : \tch{n}\} \vdash e \triangleright \sigma\uplus\{\eta' : \tch{n'}\}}
        {\sigma\uplus\{\eta \sqcupplus \eta' : \tch{n}\} \vdash \sifq{b}{e} \triangleright \sigma\uplus\{\eta \sqcupplus \eta' : \tch{n'}\} }
     \end{mathpar}
}

In Rule \rulelab{P-If}, the bottom layer selects basis-kets $w$, where $\denote{b}^{\eta}(w)$ is \cn{true} and trims the basis-ket states to the locus region $\eta'$.
With the shrunk basis-ket states, the top level applies the $\mu$ operation.
We first demonstrate the \qafny proof via the Bell pair example.

\begin{example}[The Quantum Conditional Example in Bell Pair]\label[example]{example:bellpair}\rm
The Bell pair circuit (\Cref{sec:background}) applies a Hadamard to qubit $x[0]$ followed by a controlled-not on $y[0]$ guarded by $x[0]$, and $e=\qass{y\srange{0}{1}}{y\srange{0}{1}\,\splus\,1}$.
With pre- and post-conditions, the \qafny triple is:

{\small
      \begin{mathpar}
\inferrule[]{
      {\qpule{\sigma[y\srange{0}{1} \mapsto \tch{1}]}{y\srange{0}{1} {:}\,\tch{1} \mapsto \frac{1}{\sqrt{2}} \textcolor{auto}{\ket{1}} \ket{0}}{e}
             {y\srange{0}{1}{:}\,\tch{1} \mapsto \frac{1}{\sqrt{2}} \textcolor{auto}{\ket{1}} \ket{1} }}}
      {\qpule{\sigma}{x\srange{0}{1}\sqcupplus y\srange{0}{1} {:}\,\tch{1} \mapsto \sum_{b=0}^1{\frac{1}{\sqrt{2}}\ket{b}\ket{0}}}
      {\qif{x[0]}{e}}{x\srange{0}{1}\sqcupplus y\srange{0}{1}{:}\,\tch{1} \mapsto \sum_{b=0}^1{\frac{1}{\sqrt{2}}\ket{b}\ket{b}}}}
     \end{mathpar}
}

The pre-condition is an $\tch{1}$ typed entangled state for $x\srange{0}{1}\sqcupplus y\srange{0}{1}$.
The upper-level pre-condition checks $x\srange{0}{1}$'s basis-ket fragment $\ket{b}$, masks the $\textcolor{spec}{\frac{1}{\sqrt{2}}\ket{0}\ket{0}}$ case, and leaves the case of $\textcolor{spec}{\frac{1}{\sqrt{2}}\ket{1}\ket{0}}$, where we hide $x\srange{0}{1}$'s basis-vector $\textcolor{auto}{\ket{1}}$ since it is not mentioned in the locus domain.
Finally, we assemble the upper level post-condition with $\textcolor{auto}{\ket{1}}$ and $\textcolor{spec}{\frac{1}{\sqrt{2}}\ket{1}\ket{1}}$. 

\end{example}

{\small
      \begin{mathpar}
      \inferrule[]{ FV(b) = \eta  \\ \qpule{\{\eta' : \tch{n}\}}{\varphi[\frz(b,\eta,\eta')/ \eta \sqcupplus \eta']}{\mu}{\varphi_1[\frz(b,\eta,\eta')/ \eta \sqcupplus \eta']}}
      {\qpule{\{\eta \sqcupplus \eta' : \tch{n}\}}{\varphi}{\sifq{b}{\mu}}{\varphi_1}}
     \end{mathpar}
}

The key to embedding the proof rule is to connect the upper- and bottom-level (pre- and post-) conditions.
To do so, we can think of two program points ${\bullet}$ and ${\circ}$ around the translated expression $e$, where we prepare inserted predicates before ${\bullet}$ to transform $P$ to the right form of the upper-level pre-condition and transform ${\circ}$, representing the upper-level post-condition, to the bottom-level post-condition $Q$.
We manage the translation as follows.

\begin{lemma}[The Quantum Conditional Transformation]\label[lemma]{proof:qif}\rm 
Given $\sifq{b}{\mu}$, where we do expression translation for $\mu$,  $\{ \eta' : \tch{n} \}\vdash \mu \triangleright \{ \eta' : \tch{n'} \} \gg e$ and pre- and post-condition translations, as $\{\eta \sqcupplus \eta' : \tch{n} \} \vdash \varphi \gg P$ and $\{\eta \sqcupplus \eta' : \tch{n'} \} \vdash \varphi' \gg Q$; without losing generality, we assume the locus $\eta \sqcupplus \eta'$ is translated to an $n$ dimensional array $A$, such that $\eta \sqcupplus \eta' \vdash A$, we express the valid Hoare logic triple as follows.
$\overline{t}$ and $(\overline{t},\overline{1})$ represent the lengths for an $n$ and $n'$ ($>n$) dimensional arrays, respectively, where $(\overline{t},\overline{1})$ extends 
$\overline{t}$ with additional $n'-n$ dimensions with length $1$ in each dimension; $\overline{u'}+(\overline{t},\overline{1})$ adds the lengths in each dimensions.

{\small
\[\hspace*{-1em}
\begin{array}{l}
\textcolor{spec}{\bil P \wedge Q_{\cn{len}}(A,\overline{m}) \bir}\\[0.8em]
\;\;
\cassign{A_1}{\cn{filterN}(A,b,\cn{ind}^{\varphi}(\eta))}; 
\\[0.1em]
\quad
\cassign{\overline{u}}{\slen{A_1}};
\\[0.2em]
\;\;\quad
\cassign{A_2}{\cn{filterN}(A,\neg b,\cn{ind}^{\varphi}(\eta))};
\\[0.1em]
\qquad
\cassign{\overline{t}}{\slen{A_2}};
\\[0.2em]
\;\;\qquad
\cassign{A_3}{\cn{locateN}(A_1,\overline{u},\cn{ind}^{\varphi}(\eta'))};
\\[0.2em]
\qquad\qquad\bullet\;
\textcolor{spec}{\bil
\begin{array}{l}
P[\cn{concat}(\cn{upDim}(A_2,\overline{t},(\overline{t},\overline{1})),
\cn{upDim}(A_1,\overline{u},\slen{A_3})\{\forall \overline{j}\in \cn{range}(\overline{\slen{A_3}}). \overline{j}\mapsto \cn{up}^a(\cn{up}^b(
\\[0.1em]
\cn{upDim}(A_1,\overline{u},\slen{A_3})[\overline{j}],\cn{vector}(A_3[\overline{j}]),\cn{ind}^{\varphi}(\eta')),\cn{amp}(A_3[\overline{j}])) \})/A]
\end{array}
\bir}
\\[0.2em]
\qquad\qquad  e ; 
\\[0.2em]
\qquad\qquad\circ\;
\textcolor{spec}{\bil
\begin{array}{l}
Q[\cn{concat}(\cn{upDim}(A_2,\overline{t},(\overline{t},\overline{1})),
\cn{upDim}(A_1,\overline{u},\slen{A_3})\{\forall \overline{j}\in \cn{range}(\overline{\slen{A_3}}). \overline{j}\mapsto \cn{up}^a(\cn{up}^b(
\\[0.1em]
\cn{upDim}(A_1,\overline{u},\slen{A_3})[\overline{j}],\cn{vector}(A_3[\overline{j}]),\cn{ind}^{\varphi}(\eta')),\cn{amp}(A_3[\overline{j}])) \})/A]
\end{array}
\bir}
\\[0.2em]
\;\;\qquad
 \cassign{\overline{u'}}{\slen{A_3}};
 \\[0.1em]
 \qquad
\cassign{A_4}{\cn{upDim}(A_1,\overline{u},\overline{u'})};
\\[0.2em]
\;\;\quad
\cassign{A_5}{\cn{upDim}(A_2,\overline{t},(\overline{t},\overline{1}))};
\\[0.1em]
\quad
\cassign{A_6}{\cn{assignN}(A_4,A_3,\overline{u'},\cn{ind}^{\varphi}(\eta'),\lambda v\,l.\cn{amp}(l))};
\\[0.2em]
\;\;
\cassign{A_7}{\cn{concat}(A_5,A_6)};
\\[0.1em]
\;
\cassign{A}{A_7}
\\[0.5em]
\textcolor{spec}{\bil Q \wedge Q_{\cn{len}}(A,\overline{u'}+(\overline{t},\overline{1})) \bir}
\end{array}
\]
}
\end{lemma}

In the lemma, \cn{filterN} collects all the elements in $A$ satisfying $b$ on the vector-bits corresponding to $\eta$, $\cn{upDim}(A,\overline{u},\overline{v})$ updates $A$'s dimension from $\overline{u}$ to $\overline{v}$, and $\cn{concat}(A_5,A_6)$ creates a new array $A_7$ by concatenating $A_6$ to the end of $A_5$.

As an example of compiling \Cref{example:bellpair}, we name all the compiled operations in \Cref{proof:qif} as $\cn{if\_h}(e)$, where the \qafny triple in \Cref{example:bellpair} is translated to the Hoare triple below.

{\small
      \begin{mathpar}
      \inferrule[]{\triple{A_3 = [(\tfrac{1}{\sqrt{2}},\ket{0})]}
       {e}
       {A_3 = [(\tfrac{1}{\sqrt{2}},\ket{1})]}}
      {\triple{A = [(\tfrac{1}{\sqrt{2}},\ket{00}),\;(\tfrac{1}{\sqrt{2}},\ket{10})]}
       {\cn{if\_h}(e)}
       {A = [(\tfrac{1}{\sqrt{2}},\ket{00}),\;(\tfrac{1}{\sqrt{2}},\ket{11})]}}
     \end{mathpar}
}

We examine the step-by-step effect of the operation below.
By applying the \cn{filterN} functions to $A$ for Boolean expressions $b$ and $\neg b$, we assign $A_1$ and $A_2$, with both lengths being $1$ ($\overline{u}$ for $A_1$ and $\overline{t}$ for $A_2$). We list the post-condition after applying the first four operations below.

{
\[
\textcolor{spec}{
\{
A_1 = [(\tfrac{1}{\sqrt{2}},\ket{10})]
\wedge
\overline{u}=1
\wedge
A_2 = [(\tfrac{1}{\sqrt{2}},\ket{00})]
\wedge
\overline{t}=1
\}
}
\]
}

We then apply \cn{locateN} to $A_1$ and assign $A_3$ to be $\textcolor{spec}{[(\tfrac{1}{\sqrt{2}},\ket{0})]}$.
By applying pre-condition strengthening, we reach the pre-condition $\textcolor{spec}{\{A_3 = [(\tfrac{1}{\sqrt{2}},\ket{0})]\}}$ on the top level above.

In \Cref{example:bellpair}, the upper level $e$ is a translation of $\mu$ given by the oracle operation compilation in \Cref{def:mult-mut}. We list the translated version below.

{
\[
\begin{array}{l}
\cassign{A_a}{\cn{locateN}(A_3,\overline{u},\cn{ind}^{\varphi}(\eta'))}; \\
\quad
\cassign{A_b}{\cn{mapN}(A_a,\overline{u},\lambda (x,y).(x,(y+1)\% 2))};
\\[0.1em]
\qquad
 \cassign{A_c}{\cn{assignN}(A_3,A_b,\overline{u},\cn{ind}^{\varphi}(\eta),\lambda v\,l.\cn{amp}(l))};\\
 \qquad\quad\cassign{A_3}{A_c}
 \end{array}
 \]
}

After applying the expression $e$, we reach the upper level post-condition  $\textcolor{spec}{\{A_3 = [(\tfrac{1}{\sqrt{2}},\ket{1})]\}}$, with the length $\overline{u'}=1$.
Since $A_3$'s dimension is not modified in the upper-level computation, i.e., $A_3$'s dimension is the same as $A_1$'s, the first \cn{upDim} operation sets $A_4$ to be the same as $A_1$. The same situation applies to the second \cn{upDim}. We list below the post-condition, including only $A_3$, $A_4$, and $A_5$ after the two \cn{upDim} operations.

{
\[
\textcolor{spec}{
\{
A_3 = [(\tfrac{1}{\sqrt{2}},\ket{1})]
\wedge
A_4 = [(\tfrac{1}{\sqrt{2}},\ket{10})]
\wedge
A_5 = [(\tfrac{1}{\sqrt{2}},\ket{00})]
\}
}
\]
}

The \cn{assignN} operation assigns to $A_6$ with $A_3$'s amplitude as well as the basis-vector bit shrunk to the second position ($y[0]$'s position),
where $A_6$ is assigned as $\textcolor{spec}{[(\tfrac{1}{\sqrt{2}},\ket{11})]}$.
Then, we concatenate $A_5$ and $A_6$ and assign $A_7$ back to $A$.
Finally, we reach the bottom-level post-condition listed in the Hoare triple above.

\myparagraph{Validity Proof.}
The validity proof is based on the \qafny proof rule for quantum conditionals, with generalized pre- and post-conditions, where we set the bottom-level pre- and post-conditions to be arbitrary $P$ and $Q$ in \rulelab{P-If}.
In \Cref{proof:qif}, the operations before $\bullet$ transform $P$ to the pre-condition of the upper level (\rulelab{P-If}) and the operations after $\circ$ transform the post-condition in the upper level to $Q$, i.e., the pre- and post-conditions around $e$ can be thought of as $\triple{\bullet}{e}{\circ}$, where $\bullet$ and $\circ$ represent the translated upper-level pre- and post-conditions. We outline the proof below.

\begin{proof}\rm 
The proof requires deriving a Hoare logic proof tree, including the array-length axioms for $A$ and $A_1$ to $A_7$.
We assume that $(\overline{t},\overline{1})$ defines the length for $i+j$ dimensional array, where $\overline{t}$ defines the lengths for $i$ dimensions and the remaining $j$ dimensions have length $1$.
The proof relies on the axiom that $\cn{concat}(\cn{filterN}(A,b),\cn{filterN}(A,\neg b))$ gives back $A$, i.e., if we partition $A$ into two subsets and merge them back, it will produce the same $A$. This is valid in quantum states since they do not care about ordering.
The proof-tree derivation is given in \Cref{appx:qif}.
\end{proof}

We now present the proofs for some of the functions involved above.
The filter (\cn{filterN}), array concatenation (\cn{concat}), and array length ($\slen{A}$) are well known, as many previous works \cite{10.1007/s10817-022-09654-y} discuss their Hoare logic proofs. Their mechanized proof is in Rocq and has a similar style to \Cref{def:locate}.
We mainly focus on $\cn{upDim}(A,\overline{m},\overline{u})$, which updates $A$'s dimension from $\overline{m}$ to $\overline{u}$, provided that $\overline{m}\sqsubseteq \overline{u}$, i.e., $\overline{m}$ is a prefix of $\overline{u}$ as $\exists \overline{m'}. \overline{u} = (\overline{m},\overline{m'})$.

\begin{lemma}[The Valid UpDim Hoare Triple]\label[lemma]{def:updim}\rm 
Given a locus $\eta' \sqcupplus \eta$, an $n$ dimensional array $A$, with the bounds $\overline{m}$ and a length list $\overline{m'}$, and the valid Hoare triple of $\cn{upDim}(A,\overline{m},(\overline{m},\overline{m'}))$ is listed below with its implementation. We assume that the result array $A_1$ is initialized to $n+\slen{\overline{m'}}$ dimension with the additional dimension bounds $\overline{m'}$.

{\small
\[\hspace{-0.8em}
\begin{array}{l}
\textcolor{spec}{\bil P[A_1\{\forall \overline{j}\in\cn{range}(\overline{m}). \forall \overline{k}\in\cn{range}(\overline{m'}). (\overline{j},\overline{k}) \mapsto A[\overline{j}]\}/A_1] \wedge Q_{\cn{len}}(A,\overline{m}) \wedge Q_{\cn{len}}(A_1,(\overline{m},\overline{m'}))  \bir}\\[0.2em]
\;\;\cassign{A_1}{\cn{upDim}(A,\overline{m},(\overline{m},\overline{m'}))}
\\[0.2em]
\textcolor{spec}{\bil 
P \wedge Q_{\cn{len}}(A,\overline{m}) \wedge Q_{\cn{len}}(A_1,(\overline{m},\overline{m'}))
\bir }
\end{array}
\]
}

where ${\cn{upDim}(A,\overline{m},(\overline{m},\overline{m'}))}$ is defined as: ${\cn{mapN}(A,\overline{m},\lambda (x,y) . \cn{copy}(\overline{m'},(x,y)))}$;
$\cn{copy}(\overline{m'},v)$ is a generalized version of the copy function in \Cref{def:copyN} to produce a $\overline{m'}$ dimensional array of elements $v$.
\end{lemma}

\begin{proof}\rm
The proof is given by deriving a proof tree using a single operation \cn{mapN}, and it follows a similar style to the proof of \Cref{def:copyN}. Its detail is mechanized in Rocq.
\end{proof}

\subsection{Translating Quantum Measurement}\label{sec:meatrans}

Below is the \qafny proof rule for quantum measurements, containing a predicate transformer $\mea$.

{
      \begin{mathpar}
    \inferrule[]{}
     {\qpule{\sigma\uplus\{\eta\sqcupplus\eta': \tch{n}\}}{P[\mea(x,r,\slen{\eta},\eta')/\eta\sqcupplus\eta']}{\cassign{r,x}{\smea{\eta}} }{P} }
     \end{mathpar}
}

The predicate transformer $\mea$ performs two tasks: 1) trim down the locus $\eta \sqcupplus \eta'$ to be $\eta'$ since $\eta$ is measured away, and 2) assign the measurement result, both measurement output bits and probability value, to the binding variables $x$ and $r$, respectively.
The operation is similar to a series of array assignments, as in the quantum oracle operation, except that the measurement result is nondeterministic. If we view a quantum state as an array representing a linear sum, we can always assume the existence of the index $\overline{i}$ to be picked; thus, $x=A_1[\overline{i}]$, if $A_1$ refers to the sub-array of $\eta$, and the probability value $r$ is the sum of the squares of all amplitudes in $A_1[\overline{k}]$ (for all $\overline{k}$), such that $A_1[\overline{i}] = A_1[\overline{k}]$. We define the following valid Hoare triple for quantum measurement.

\begin{lemma}[The Measurement Valid Triple]\label[lemma]{def:mea}\rm 
Given a quantum state $\varphi$, an operation $\cassign{r,x}{\smea{\eta}}$ in the locus $\eta \sqcupplus \eta'$ with the $\tch{n}$-typed quantum datum $\varphi(\eta\sqcupplus \eta')$, transformed to an $n$ dimensional array $A$, where the collection of indices $\cn{ind}^{\varphi}(\eta\sqcupplus \eta')$.
Assuming that the index $\overline{i}\in[0,\overline{m})$ is picked for $A$, the valid Hoare triple is listed as follows, where $\textcolor{spec}{E_1} := 
\textcolor{spec}{\cn{filterN}(A,\lambda y.y = x,\cn{ind}^{\varphi}(\eta))}$.

{\small
\[
\begin{array}{l}
    \textcolor{spec}{
\Big{ \{ }
\begin{array}{l} 
P[E_1\{\forall \overline{j}\in\cn{range}(\slen{E_1}).\overline{j}\mapsto \cn{up}^a(\cn{up}^b(E_1[\overline{j}],\cn{vector}(E_1[\overline{j}]),\cn{ind}^{\varphi}(\eta'))
,\frac{\cn{amp}(E_1)}{r})\}/A]
\\[0.2em]\qquad[\sqrt{\sum_{\overline{j}\in\cn{range}(\slen{E_1})}\slen{\cn{amp}(E_1)}^2}/r]
[\cn{vector}(A[\overline{i}])[\cn{ind}^{\varphi}(\eta)]/x]
\wedge Q_{\cn{len}}(A,\overline{m})
\end{array}
\Big{ \} } }\\[0.2em]
\qquad 
\begin{array}{l}
\;\;\cassign{x}{\cn{vector}(A[\overline{i}])[\cn{ind}^{\varphi}(\eta)]};
\\[0.2em]
\quad\cassign{A_1}{\cn{filterN}(A,\lambda y.y = x,\cn{ind}^{\varphi}(\eta))};
\\[0.2em]
\quad\;\;
\cassign{\overline{t}}{\slen{A_1}};
\\[0.2em]
\qquad
      \cassign{A_2}{\cn{locateN}(A_1,\overline{t},\cn{ind}^{\varphi}(\eta'))};
\\[0.2em]
\qquad\;\;
      \cassign{r}{\cn{sqrtSum}(A_2,\overline{t})};
 \\[0.2em]
\qquad\quad      \cassign{A_3}{\cn{mapN}(A_2,\overline{t}, \lambda (y,u) . (\frac{y}{r},u))}; \cassign{A}{A_3}
 \end{array}
       \\[0.2em]
\textcolor{spec}{\bil P \wedge Q_{\cn{len}}(A,\overline{t}) \wedge \overline{t} \le \overline{m} \bir }
\end{array}
\]
}
\end{lemma}

Following the Bell pair example in \Cref{example:bellpair}, assume that we are given the task of verifying whether $0$ and $\ket{0}$ are the valid outputs after measuring the first qubit, and the state before the measurement is represented as an array $A=\textcolor{spec}{[\frac{1}{\sqrt{2}}\ket{0}\ket{0},\frac{1}{\sqrt{2}}\ket{1}\ket{1}]}$.
When applying the proof sequence in \Cref{def:mea} to the example, we first pick $0$ as the output index, as $A[0]$, where we assign $x$ to $\cn{vector}(A[0],\cn{ind}^{\varphi}(x[0]))=0$, and we assign $A_1$ and $A_2$ to be $\textcolor{spec}{[\frac{1}{\sqrt{2}}\ket{0}\ket{0}]}$ and $\textcolor{spec}{[\frac{1}{\sqrt{2}}\ket{0}]}$ with $\overline{t}=1$.
Finally, $r$ is assigned to $\frac{1}{\sqrt{2}}$ where we assign $A_3$ to be $\textcolor{spec}{[\ket{0}]}$ by dividing the element amplitude by $r$.

Measurement need not be the last step of a program: $\cassign{r,x}{\smea{\eta}}$ binds the outcome $x$ and its probability $r$ for use in the continuation. Because the outcome index is chosen once and the remaining array is a function of the pre-state and that choice, the continuation sees $x$ as an ordinary classical variable; feed-forward control on measurement outcomes, as in teleportation or repeat-until-success loops, therefore translates to ordinary classical control in Hoare logic, and no black-box probability distribution is involved.

\myparagraph{Validity Proof.}
We show the proof of \Cref{def:mea} below.

\begin{proof}\rm 
The proof derives a Hoare logic proof tree by combining the proof results for the \cn{filterN}, \cn{locateN}, \cn{sqrtSum}, and other functions.
For simplicity, we omit the $Q_{\cn{len}}$ judgments.
The proof-tree derivation is given in \Cref{appx:mea}.
\end{proof}

For the functions involved in the proof, aside from \cn{sqrtSum}, the others have been discussed above.
We now present the proof of the \cn{sqrtSum} function below, which is essentially the same as the proof for a classical sum function.

\begin{lemma}[The SqrtSum Hoare Triple]\label[lemma]{def:sqrtsum}\rm
Given a locus $\eta$, an $n$-dimensional array $A$ with bounds $\overline{m}$, the valid Hoare triple for the $\cn{sqrtSum}(A,\overline{m})$ function is listed below. 

{\small
\[\hspace{-0.8em}
\begin{array}{l}
\textcolor{spec}{\bil P[\sqrt{\sum_{\overline{j}\in\cn{range}(\overline{m})}\slen{\cn{amp}(A[\overline{j}])}^2}/r] \wedge Q_{\cn{len}}(A,\overline{m})  \bir}\\[0.2em]
\;\;\cassign{r}{\cn{sqrtSum}(A,\overline{m})}
\\[0.2em]
\textcolor{spec}{\bil 
P \bir }
\end{array}
\]
}

where ${\cn{sqrtSum}(A,\overline{m})}$ is defined as:

{\small
\[
\begin{array}{l}
\cassign{r}{0};\\
\cassign{\overline{j}}{0};\\
\WHILE(j_{0} < m_{0})\{ ... \\
\qquad\WHILE(j_{n\sminus 1} < m_{n\sminus 1}) \{ 
\\
\qquad\qquad{r}\,\cn{+=}\,{\slen{\cn{amp}(A[k_0,...,k_{n\sminus 1}])}^2};  \cassign{j_{n\sminus 1}}{j_{n\sminus 1}+1} 
\\
\qquad\}
\\
...\; \cassign{j_{0}}{j_{0}+1}\}
\\
\cassign{r}{\sqrt{r}}
\end{array}
\]
}
\end{lemma}

\begin{proof}\rm
The proof is given by deriving a proof tree with a structure similar to the proofs of the classical sum function and of \Cref{def:locate}.
\end{proof}

\subsection{Translating Classical Operations}\label{sec:transclassical}

The \qafny proof system contains some proof rules for classical operations, where we list two below for sequencing and for-loop operations.

{\small
\begin{mathpar}
    \inferrule[P-Seq]{\sigma\vdash \mu_1 \triangleright\sigma_1\\\\\qpule{\sigma}{\varphi}{\mu_1 }{\varphi'} \;\quad \qpule{\sigma_1}{\varphi'}{\mu_2 }{\varphi''}}
     {\qpule{\sigma}{\varphi}{\sseq{\mu_1}{\mu_2} }{\varphi''} }
     
    \inferrule[P-Loop]{\sigma[a_1/j]\vdash \cfor{j}{a_1}{a_2}{\{e\}} \triangleright\sigma[a_2/j] \\\\ \qpule{\sigma}{\varphi \wedge a_1 \le j < a_2}{e}{\varphi} }
     {\qpule{\sigma[a_1/j]}{\varphi}{ \cfor{j}{a_1}{a_2}{\{e\}} }{\varphi \wedge j = a_2} }
\end{mathpar}
}

Except for the \qafny typing, the proof rules are similar to the sequencing and for-loop rules in Hoare logic. These rule transformations are natural, i.e., we transform the sub-terms and simply transform these operations to the proof rules for similar operations in Hoare logic.
Sequencing is compositional: the translation of $\sseq{\mu_1}{\mu_2}$ is the sequence of the two translations, and \rulelab{P-Seq} becomes the Hoare sequencing rule. A for-loop is a classical loop whose body is translated as above, so \rulelab{P-Loop} becomes the Hoare loop rule with the translated invariant. The only non-trivial step is the invariant itself: its locus typically grows with the loop index, as $x\srange{0}{j}$ does in a loop that entangles one more qubit per iteration, so the invariant is translated with the index map $\cn{ind}^{\varphi}(x\srange{0}{j})$ and the array bounds $\overline{m}$ parameterised by $j$. The translated invariant $P(j)$ is then an ordinary Hoare invariant, and its preservation over one iteration is the translation of the loop body from $P(j)$ to $P(j\splus 1)$.
 \subsection{Translation Completeness and Soundness}
\label{sec:proofs}

This section presents the two translation theorems; translation completeness (\Cref{thm:compile-soundness}) and its component lemmas are mechanized in Rocq, while translation soundness (\Cref{thm:relative-completeness}) is proved on paper.
Based on the individual component proofs in \Cref{sec:trans}, we express the combined theorem below.

\begin{theorem}[Translation Completeness]\label{thm:compile-soundness}\rm 
Given a \qafny proof $\qpule{\sigma}{\varphi}{\mu}{\varphi'}$, let $\sigma \vdash \mu \triangleright \sigma' \gg e$, $\sigma \vdash \varphi \gg P$, and $\sigma' \vdash \varphi' \gg Q$; thus, $\triple{P}{e}{Q}$.
\end{theorem}

\begin{proof}\rm 
The proof is based on induction on the \qafny type system and a combination of the proofs for the cases in \Cref{def:mult-mut,def:mult-had,proof:qif,def:mea}.
The proof of the for-loop and a Hoare logic loop are the same except that we need to translate the loop invariant, since the for-loop operation is classical in \qafny.
\end{proof}

The above theorem is the completeness direction: as we translate the program $\mu$, pre-condition $\varphi$, and post-condition $\varphi'$ to a Hoare logic program $e$ and conditions $P$ and $Q$, the triple $\triple{P}{e}{Q}$ is valid in Hoare logic, so nothing \qafny can prove is lost in translation.
Since the translation is guided by the \qafny type system $\sigma \vdash \mu \triangleright \sigma'$, the converse direction, translation soundness, must be shown separately: when a translated $\triple{P}{e}{Q}$ is valid, \qafny derives a triple for $\mu$ with semantically equivalent conditions, so the classical verifier cannot certify a false quantum claim. This is the direction that makes a classical verdict about a quantum program trustworthy.
In the following theorem, $\denote{\varphi}$ is the semantic interpretation of the quantum state $\varphi$.

\begin{theorem}[Translation Soundness]\label{thm:relative-completeness}\rm
Given a \qafny program $\mu$, with the typed translation $\sigma \vdash \mu \triangleright \sigma' \gg e$ where $\sigma \vdash \mu \triangleright \sigma'$, two conditions $\varphi$ and $\varphi'$, such that $\sigma \vdash \varphi \gg P$ and $\sigma' \vdash \varphi' \gg Q$, if $\triple{P}{e}{Q}$, there exists $\varphi_1$ and $\varphi'_1$, such that $\qpule{\sigma}{\varphi_1}{\mu}{\varphi'_1}$, $\denote{\varphi}\equiv\denote{\varphi_1}$, and $\denote{\varphi'}\equiv\denote{\varphi'_1}$.
\end{theorem}

\begin{proof}\rm 
The proof is based on induction on the \qafny typing rules $\sigma \vdash \mu \triangleright \sigma'$.
Note that the expression $e$ is translated from $\mu$ via fixed translation functions.
If the three entities can construct a valid Hoare triple, $\triple{P}{e}{Q}$, we can find conditions $\varphi_1$ and $\varphi'_1$ to construct a valid \qafny triple $\qpule{\sigma}{\varphi_1}{\mu}{\varphi'_1}$. When we translate $\mu$ to $e$, we translate an individual quantum operation to a composition of HLIR operations.
If we can construct a Hoare triple for the HLIR operations, we can also construct a \qafny triple for $\mu$.

Now, we need to construct the relations between $\varphi_1$ and $\varphi$ and between $\varphi'_1$ and $\varphi'$. When we translate $\varphi$/$\varphi'$ to $P$/$Q$, we restrict their types to $\sigma$/$\sigma'$; these types are based on the locus structures. Given $\triple{P}{e}{Q}$, we can show that the types for $\varphi_1$ and $\varphi'_1$ are also $\sigma$ and $\sigma'$, respectively. With the same program $\mu$, we are able to show that $\varphi_1$/$\varphi'_1$ and $\varphi$/$\varphi'$ are semantically equivalent, as $\denote{\varphi}\equiv\denote{\varphi_1}$ and $\denote{\varphi'}\equiv\denote{\varphi'_1}$.
\end{proof}

Finally, we bound the size of the translation. Sizes are measured in the symbolic array representation used throughout \Cref{sec:trans}: a dimension of length $2^m$ is a single bound expression and assertions quantify over index vectors, so a Hadamard operation on $m$ qubits adds one dimension (\cn{copyN}, \cn{enumN}, \cn{ampN}) rather than $2^m$ elements.

\begin{theorem}[Translation Size]\label{thm:verification-efficiency}\rm
Let $\sigma \vdash \mu \triangleright \sigma' \gg e$, $\sigma \vdash \varphi \gg P$, and $\sigma' \vdash \varphi' \gg Q$, and let $n$ be the largest array dimension (the largest $k$ in an $\tch{k}$ type) occurring in $\sigma$ and $\sigma'$. Then
(i) $e$ contains at most eleven HLIR statements per \qafny statement of $\mu$, so $\slen{e}$ is linear in $\slen{\mu}$;
(ii) $\slen{P}$ and $\slen{Q}$ are polynomial in $\slen{\varphi}+n$ and $\slen{\varphi'}+n$, respectively;
(iii) for each HLIR statement, the Hoare proof-tree skeleton of \Cref{sec:trans} and \Cref{appx:prooftrees} has size polynomial in the size of its pre- and post-conditions and in $n$.
\end{theorem}

The theorem bounds the size of the generated verification conditions; it says nothing about the time a classical verifier needs to discharge them.

\begin{proof}\rm
By cases on the translation of each \qafny statement.
For (i), an oracle application translates to four statements (\Cref{def:mult-mut}: \cn{locateN}, \cn{mapN}, \cn{assignN}, and an assignment), a Hadamard or QFT application to four (\Cref{def:mult-had}), a quantum conditional to eleven plus the translation of its body (\Cref{proof:qif}), and a measurement to seven (\Cref{def:mea}); sequencing and for-loops add no statements beyond those of their sub-statements (\Cref{sec:transclassical}).
For (ii), the translation of a state predicate maps each locus $\eta : \tch{k} \mapsto q$ to one $k$-dimensional array, a quantifier over its index vector, and the basis-ket expression of $q$, whose size is linear in $\slen{q}$ plus $k \le n$; each predicate transformer in \Cref{sec:trans} substitutes one array expression of that size per occurrence of the array variable, so the translated assertions remain polynomial.
For (iii), the skeleton of each primitive (\Cref{def:locate} and \Cref{appx:prooftrees}) consists of one loop per dimension whose invariant repeats the post-condition with a bounded index prefix, that is, $O(n)$ copies of an assertion of polynomial size; \cn{copyN}, \cn{enumN}, \cn{ampN}, \cn{locateN}, and \cn{upDim} are instances of \cn{mapN} or \cn{assignN}, \cn{sqrtSum} is a classical sum, and the conditional and measurement skeletons compose a constant number of such steps.
\end{proof}
 \section{Related Work}
\label{sec:related}

We discuss the related work below.

\noindent\textbf{\textit{Quantum Proof Systems.}}
Quantum proof systems include \qafny \cite{li2024qafnyquantumprogramverifier,10.1145/3763157}, \qbricks \cite{qbricks} based on the path-sum method \cite{Amy_2019,10.1145/3808314},
Quantum Hoare Logic (QHL) \cite{qhoare,qhoreusage,10.1145/3456877,10.1007/s00165-018-0465-3,10.1145/3808273,10.1145/3808284}, quantum separation logic (QSL) \cite{qseplocal,qsepa}, quantum relational logic
\cite{relationlogic,10.1145/3290346}, and probabilistic Hoare logic for quantum programs \cite{10.1007/978-3-642-10622-4_7}.
These proof systems are developed based on quantum program semantics and states, such as interpreting superoperators \cite{qhoare} or Dirac notation \cite{qbricks} as program-state syntax, whereas our work demonstrates the translation of a quantum proof system to the classical Hoare logic system.

The key difference between the quantum proof systems and our work is that we develop translation completeness and soundness theorems to show that quantum verification can be effectively done in Hoare logic, i.e., we compile proof rules in a quantum proof system to Hoare logic directly, and show the soundness and completeness based on the compilation. Moreover, the compilation is not specific to \qafny; we illustrate it for QHL as well.
As mentioned in \Cref{sec:frame}, there are also quantum separation logic proof systems \cite{qsepa,qseplocal} that view tensor products as a kind of separation conjunction, and we show that these separation conjunctions are pure.

\noindent\textbf{\textit{Qafny in Dafny.}}
\qafnyc \cite{10.1145/3763157} compiles \qafny verification conditions into Dafny sequences and verifies 37 quantum programs with the Dafny verifier, but proves nothing about the compilation itself. Our translation targets an idealized Hoare logic with array theory, of which Dafny is one instance, and \Cref{thm:compile-soundness,thm:relative-completeness} are what justify trusting a Dafny verdict about a quantum program: nothing \qafny can prove is lost, and nothing false is certified. Beyond that foundation, this paper contains results that \qafnyc does not: the pureness of tensor-based separating conjunction (\Cref{sec:frame}), the representation of an $\tch{n}$ state as an $n$-dimensional array with the index map $\cn{ind}$ (\Cref{sec:semantics}), the filter-and-renormalize treatment of measurement (\Cref{sec:meatrans}), and the size bound (\Cref{thm:verification-efficiency}).

\noindent\textbf{\textit{Symbolic Execution of Quantum Programs.}}
\citet{10.1145/3656419} and \citet{10.1145/3808314} execute quantum programs symbolically, the former with symbolic stabilizer states for error-correction programs and the latter with hybrid path-sums for hybrid quantum programs, and decide properties of the symbolic state of a given program. These are engines; our contribution is meta-theoretic: a translation between two proof systems with guarantees in both directions (\Cref{thm:compile-soundness,thm:relative-completeness}), with the completeness direction mechanized in Rocq, together with the pureness result of \Cref{sec:frame}. We share the engine-level view that superposition should be represented symbolically. What the translation adds, and what a symbolic execution engine does not imply, is that the classical Hoare rules for arrays already suffice, and why: separating conjunction is pure, dynamic locus reordering is handled by an index map, and measurement is filter-and-renormalize.

\noindent\textbf{\textit{Classical Proof Systems.}}
We are informed by Hoare and separation logic, as articulated in the classic paper \cite{Gordon12backgroundreading,separationlogic}, and other papers as well
\cite{10.1145/3453483.3454087,arxiv.1609.00919,10.1007/978-3-319-89960-2_13,10.1007/978-3-319-89960-2_2,nat-proof-fun,10.1145/3453483.3454036}.
Primarily, we show a compilation from \qafny to Hoare logic with array operators \cite{Gordon12backgroundreading}.

 \section{Conclusion, Future Work, and Limitations}\label{sec:limit}

We have shown that verifying a quantum program need be no harder than verifying a classical program, e.g., a parallel program: the \qafny proof system translates into Hoare logic with array theory completely and soundly, and the completeness direction is mechanized in Rocq. Three observations make the embedding tractable. First, \qafny's separation conjunctions over disjoint loci are pure and reduce to ordinary conjunctions. Second, an $\tch{n}$-typed quantum state admits a faithful representation as an $n$-dimensional array of basis-kets. Third, the semantics of quantum operations, applied to that representation, factor through a small set of array-aggregate primitives, packaged here as the Hoare-logic intermediate representation HLIR. We also bound the size of the translation: in the symbolic array representation, the generated verification conditions are polynomial in the size of the \qafny triple and the array dimension.

The embedding is constructive: every \qafny proof rule maps to an explicit HLIR derivation, and every HLIR primitive is realized by a proof-tree skeleton in Hoare logic with array theory. Quantum proof obligations therefore become deterministic and can be discharged by classical automated verifiers (e.g. Dafny, SMT-based separation logic, natural-proof engines) with no quantum-specific infrastructure. Future work includes extending the embedding to richer quantum types (such as oracle and Hadamard variants), connecting HLIR to a concrete automated back-end for end-to-end verified quantum software, and investigating whether the same determinization technique generalizes to other hybrid classical/quantum proof systems.

Our compilation only shows that quantum program verification is no harder than classical program verification, i.e., we do not claim that all quantum program verification is easy, as many quantum programs based on statistical methods, such as quantum machine learning \cite{Biamonte2017}, are hard to verify.
However, we believe the same holds for classical program verification, as many classical AI and machine-learning-based programs are also hard to verify.

\bibliography{reference}

\newpage
\appendix
\section{Proof-Tree Derivations for the Translation Lemmas}\label{appx:prooftrees}

This appendix\footnote{We plan to upload a complete technical report to arXiv in lieu of appendices.} collects the Hoare logic proof-tree derivations for the lemmas in \Cref{sec:trans}; the main text keeps each lemma, its translated triple, and the reduction step, and points here for the derivation. In every derivation, the right-hand column names the Hoare rule applied at each step.

\subsection{Proof Tree for \cn{mapN}}\label{appx:mapN}

We continue the proof of \Cref{def:locate} (\Cref{sec:oracletrans}), showing the validity of the triple for $\cassign{A'}{\cn{mapN}(A,\overline{m},f)}$ obtained there via \Cref{cor:link}.

We expand $\overline{m}$ to $(m_0,...,m_{n\sminus 1})$.
The proof requires deriving a Hoare logic proof tree. We demonstrate the simplified proof sequences with loop invariants at each level. 
The right-hand side indicates the Hoare rules to apply.

{\small
\[
\begin{array}{ll}
\textcolor{spec}{\bil Q_{\cn{len}}(A,A',(m_0,...,m_{n\sminus 1})) \bir}\\[0.2em]
\qquad \cassign{j_0}{0};... \cassign{j_{n\sminus 1}}{0}; & \textbf{Assign}\\[0.2em]
\textcolor{spec}{\bil j_0 = 0 \wedge ... \wedge j_{n\sminus 1} = 0 \wedge Q_{\cn{len}}(A,A',(m_0,...,m_{n\sminus 1})) \bir}
\\[0.2em]
\Rightarrow
\textcolor{spec}{\bil \begin{array}{l} Q_{\cn{len}}(A,A',(m_0,...,m_{n\sminus 1})) \wedge j_0 \in[0,m_0] \wedge ...\wedge j_{n\sminus 1} \in[0,m_{n\sminus 1}] \\
\wedge \forall k_0 \in [0,j_0).\,...\forall k_{n\sminus 1} \in[0,j_{n\sminus 1}).\, A'[k_0,...,k_{n\sminus 1}] = f(A[k_0,...,k_{n\sminus 1}])
\end{array}
\bir}
&
\textbf{Pre-Strength}
\\[0.2em]
\qquad \WHILE(j_0 < m_0) 
&
\textbf{While}
\\[0.2em]
\qquad \textcolor{spec}{
\bil 
\begin{array}{l}
Q_{\cn{len}}(A,A',(m_0,...,m_{n\sminus 1})) \wedge
j_0 \in[0,m_0] \wedge ...\wedge j_{n\sminus 1} \in[0,m_{n\sminus 1}] \\
\wedge \forall k_0 \in [0,j_0).\,...\forall k_{n\sminus 1} \in[0,j_{n\sminus 1}).\, A[k_0,...,k_{n\sminus 1}] = f(A[k_0,...,k_{n\sminus 1}])
\end{array}
\bir}
&
\textbf{Invariant}
\\[0.2em]
\qquad \{\\[0.2em]
\qquad\quad 
\{...
\\[0.2em]
\qquad\qquad
\WHILE(j_{n\sminus 1} < m_{n\sminus 1}) \\[0.2em]
\qquad\qquad \textcolor{spec}{
\bil
\begin{array}{l}
Q_{\cn{len}}(A,A',(m_0,...,m_{n\sminus 1})) \wedge
j_0 \in[0,m_0] \wedge ...\wedge j_{n\sminus 1} \in[0,m_{n\sminus 1}] \\
\wedge \forall k_{n\sminus 1} \in[0,j_{n\sminus 1}).\, A[j_0,...,j_{n\sminus 2},k_{n\sminus 1}] = f(A[j_0,...,j_{n\sminus 2},k_{n\sminus 1}])
\end{array}
\bir
}
&
\textbf{Invariant}
\\[0.2em]
\qquad\qquad\{\\[0.2em]
\qquad\qquad\quad
\cassign{A'[k_0,...,k_{n\sminus 1}]}{f(A[k_0,...,k_{n\sminus 1}])};
\\[0.2em]
\qquad\qquad\quad
\cassign{j_{n\sminus 1}}{j_{n\sminus 1}+1}
\\[0.2em]
\qquad\qquad \}\\[0.2em]
\qquad\qquad ...\\[0.2em]
\qquad\quad \}\\[0.2em]
\qquad\quad
\cassign{j_{0}}{j_{0}+1}\\[0.2em]
\qquad \}\\[0.2em]
\textcolor{spec}{
\Big{ \{ }
\begin{array}{l}
Q_{\cn{len}}(A,A',(m_0,...,m_{n\sminus 1})) \wedge
\\
\neg (j_0 < m_0)\wedge ... \wedge \neg (j_{n\sminus 1} < m_{n\sminus 1}) \wedge j_0 \in[0,m_0]  \wedge ...\wedge j_{n\sminus 1} \in[0,m_{n\sminus 1}] 
\\
\qquad
\wedge \forall k_0 \in [0,j_0).\,...\forall k_{n\sminus 1} \in[0,j_{n\sminus 1}).\, A\{k_0,...,k_{n\sminus 1} \mapsto f(A[k_0,...,k_{n\sminus 1}]) \}
\end{array}
\Big{ \} }
}\\[0.5em]
\Rightarrow
\textcolor{spec}{
\bil 
\begin{array}{l}
Q_{\cn{len}}(A,A',(m_0,...,m_{n\sminus 1})) \wedge
\\[0.2em]
\forall k_0 \in [0,m_0).\,...\forall k_{n\sminus 1} \in[0,m_{n\sminus 1}).\, A\{k_0,...,k_{n\sminus 1} \mapsto f(A[k_0,...,k_{n\sminus 1}])\}
\end{array}
\bir
}
&
\textbf{Post-Weakening}
\\[0.2em]
\end{array}
\]
}

\subsection{Proof Tree for \cn{assignN}}\label{sec:assignproof}

We demonstrate the proof-tree derivation for \Cref{def:gassign} (\cn{assignN}) below, for the instance $f = \lambda u\,v.\cn{amp}(v)$ in which the amplitude is taken from $A_1$; the general case replaces $\cn{amp}(A_1[k_0,...,k_{n\sminus 1}])$ by $f(A[k_0,...,k_{n\sminus 1}],A_1[k_0,...,k_{n\sminus 1}])$ throughout.

\begin{proof}\label{appx:gassignProof}\rm 
The proof requires a derivation of a Hoare logic proof tree, having a simplified proof sequence as follows. 
Mainly, we implement the \cn{assignN} function as the program with assertions below. Here, $\cn{up}^b$ inserts $A_1$'s bits into $A$'s vector at the indices $\cn{ind}^{\varphi}(\eta)$ (via \cn{merge}), and its proof is given in the Rocq mechanization.

{\small
\[
\begin{array}{ll}
\textcolor{spec}{\bil Q_{\cn{len}}(A,A_1,A_2,(m_0,...,m_{n\sminus 1})) \bir}\\[0.2em]
\qquad \cassign{j_0}{0};... \cassign{j_{n\sminus 1}}{0}; & \textbf{Assign}\\[0.2em]
\textcolor{spec}{\bil Q_{\cn{len}}(A,A_1,A_2,(m_0,...,m_{n\sminus 1})) \wedge j_0 = 0 \wedge ... \wedge j_{n\sminus 1} = 0  \bir}
\\[0.2em]
\Rightarrow
\textcolor{spec}{
\Big{ \{ } 
\begin{array}{l} Q_{\cn{len}}(A,A_1,A_2,(m_0,...,m_{n\sminus 1}))\wedge j_0 \in[0,m_0] \wedge ...\wedge j_{n\sminus 1} \in[0,m_{n\sminus 1}] \\
\wedge
\forall k_0 \in [0,j_0).\,...\forall k_{n\sminus 1} \in[0,j_{n\sminus 1}).\, A_2[k_0,...,k_{n\sminus 1}]= 
\cn{up}^a(\cn{up}^b(
\\[0.2em]
\quad
A[k_0,...,k_{n\sminus 1}],
\cn{vector}(A_1[k_0,...,k_{n\sminus 1}]),\cn{ind}^{\varphi}(\eta)),\cn{amp}(A_1[k_0,...,k_{n\sminus 1}]))
\end{array}
\Big{ \} }
}
&
\textbf{Pre-Strength}
\\[0.2em]
\qquad \WHILE(j_0 < m_0) 
&
\textbf{While}
\\[0.2em]
\qquad
\textcolor{spec}{
\Big{ \{ } 
\begin{array}{l} Q_{\cn{len}}(A,A_1,A_2,(m_0,...,m_{n\sminus 1}))\wedge j_0 \in[0,m_0] \wedge ...\wedge j_{n\sminus 1} \in[0,m_{n\sminus 1}] \\
\wedge
\forall k_0 \in [0,j_0).\,...\forall k_{n\sminus 1} \in[0,j_{n\sminus 1}).\, A_2[k_0,...,k_{n\sminus 1}]= 
\cn{up}^a(\cn{up}^b(
\\[0.2em]
\quad
A[k_0,...,k_{n\sminus 1}],
\cn{vector}(A_1[k_0,...,k_{n\sminus 1}]),\cn{ind}^{\varphi}(\eta)),\cn{amp}(A_1[k_0,...,k_{n\sminus 1}]))
\end{array}
\Big{ \} }
}
&
\textbf{Invariant}
\\[0.2em]
\qquad \{\\[0.2em]
\qquad\quad 
\{...
\\[0.2em]
\qquad\qquad
\WHILE(j_{n\sminus 1} < m_{n\sminus 1}) \\[0.2em]
\qquad\qquad \textcolor{spec}{
\Big{ \{ } 
\begin{array}{l}
 Q_{\cn{len}}(A,A_1,A_2,(m_0,...,m_{n\sminus 1}))\wedge
j_0 \in[0,m_0] \wedge ...\wedge j_{n\sminus 1} \in[0,m_{n\sminus 1}] \\
\wedge
\forall k_{n\sminus 1} \in[0,j_{n\sminus 1}).\, A_2[j_0,...,j_{n\sminus 2},k_{n\sminus 1}]= \cn{up}^a(\cn{up}^b(A[{j_0,...,j_{n\sminus 2},k_{n\sminus 1}}],
\\[0.2em]
\qquad
\cn{vector}(A_1[{j_0,...,j_{n\sminus 2},k_{n\sminus 1}}]),\cn{ind}^{\varphi}(\eta)),\cn{amp}(A_1[{j_0,...,j_{n\sminus 2},k_{n\sminus 1}}]))
\end{array}
\Big{ \} }
}
&
\textbf{Invariant}
\\[0.2em]
\qquad\qquad\{\\[0.2em]
\qquad\qquad\quad
\cassign{t}{\cn{up}^b(A[k_0,...,k_{n\sminus 1}],\cn{vector}(A_1[k_0,...,k_{n\sminus 1}]),\cn{ind}^{\varphi}(\eta))};
\\[0.2em]
\qquad\qquad\quad
\cassign{A_2[k_0,...,k_{n\sminus 1}]}{\cn{up}^a(t,\cn{amp}(A_1[k_0,...,k_{n\sminus 1}]))};
\\[0.2em]
\qquad\qquad\quad
\cassign{j_{n\sminus 1}}{j_{n\sminus 1}+1};
\\[0.2em]
\qquad\qquad \}\\[0.2em]
\qquad\qquad ...\\[0.2em]
\qquad\quad \}\\[0.2em]
\qquad\quad
\cassign{j_{0}}{j_{0}+1};\\[0.2em]
\qquad \}\\[0.2em]
\textcolor{spec}{
\Big{ \{ } 
\begin{array}{l}
 Q_{\cn{len}}(A,A_1,A_2,(m_0,...,m_{n\sminus 1}))\wedge
 \\
\neg (j_0 < m_0)\wedge ... \wedge \neg (j_{n\sminus 1} < m_{n\sminus 1}) \wedge j_0 \in[0,m_0]  \wedge ...\wedge j_{n\sminus 1} \in[0,m_{n\sminus 1}] 
\\
\wedge
\forall k_0 \in [0,j_0).\,...\forall k_{n\sminus 1} \in[0,j_{n\sminus 1}).\, A_2[k_0,...,k_{n\sminus 1}]= 
\cn{up}^a(\cn{up}^b(
\\[0.2em]
\quad
A[k_0,...,k_{n\sminus 1}],
\cn{vector}(A_1[k_0,...,k_{n\sminus 1}]),\cn{ind}^{\varphi}(\eta)),\cn{amp}(A_1[k_0,...,k_{n\sminus 1}]))
\end{array}
\Big{ \} }
}\\[2.5em]
\Rightarrow
\textcolor{spec}{\Big{ \{ } 
\begin{array}{l}
 Q_{\cn{len}}(A,A_1,A_2,(m_0,...,m_{n\sminus 1}))\wedge
\\[0.2em]
\forall k_0 \in [0,m_0).\,...\forall k_{n\sminus 1} \in[0,m_{n\sminus 1}).\, A_2[k_0,...,k_{n\sminus 1}]= 
\cn{up}^a(\cn{up}^b(
\\[0.2em]
\quad
A[k_0,...,k_{n\sminus 1}],
\cn{vector}(A_1[k_0,...,k_{n\sminus 1}]),\cn{ind}^{\varphi}(\eta)),\cn{amp}(A_1[k_0,...,k_{n\sminus 1}]))
\end{array}
\Big{ \} } }
&
\textbf{Post-Weakening}
\\[0.2em]
\end{array}
\]
}
\end{proof}

\subsection{Proof Tree for \cn{enumN}}\label{appx:enumN}

We continue the proof of \Cref{def:enumN} (\Cref{sec:hadtrans}), showing the validity of the triple obtained there via \Cref{cor:link}.

The proof tree is shown below.
The pre-condition strengthening rules utilize the $Q_{\cn{len}}$ axioms for $A_1$, $A_a$, $A_b$, $A_2$, $A_a\{E_1\}$, $A_b\{E_2\}$, and $A_2\{E_3\}$, as they all have the same bounds ($\overline{m},2^{\slen{\eta}}$).

{\small
\[\hspace*{-0.5em}
\begin{array}{l}
\cn{let }\\
\;\;\textcolor{spec}{E_1} := 
\textcolor{spec}{
\forall \overline{j}\in\cn{range}(\overline{m}).\forall k\in[0,2^{\slen{\eta}}). (\overline{j},k) \mapsto\cn{up}^{b}(A_1[\overline{j},k],\cn{vector}(A_1[\overline{j},k])[\cn{ind}^{\varphi}(\eta)],\cn{ind}^{\varphi}(\eta))}
\\[0.2em]
\;\;\textcolor{spec}{E_2} := 
\textcolor{spec}{
\forall \overline{j}\in\cn{range}(\overline{m}).\forall k\in[0,2^{\slen{\eta}}). (\overline{j},k) \mapsto\cn{up}^{b}(A_a[\overline{j},k],k,\cn{ind}^{\varphi}(\eta))}
\\[0.2em]
\;\;\textcolor{spec}{E_3} := 
\textcolor{spec}{\forall \overline{j}\in\cn{range}(\overline{m}).\forall k\in[0,2^{\slen{\eta}}). (\overline{j},k) \mapsto \cn{up}^a(\cn{up}^b(A_1[\overline{j},k],\cn{vector}(A_b[\overline{j},k]),\cn{ind}^{\varphi}(\eta)),\cn{amp}(A_b[\overline{j},k]))}
\end{array}
\]
}

{\small
\[\hspace{-0.8em}
\begin{array}{ll}
\textcolor{spec}{\bil 
Q_{\cn{len}}(A_1,\overline{m},2^{\slen{\eta}})
\bir }
\\[0.2em]
\Rightarrow
    \textcolor{spec}{
\bil 
\begin{array}{l}
\forall \overline{j} \in \cn{range}(\overline{m}).\forall k\in[0,2^{\slen{\eta}}).
\overline{j},k\mapsto \cn{up}^b(A_1[\overline{j},k],k,\cn{ind}^{\varphi}(\eta))
\\[0.2em]
\;\;
= \cn{up}^{b}(A_1[\overline{j},k],k, \cn{ind}^{\varphi}(\eta))
\wedge Q_{\cn{len}}(A_1,A_a\{E_1\},\overline{m},2^{\slen{\eta}})
\end{array}
\bir
}
&
\textbf{Pre-Strength}
\\[0.3em]
\equiv
    \textcolor{auto}{
\Big{ \{ }
\begin{array}{l}
\forall \overline{j} \in \cn{range}(\overline{m}).\forall k\in[0,2^{\slen{\eta}}).
A_2\{E_3[A_b\{E_2[A_a\{E_1\}/A_a]\}/A_b]\}[\overline{j},k]
\\[0.1em]
\;\;
= \cn{up}^{b}(A_1[\overline{j},k],k,
 \cn{ind}^{\varphi}(\eta))
\wedge Q_{\cn{len}}(A_1,A_a\{E_1\},(\overline{m},2^{\slen{\eta}})) 
\end{array}
\Big{ \} } }
\\[0.3em]
\quad \cassign{A_a}{\cn{locateN}(A_1,\overline{m},\cn{ind}^{\varphi}(\eta))};
&
\textbf{LocateN}
\\[0.2em]
    \textcolor{spec}{
\Big{ \{ }
\begin{array}{l}
\forall \overline{j} \in \cn{range}(\overline{m}).\forall k\in[0,2^{\slen{\eta}}).
A_2\{E_3[A_b\{E_2\}/A_b]\}[\overline{j},k]
= \cn{up}^{b}(A_1[\overline{j},k],k,
\\[0.2em]
\;\;
 \cn{ind}^{\varphi}(\eta))
\wedge Q_{\cn{len}}(A_1,A_a,(\overline{m},2^{\slen{\eta}})) 
\end{array}
\Big{ \} } }
\\[0.2em]
\Rightarrow
    \textcolor{spec}{
\Big{ \{ }
\begin{array}{l}
\forall \overline{j} \in \cn{range}(\overline{m}).\forall k\in[0,2^{\slen{\eta}}).
A_2\{E_3[A_b\{E_2\}/A_b]\}[\overline{j},k]
= \cn{up}^{b}(A_1[\overline{j},k],k,
\\[0.2em]
\;\;
 \cn{ind}^{\varphi}(\eta))
\wedge Q_{\cn{len}}(A_1,
A_2\{E_3[A_b\{E_2\}/A_b]\},A_a,A_b\{E_2\},(\overline{m},2^{\slen{\eta}})) 
\end{array}
\Big{ \} } }
&
\textbf{Pre-Strength}
\\[0.2em]
\quad
    \cassign{A_b}{\cn{mapN}(A_a,\overline{m},\lambda (x,y) . (x,\cn{enum}(2^{\slen{\eta}})))};
    &
    \textbf{MapN}
\\[0.2em]
    \textcolor{spec}{
\Big{ \{ }
\begin{array}{l}
\forall \overline{j} \in \cn{range}(\overline{m}).\forall k\in[0,2^{\slen{\eta}}).
A_2\{E_3\}[\overline{j},k]
= \cn{up}^{b}(A_1[\overline{j},k],k, \cn{ind}^{\varphi}(\eta))
\\[0.2em]
\;\;
\wedge Q_{\cn{len}}(A_1,
A_2\{E_3\},A_a,A_b,(\overline{m},2^{\slen{\eta}})) 
\end{array}
\Big{ \} } }
\\[0.2em]
\quad\cassign{A_2}{\cn{assignN}(A_1,A_b,\overline{m},\cn{ind}^{\varphi}(\eta),\lambda u\,v.\cn{amp}(u))};
&
\textbf{AssignN}
\\[0.2em]
\textcolor{spec}{\bil 
\begin{array}{l}
\forall \overline{j} \in \cn{range}(\overline{m}).\forall k\in[0,2^{\slen{\eta}}). A_2[\overline{j},k]=\cn{up}^{b}(A_1[\overline{j},k],k,\cn{ind}^{\varphi}(\eta)) 
\\[0.2em]
\wedge Q_{\cn{len}}(A_1,A_2,A_a,A_b,(\overline{m},2^{\slen{\eta}})) 
\end{array}
\bir}
\\[0.3em]
\Rightarrow
\textcolor{spec}{\bil 
\begin{array}{l}
\forall \overline{j} \in \cn{range}(\overline{m}).\forall k\in[0,2^{\slen{\eta}}). A_2[\overline{j},k]=\cn{up}^{b}(A_1[\overline{j},k],k,\cn{ind}^{\varphi}(\eta)) 
\\[0.1em]
\wedge Q_{\cn{len}}(A_1,A_2,(\overline{m},2^{\slen{\eta}})) 
\end{array}
\bir}
&
\textbf{Post-Weakening}
\end{array}
\]
}

There is a single $\cn{enum}$ function above. We show the \cn{enum} function proof below.

\begin{lemma}[A Single Enumeration Function]\label[lemma]{def:singleEnum}\rm
Given a length $n$, we define a single enumeration function $\cn{enum}(n)$ to generate an array $A$ as the following valid Hoare triple.

{\small
\[
\begin{array}{l}
\triple{0\le n}{\cassign{A}{\cn{enum}(n)}}{\forall k\in[0,n).\,A[k] = k \wedge \slen{A} = n}
\end{array}
\]
}
\end{lemma}

\begin{proof}\rm
The proof of the valid Hoare triple is given as follows.

{\small
\[
\begin{array}{l}
\textcolor{spec}{\{0 \le n\}}\\[0.2em]
\qquad \cassign{A}{\cn{init}(n)};\\[0.2em]
\textcolor{spec}{\{0 \le n\wedge \slen{A} = n\}}\\[0.2em]
\Rightarrow
\textcolor{spec}{\{j\in[0,n]\wedge \forall k\in[0,j).\, A[k] = k\wedge \slen{A} = n\}}\\[0.2em]
\qquad \WHILE(j < n) \\[0.2em]
\qquad \textcolor{spec}{\{j\in[0,n]\wedge \forall k\in[0,j).\, A[k] = k\wedge \slen{A} = n \}}\\[0.2em]
\qquad \{\\[0.2em]
\qquad\quad \cassign{A[j]}{j};\\[0.2em]
\qquad\quad \cassign{j}{j+1};\\[0.2em]
\qquad \}\\[0.2em]
\textcolor{spec}{\{\neg (j < n)\wedge j\in[0,n]\wedge \forall k\in[0,j).\, A[k] = k\wedge \slen{A} = n \}}\\[0.2em]
\Rightarrow
\textcolor{spec}{\{\forall k\in[0,n).\, A[k] = k\wedge \slen{A} = n\}}\\[0.2em]
\end{array}
\]
}
\end{proof}

We show the proof for a single \cn{ampMult} function below.

\begin{lemma}[A Single Amplitude Multiplication Function]\label[lemma]{def:singleamp}\rm
Given two bitstring arrays $A_1$ and $A_2$, a collection of indices $B$, and a complex number $\alpha$, we define a single amplitude multiplication function to compute a complex number $\alpha'$ below. $\cdot$ is a dot product.

{\small
\[
\begin{array}{l}
\textcolor{spec}{\{\slen{B} \le \slen{A_1} \wedge \slen{B} \le \slen{A_2} \wedge \forall k \in[0,\slen{B}) .\, B[k]\in[0,\slen{A_1}) \wedge B[k] \in [0,\slen{A_2})\}}
\\[0.2em]
\qquad{\cassign{\alpha'}{\cn{ampMult}(A_1,A_2,\alpha,B)}}
\\[0.2em]
\textcolor{spec}{\{\alpha' = \alpha * \eexp{(\cn{group}(|A_1|,A_1,B)\cdot \cn{group}(|A_2|,A_2,B))}\}}
\end{array}
\]
}
\end{lemma}

\begin{proof}\rm
We implement $\cn{ampMult}(A_1,A_2,\alpha,B)$ to be the program $\cassign{A'_1}{\cn{group}(|A_1|,A_2,B)};\cassign{A'_2}{\cn{group}(|A_2|,A_2,B)};\cassign{\alpha'}{\alpha \cdot \eexp{(A'_1\cdot A'_2)}}$.

{\small
\[
\begin{array}{l}
\textcolor{spec}{\{\slen{B} \le \slen{A_1} \wedge \slen{B} \le \slen{A_2} \wedge \forall k \in[0,\slen{B}) .\, B[k]\in[0,\slen{A_1}) \wedge B[k] \in [0,\slen{A_2})\}}\\[0.2em]
\qquad \cassign{A'_1}{\cn{group}(|A_1|,A_1,B)};\\[0.2em]
\textcolor{spec}{\{\slen{B} \le \slen{A_1} \wedge \slen{B} \le \slen{A_2} \wedge \forall k \in[0,\slen{B}) .\, B[k]\in[0,\slen{A_1}) \wedge B[k] \in [0,\slen{A_2}) \wedge \forall k \in [0,\slen{B}) .\, A'_1[k] = A_1[B[k]]\}}\\[0.2em]
\qquad \cassign{A'_2}{\cn{group}(|A_2|,A_2,B)};\\[0.2em]
\textcolor{spec}{\{\slen{B} \le \slen{A_1} \wedge \slen{B} \le \slen{A_2} \wedge \forall k \in[0,\slen{B}) .\, B[k]\in[0,\slen{A_1}) \wedge B[k] \in [0,\slen{A_2}) }
\\[0.2em]
\qquad\quad\textcolor{spec}{\wedge \forall k \in [0,\slen{B}) .\, A'_1[k] = A_1[B[k]] \wedge  A'_2[k] = A_2[B[k]]\}}\\[0.2em]
\qquad \cassign{\alpha'}{\alpha \cdot \eexp{(A'_1\cdot A'_2)}} \\[0.2em]
\textcolor{spec}{\{\slen{B} \le \slen{A_1} \wedge \slen{B} \le \slen{A_2} \wedge \forall k \in[0,\slen{B}) .\, B[k]\in[0,\slen{A_1}) \wedge B[k] \in [0,\slen{A_2}) }
\\[0.2em]
\qquad\quad\textcolor{spec}{\wedge \forall k \in [0,\slen{B}) .\, A'_1[k] = A_1[B[k]] \wedge  A'_2[k] = A_2[B[k]]
\wedge \alpha' = \alpha \cdot \eexp{(\sum_{k=0}^{|B|-1}\cdot A'_1[k] \cdot A'_2[k])}
\}}\\[0.2em]
\Rightarrow
\textcolor{spec}{\{\alpha' = \alpha * \eexp{(\cn{group}(|A_1|,A_1,B)\cdot \cn{group}(|A_2|,A_2,B))}\}}\\[0.2em]
\end{array}
\]
}
\end{proof}

\subsection{Proof Tree for Quantum Conditionals}\label{appx:qif}

We continue the proof of \Cref{proof:qif} (\Cref{sec:condtrans}), using the array-length axioms and the \cn{concat}/\cn{filterN} axiom stated there.

{\small
\[\hspace*{-0.5em}
\begin{array}{l}
\cn{let }\\
\;\;\textcolor{spec}{E_1} := 
\textcolor{spec}{
\forall \overline{j}\in\cn{range}(\overline{u}). \overline{j} \mapsto\cn{up}^{b}(A_1[\overline{j}],\cn{vector}(A_1[\overline{j}]),\cn{ind}^{\varphi}(\eta))}
\\[0.2em]
\;\;\textcolor{spec}{E_2} := 
\textcolor{spec}{\forall \overline{j}\in \cn{range}(\overline{u'}).\overline{j}\mapsto \cn{up}^a(\cn{up}^b(A_4[\overline{j}],\cn{vector}(A_3[\overline{j}]),\cn{ind}^{\varphi}(\eta')),\cn{amp}(A_3[\overline{j}]))}
\end{array}
\]
}

{\small
\[\hspace*{-0.5em}
\begin{array}{ll}
\textcolor{spec}{\bil P \bir}
\\[0.1em]
\Rightarrow
\textcolor{spec}{
\bil
\begin{array}{l}
P[\cn{concat}(\cn{filterN}(A,\neg b,\cn{ind}^{\varphi}(\eta)),
\cn{filterN}(A,b,\cn{ind}^{\varphi}(\eta)))/A]
\end{array}
\bir
}
&
\textbf{Pre-Strength}
\\[0.2em]
\quad   \cassign{A_1}{\cn{filterN}(A,b,\cn{ind}^{\varphi}(\eta))};\cassign{\overline{u}}{\slen{A_1}} ;
&
\textbf{FilterN, Len}
\\[0.2em]
\textcolor{spec}{
\bil
\begin{array}{l}
P[\cn{concat}(\cn{filterN}(A,\neg b,\cn{ind}^{\varphi}(\eta)),
A_1)/A]
\end{array}
\bir
}
\\[0.2em]
\quad   \cassign{A_2}{\cn{filterN}(A,\neg b,\cn{ind}^{\varphi}(\eta))}; \cassign{\overline{t}}{\slen{A_2}} ;
&
\textbf{FilterN, Len}
\\[0.2em]
\textcolor{spec}{
\bil
\begin{array}{l}
P[\cn{concat}(A_2,A_1)/A]
\end{array}
\bir
}
\\[0.2em]
\Rightarrow
\textcolor{spec}{
\Big{\{}
\begin{array}{l}
P[\cn{concat}(A_2,\overline{t},(\overline{t},\overline{1})),
\cn{upDim}(A_1,\overline{u},\slen{A_3\{E_1\}})\{\forall \overline{j}\in \cn{range}({\slen{A_3\{E_1\}}}).
\\[0.1em]\quad
\overline{j}\mapsto \cn{up}^a(
\cn{up}^b(\cn{upDim}(A_1,\overline{u},\slen{A_3\{E_1\}})[\overline{j}],
\\[0.1em]
\quad\quad\cn{vector}(A_3\{E_1\}[\overline{j}]),\cn{ind}^{\varphi}(\eta')),\cn{amp}(A_3\{E_1\}[\overline{j}])) \})/A]
\end{array}
\Big{\}}
}
&
\textbf{Pre-Strength}
\\[0.2em]
\quad   \cassign{A_3}{\cn{locateN}(A_1,\overline{u},\cn{ind}^{\varphi}(\eta'))};
&
\textbf{LocateN}
\\[0.2em]
\textcolor{spec}{\bil
\begin{array}{l}
P[\cn{concat}(\cn{upDim}(A_2,\overline{t},(\overline{t},\overline{1})),
\cn{upDim}(A_1,\overline{u},\slen{A_3})\{\forall \overline{j}\in \cn{range}(\overline{\slen{A_3}}). \overline{j}\mapsto \cn{up}^a(\cn{up}^b(
\\[0.1em]
\quad
\cn{upDim}(A_1,\overline{u},\slen{A_3})[\overline{j}],\cn{vector}(A_3[\overline{j}]),\cn{ind}^{\varphi}(\eta')),\cn{amp}(A_3[\overline{j}])) \})/A]
\end{array}
\bir}
\\[0.2em]
\quad   e;
&
\textbf{Upper}
\\[0.2em]
\textcolor{spec}{\bil
\begin{array}{l}
Q[\cn{concat}(\cn{upDim}(A_2,\overline{t},(\overline{t},\overline{1})),
\cn{upDim}(A_1,\overline{u},\slen{A_3})\{\forall \overline{j}\in \cn{range}({\slen{A_3}}). \overline{j}\mapsto \cn{up}^a(\cn{up}^b(
\\[0.1em]
\quad
\cn{upDim}(A_1,\overline{u},\slen{A_3})[\overline{j}],\cn{vector}(A_3[\overline{j}]),\cn{ind}^{\varphi}(\eta')),\cn{amp}(A_3[\overline{j}])) \})/A]
\end{array}
\bir}
\\[0.2em]
\Rightarrow
&
\textbf{Pre-Strength}
\\[0.2em]
\textcolor{spec}{\bil
\begin{array}{l}
Q[\cn{concat}(\cn{upDim}(A_2,\overline{t},(\overline{t},\overline{1})),A_6\{E_2[\cn{upDim}(A_1,\overline{u},\overline{u'})/A_4][\slen{A_3}/\overline{u'}]\})/A]
\wedge 
\\[0.1em]
Q_{\cn{len}}(A_6\{E_2[\cn{upDim}(A_1,\overline{u},\overline{u'})/A_4][\slen{A_3}/\overline{u'}]\},\slen{A_3}) \wedge Q_{\cn{len}}(\cn{upDim}(A_2,\overline{t},(\overline{t},\overline{1})),(\overline{t},\overline{1}))
\end{array}
\bir}
\\[0.2em]
\quad \cassign{\overline{u'}}{\slen{A_3}}; \cassign{A_4}{\cn{upDim}(A_1,\overline{u},\overline{u'})};
&
\textbf{Len, UpDim}
\\[0.2em]
\textcolor{spec}{\bil
\begin{array}{l}
Q[\cn{concat}(\cn{upDim}(A_2,\overline{t},(\overline{t},\overline{1})),A_6\{E_2\})/A]
\wedge Q_{\cn{len}}(A_6\{E_2\},\overline{u'})
\\[0.1em]
\wedge Q_{\cn{len}}(\cn{upDim}(A_2,\overline{t},(\overline{t},\overline{1})),(\overline{t},\overline{1}))
\end{array}
\bir}
\\[0.2em]
\quad
    \cassign{A_5}{\cn{upDim}(A_2,\overline{t},(\overline{t},\overline{1}))};
    &
    \textbf{UpDim}
    \\[0.2em]
\textcolor{spec}{\bil
Q[\cn{concat}(A_5,A_6\{E_2\})/A]
\wedge Q_{\cn{len}}(A_6\{E_2\},\overline{u'}) \wedge Q_{\cn{len}}(A_5,(\overline{t},\overline{1}))
\bir}
\\[0.2em]
\quad\cassign{A_6}{\cn{assignN}(A_4,A_3,\overline{u'},\cn{ind}^{\varphi}(\eta'),\lambda v\,l.\cn{amp}(l))};
&
\textbf{AssignN}
\\[0.2em]
\textcolor{spec}{\bil 
Q[\cn{concat}(A_5,A_6)/A] \wedge Q_{\cn{len}}(\cn{concat}(A_5,A_6),\overline{u'}+(\overline{t},\overline{1}))
\bir}
\\[0.2em]
\quad\cassign{A_7}{\cn{concat}(A_5,A_6)};
&
\textbf{Concat}
\\[0.2em]
\textcolor{spec}{\bil 
Q[A_7/A] \wedge Q_{\cn{len}}(A_7,\overline{u'}+(\overline{t},\overline{1})) 
\bir}
\\[0.2em]
\quad\cassign{A}{A_7}
&
\textbf{Assign}
\\[0.2em]
\textcolor{spec}{\bil 
Q
\wedge Q_{\cn{len}}(A,\overline{u'}+(\overline{t},\overline{1})) 
\bir}
\end{array}
\]
}

\subsection{Proof Tree for Quantum Measurement}\label{appx:mea}

We continue the proof of \Cref{def:mea} (\Cref{sec:meatrans}); as there, the $Q_{\cn{len}}$ judgments are omitted.

{\small
\[\hspace*{-0.5em}
\begin{array}{l}
\cn{let }
\textcolor{spec}{E_1} := 
\textcolor{spec}{\cn{filterN}(A,\lambda y.y = x,\cn{ind}^{\varphi}(\eta))}
\;\;\cn{ and }
\;\;\textcolor{spec}{E_2} := 
\textcolor{spec}{\cn{up}^b(A_1[\overline{j}],\cn{vector}(A_1[\overline{j}]),\cn{ind}^{\varphi}(\eta'))}
\end{array}
\]
}

{\small
\[\hspace{-0.5em}
\begin{array}{ll}
    \textcolor{spec}{
\Big{ \{ }
\begin{array}{l} 
P[E_1\{\forall \overline{j}\in\cn{range}(\slen{E_1}).\overline{j}\mapsto \cn{up}^a(\cn{up}^b(E_1[\overline{j}],\cn{vector}(E_1[\overline{j}]),\cn{ind}^{\varphi}(\eta'))
,\frac{\cn{amp}(E_1)}{r})\}/A]
\\[0.2em]\qquad[\sqrt{\sum_{\overline{j}\in\cn{range}(\slen{E_1})}\slen{\cn{amp}(E_1)}^2}/r]
[\cn{vector}(A[\overline{i}])[\cn{ind}^{\varphi}(\eta)]/x]
\end{array}
\Big{ \} } }
&
\textbf{Indexing}
\\[0.2em]
\quad \cassign{x}{\cn{vector}(A[\overline{i}])[\cn{ind}^{\varphi}(\eta)]};
&
\textbf{Indexing}
\\[0.2em]
    \textcolor{spec}{
\Big{ \{ }
\begin{array}{l} 
P[E_1\{\forall \overline{j}\in\cn{range}(\slen{E_1}).\overline{j}\mapsto \cn{up}^a(\cn{up}^b(E_1[\overline{j}],\cn{vector}(E_1[\overline{j}]),\cn{ind}^{\varphi}(\eta'))
\\[0.2em]\qquad
,\frac{\cn{amp}(E_1)}{r})\}/A][\sqrt{\sum_{\overline{j}\in\cn{range}(\slen{E_1})}\slen{\cn{amp}(E_1)}^2}/r]
\end{array}
\Big{ \} } }
\\[0.2em]
\quad \cassign{A_1}{\cn{filterN}(A,\lambda y.y = x,\cn{ind}^{\varphi}(\eta))};
&
\textbf{FilterN}
\\[0.2em]
    \textcolor{auto}{
\bil 
P[A_1\{\forall \overline{j}\in\cn{range}(\slen{A_1}).\overline{j}\mapsto \cn{up}^a(E_2,\frac{\cn{amp}(A_1)}{r})\}/A]
[\sqrt{\sum_{\overline{j}\in\cn{range}(\slen{A_1})}\slen{\cn{amp}(A_1)}^2}/r]
\bir }
\\[0.4em]
\Rightarrow
    \textcolor{spec}{
\Big{ \{ }
(
P[A_3\{\forall \overline{j}\in\cn{range}(\slen{A_1}).\overline{j}\mapsto \cn{up}^a(E_2,\frac{\cn{amp}(E_2)}{r})\}/A]
)[\sqrt{\sum_{\overline{j}\in\cn{range}(\slen{A_1})}\slen{\cn{amp}(E_2)}^2}/r]
\Big{ \} } }
&
    \textbf{Pre-Strength}
\\[0.2em]
\quad
\cassign{\overline{t}}{\slen{A_1}};
    &
    \textbf{Len}
\\[0.2em]
    \textcolor{spec}{
\Big{ \{ }
(
P[A_3\{\forall \overline{j}\in\cn{range}(\overline{t}).\overline{j}\mapsto \cn{up}^a(E_2,\frac{\cn{amp}(E_2)}{r})\}/A]
)[\sqrt{\sum_{\overline{j}\in\cn{range}(\overline{t})}\slen{\cn{amp}(E_2)}^2}/r]
\Big{ \} } }
\\[0.2em]
\quad
\cassign{A_2}{\cn{locateN}(A_1,\overline{t},\cn{ind}^{\varphi}(\eta'))};
    &
    \textbf{LocateN}
    \\[0.2em]
    \textcolor{spec}{
\Big{ \{ }
(
P[A_3\{\forall \overline{j}\in\cn{range}(\overline{t}).\overline{j}\mapsto \cn{up}^a(A_2[\overline{j}],\frac{\cn{amp}(A_2[\overline{j}])}{r})\}/A]
)[\sqrt{\sum_{\overline{j}\in\cn{range}(\overline{t})}\slen{\cn{amp}(A_2[\overline{j}])}^2}/r]
\Big{ \} } }
\\[0.2em]
\quad\cassign{r}{\cn{sqrtSum}(A_2,\overline{t})};
&
\textbf{SqrtSum}
\\[0.2em]
\textcolor{spec}{\bil 
P[A_3\{\forall \overline{j}\in\cn{range}(\overline{t}).\overline{j}\mapsto \cn{up}^a(A_2[\overline{j}],\frac{\cn{amp}(A_2[\overline{j}])}{r})\}/A]
\bir}
\\[0.3em]
\quad
\cassign{A_3}{\cn{mapN}(A_2,\overline{t}, \lambda (y,u) . (\frac{y}{r},u))};
&
\textbf{MapN}
\\[0.2em]
\textcolor{spec}{\bil 
P[A_3/A]
\bir}
\\[0.2em]
\quad\cassign{A}{A_3}
&
\textbf{Assign}
\\[0.2em]
\textcolor{spec}{\bil 
P
\bir}
\end{array}
\]
}
 \section{Transformation for QHL}\label{appx:qhl}

We show below the transformation for quantum Hoare logic (QHL) augmented with classical variables.
A chain of QHL work \cite{10.1145/3456877,YING2026105417,10.1145/3808273,ying2024verificationrecursivelydefinedquantum} combines classical variables with quantum data in a QHL proof rule. In this paper, we call this chain CQHL, and we describe the proof rule as follows.

{
\[
\qhltriple{A,\ket{\varphi}}{C}{B,\ket{\psi}}
\]
}

The rule suggests that the execution of the program $C$ under a state described by the pre-condition $\textcolor{purple}{\{A,\ket{\varphi}\}}$ results in a post-state described by the post-condition $\textcolor{purple}{\{B,\ket{\psi}\}}$, where $A$ and $B$ are pre- and post-conditions on the classical state, while $\ket{\varphi}$ and $\ket{\psi}$ are quantum data predicates written in the Dirac notation of a quantum datum.

CQHL distinguishes classical variables ($x$ or $y$) from quantum variables ($q$ or $\overline{q}$).
Quantum variables can be single-qubit variables ($q$) or quantum array variables ($\overline{q}$).
A quantum array variable can be written in the form $q[x,y)$ where $x$ and $y$ are two indices in the array $q$, i.e., we have $x\le y$ and $q[x,y)=q[x]\sqcupplus...\sqcupplus q[y-1]$. Quantum array variables are indexed like classical arrays.

CQHL labels kets with quantum variable subscripts, e.g., $\textcolor{purple}{\ket{d}_q}$ refers to a basis-ket state in the qubit $q$ position and $\textcolor{purple}{\ket{\psi}_{\overline{q}}}$ refers to a quantum datum $\ket{\psi}$ in the quantum array $\overline{q}$. 

{
\begin{mathpar}
\inferrule[]{}{\qhltriple{A,U^{\dag}{(\ket{\varphi}_{\overline{q}})\ket{\psi}}}{\qass{\overline{q}}{U}}{A,{\ket{\varphi}_{\overline{q}}\ket{\psi}}}}
\end{mathpar}
}

We show an example proof triple for lines 1-3, based on the proof rule for quantum unitary operation above.

{\small
\begin{mathpar}
\inferrule[]{\qhltriple{\cn{true}, \cn{H}^{\dag}(\frac{1}{\sqrt{2}}(\ket{0}_{q_1}+\ket{1}_{q_1}))\ket{0}_{q_2}\ket{0}_{q_3}}{\qass{q_1\;}{\;\cn{H}}}{\cn{true}, \frac{1}{\sqrt{2}}(\ket{0}_{q_1}+\ket{1}_{q_1})\otimes\ket{0}_{q_2}\ket{0}_{q_3}}}
{\qhltriple{\cn{true}, \ket{0}_{q_1}\ket{0}_{q_2}\ket{0}_{q_3}}{\qass{q_1\;}{\;\cn{H}}}{\cn{true}, \frac{1}{\sqrt{2}}(\ket{0}_{q_1}+\ket{1}_{q_1})\otimes\ket{0}_{q_2}\ket{0}_{q_3}}}
\end{mathpar}
}

The triple applies a Hadamard operation.
Here, we utilize the CQHL consequence rule to rewrite the bottom pre-condition to the upper one, where the inverse of $\cn{H}$ gate ($\cn{H}^{\dag}$) is $\cn{H}$ itself, and applying the \cn{H} gate to ${\frac{1}{\sqrt{2}}(\ket{0}_{q_1}+\ket{1}_{q_1})}$ yields $\ket{0}$.

\begin{figure*}[h]
{\hspace*{-1em}
  \small
  \[\begin{array}{llcl} 
     \text{Classical Arithmetic Expression} & e & ::= & v \mid x \mid e + e \mid ...
     \\[0.3em]
     \text{Boolean Expression} & B & ::= & \cn{true} \mid \cn{false} \mid e = e \mid \neg B \mid ...
  \\[0.3em]
  \text{Unitary Operations} & U & ::= & \cn{H} \mid \cn{X} \mid \slambda{\overline{q}}{\ket{f(\overline{q})}} \mid ...
    \\[0.3em]
     \text{Program Expression} & C & ::= & \sskip \mid \qass{\overline{q}}{U} \mid \sseq{C}{C}\mid \qifb{q}{C}{C} \mid \\[0.3em]
      & & \mid & \cif{b}{C}{C} \mid \cassign{x}{e} \mid f(\overline{e})
    \end{array}
  \]
}
  \caption{CQHL Syntax.}
  \label{fig:cqhl-syntax}
\end{figure*}

The CQHL syntax is in \Cref{fig:cqhl-syntax}: skip ($\sskip$), quantum unitary application $\qass{\overline{q}}{U}$, sequencing, classical conditional ($\cif{b}{C}{C}$), quantum conditional ($\qifb{q}{C}{C}$), classical assignment ($\cassign{x}{e}$), and function call ($f(\overline{e})$).
In the language, $e$ and $B$ refer to classical arithmetic and Boolean expressions, respectively.
Boolean expressions ($B$) consist of Boolean values (\cn{true} and \cn{false}), equality ($e=e$), negation ($\neg B$), and other operations.

A unitary operation is either a single-qubit gate ($\cn{H}$ or $\cn{X}$) or a quantum oracle $\slambda{\overline{q}}{\ket{f(\overline{q})}}$, which maps a quantum datum $\sum_j \alpha_j \ket{j}$ on $\overline{q}$ to $\sum_j \alpha_j \ket{f(j)}$.

We define a valid program $C$ to have no free classical variables, $FV(C) = \emptyset$, i.e., classical variables are all bound by a let-binding ($\sexp{x}{e}{C}$).
In CQHL, a quantum variable $q$ refers to a single qubit and a quantum array $\overline{q}$ to a group of qubits.
The range $[e_1,e_2)$ of the array $\overline{q}$ is written $\overline{q}[e_1,e_2)$, where $e_1$ and $e_2$ are classical expressions.

To transform CQHL to classical Hoare logic, we first translate CQHL to \qafny and then apply the translation of \Cref{sec:trans} to the result.

We first translate a CQHL predicate $\textcolor{spec}{A,\ket{\varphi}}$ to a \qafny state predicate.
Here, $A$ is a classical predicate, which translates directly.
In CQHL, the quantum predicate $\ket{\varphi}$ is written in Dirac notation as $\sum_j z_j \ket{c_j}_{\overline{q}}$, where the qubit array variable $\overline{q}$ gives the qubit locations of the state. We translate the predicate as follows.

{
\[
\textcolor{spec}{A \wedge \overline{q} : \tch{1} \mapsto \sum_j z_j \ket{c_j}}
\]
}

To translate a CQHL program to \qafny, note that classical assignments, classical conditionals, function calls, and sequencing are classical operations and map directly to their \qafny counterparts.
Only the two quantum operations, $\qass{\overline{q}}{U}$ and $\qifb{q}{C}{C}$, need care.

The translation of $\qass{\overline{q}}{U}$ depends on $U$. If $U$ is an oracle operation, we translate it to the oracle operation of \qafny; if $U$ is the Hadamard gate $\cn{H}$, we translate it to the Hadamard operation of \qafny.
Since $\cn{H}$ together with oracle operations is universal, any other gate can be decomposed into a sequence of these two operations.

A CQHL quantum conditional $\qifb{q}{C_1}{C_2}$ translates to a sequence of two \qafny quantum conditionals, as shown below.

{
\[
\qif{q[0]}{C_2} ; \qif{\neg q[0]}{C_1}
\]
}

This completes the translation of CQHL to \qafny; the translation of \Cref{sec:trans} then takes the resulting \qafny triple to classical Hoare logic.
 \end{document}